\documentclass[final,leqno]{siamltex}
\usepackage{amssymb,amsmath,enumerate}
\usepackage{graphicx}
\usepackage{multirow}
\usepackage{amsfonts}
\usepackage{enumerate}
\usepackage{color}
\usepackage[normalem]{ulem}
\usepackage{ifthen}
\usepackage{rotating}

\newtheorem{remark}{Remark}
\newtheorem{example}{Example}

\newcommand{\R}{\mathbb{R}}
\newcommand{\N}{\mathbb{N}}
\newcommand{\C}{\mathbb{C}}
\newcommand{\Q}{\mathbb{Q}}
\newcommand{\Ab}{\mathbf{A}}
\newcommand{\Pb}{\mathbf{P}}
\newcommand{\Yb}{\mathbf{Y}}
\newcommand{\Wb}{\mathbf{W}}

\newcommand{\Psib}{\mathbf{\Psi}}
\newcommand{\Thetab}{\mathbf{\Theta}}
\newcommand{\Fb}{\mathbf{F}}

\newcommand{\Zb}{\mathbf{Z}}
\newcommand{\Mb}{\mathbf{M}}

\def\Prob{\mathbb{P}}
\def\Exp{\mathbb{E}} 
\newcommand{\one}{\mathbf{1}}
\newcommand{\dd}{\mathrm{d}}
\newcommand{\setOmega}{\mathcal{H}}

\newif\ifanswervar
\answervarfalse
\ifanswervar
\newcommand{\ADDED}[1]{\textcolor{blue}{#1}}
\newcommand{\DELETED}[1]{\textcolor{blue}{\sout{#1}}}
\newcommand{\REPLACED}[2]{\textcolor{red}{\sout{#1}}\textcolor{blue}{\uline{#2}}}
\newcommand{\ADDEDONE}[1]{\textcolor{green}{#1}}
\newcommand{\DELETEDONE}[1]{\textcolor{green}{\sout{#1}}}

\newcommand{\DELETEDTWO}[1]{\textcolor{yellow}{\sout{#1}}}

\newcommand{\ADDEDTHREE}[1]{\textcolor{magenta}{#1}}
\newcommand{\DELETEDTHREE}[1]{\textcolor{magenta}{\sout{#1}}}
\newcommand{\REPLACEDTHREE}[2]{\textcolor{red}{\sout{#1}}\textcolor{magenta}{\uline{#2}}}
\else
\newcommand{\ADDED}[1]{#1}
\newcommand{\DELETED}[1]{}
\newcommand{\REPLACED}[2]{#2}
\newcommand{\ADDEDONE}[1]{#1}
\newcommand{\DELETEDONE}[1]{}

\newcommand{\DELETEDTWO}[1]{}

\newcommand{\ADDEDTHREE}[1]{#1}
\newcommand{\DELETEDTHREE}[1]{}
\newcommand{\REPLACEDTHREE}[2]{#2}
\fi
\newif\ifanswerTwovar
\answerTwovarfalse
\ifanswerTwovar
\newcommand{\ADDEDTwo}[1]{\textcolor{blue}{#1}}
\newcommand{\DELETEDTwo}[1]{\textcolor{blue}{\sout{#1}}}
\newcommand{\REPLACEDTwo}[2]{\textcolor{red}{\sout{#1}}\textcolor{blue}{\uline{#2}}}
\else
\newcommand{\ADDEDTwo}[1]{#1}
\newcommand{\DELETEDTwo}[1]{}
\newcommand{\REPLACEDTwo}[2]{#2}
\fi


\title{Variable density sampling with continuous trajectories. \ADDEDTwo{Application to MRI.}}

\author{Nicolas Chauffert\thanks{Inria Saclay, Parietal team. CEA/NeusoSpin, 91191 Gif-sur-Yvette, France ({\tt nicolas.chauffert@cea.fr}).}
        \and Philippe Ciuciu\thanks{Inria Saclay, Parietal team. CEA/NeusoSpin, 91191 Gif-sur-Yvette, France ({\tt philippe.ciuciu@cea.fr}).}
	\and Jonas Kahn\thanks{Laboratoire Painlev\'e, UMR8524 Universit\'e de Lille 1, CNRS. Cit\'e Scientifique B\^at. M2, 59655 Villeneuve d'Asq Cedex, France ({\tt jonas.kahn@math.univ-lille1.fr}).}
	\and Pierre Weiss\thanks{ITAV, USR 3505. PRIMO Team, Universit\'e de Toulouse, Toulouse, France.({\tt pierre.armand.weiss@gmail.com}).}
}

\begin{document}
\graphicspath{{figures/}}

\maketitle

\begin{abstract}
Reducing acquisition time is a crucial challenge for many imaging techniques. Compressed Sensing (CS) theory offers an appealing framework to address this issue since it provides theoretical guarantees on the reconstruction of sparse signals by projection on a low dimensional linear subspace. In this paper, we focus on a setting where the imaging device allows to sense a fixed set of measurements. We first discuss the choice of an optimal sampling subspace allowing perfect reconstruction of sparse signals. Its design relies on the random drawing of independent measurements. We discuss how to select the drawing distribution and show that a mixed strategy involving partial deterministic sampling and independent drawings can help breaking the so-called “coherence barrier”. Unfortunately, independent random sampling is irrelevant for many acquisition devices owing to acquisition constraints. To overcome this limitation, the notion of Variable Density Samplers (VDS) is introduced and defined as a stochastic process with a prescribed limit empirical measure. It encompasses samplers based on independent measurements or continuous curves. The latter are crucial to extend CS results to actual applications. \REPLACED{Our main contribution lies in two original continuous VDS. The first one relies on random walks over the acquisition space whereas the second one is heuristically driven and rests on the approximate solution of a Traveling Salesman Problem. Theoretical analysis and retrospective CS simulations in magnetic resonance imaging highlight that the TSP-based solution provides improved reconstructed images in terms of signal-to-noise ratio compared to standard sampling schemes (spiral, radial, 3D iid...).}{We propose two original approaches to design continuous VDS, one based on random walks over the acquisition space, and one based on Traveling Salesman Problem. Following theoretical considerations and retrospective CS simulations in magnetic resonance imaging, we intend to highlight the key properties of a VDS to ensure accurate sparse reconstructions, namely its limit empirical measure and its mixing time\footnote{\ADDEDTwo{Part of this work is based on the conference proceedings:~\cite{Chauffert13,Chauffert13b,Chauffert13c}.} \hrule}}.  
\end{abstract}



\pagestyle{myheadings}
\thispagestyle{plain}
\markboth{N. CHAUFFERT, P. CIUCIU, J. KAHN AND P. WEISS}{VARIABLE DENSITY SAMPLING WITH CONTINUOUS SAMPLING TRAJECTORIES}

\section{Introduction}

Variable density sampling \DELETEDTHREE{(VDS) }is a technique that is extensively used in various sensing devices such as magnetic resonance imaging~(MRI), in order to shorten scanning time. 
It consists in measuring only a small number of random projections of a signal/image on elements of a basis drawn according to a given density. 
For instance, in MRI where measurements consist of Fourier~(or more generally $k$-space) coefficients, it is common to sample the Fourier plane center more densely than the high frequencies. 
The image is then reconstructed from this incomplete information by dedicated signal processing methods. 
To the best of our knowledge, \emph{variable density sampling} has been proposed first in the MRI context by~\cite{spielman1995magnetic} where spiral trajectories were pushed forward. Hereafter, it has been used in this application~(see e.g.~\cite{tsai2000reduced,kim2003simple,park2005artifact} to quote a few), but also in other applications such as holography \cite{rivenson2010compressive,marim2010compressed}. This technique can hardly be avoided in specific imaging techniques such as radio interferometry or tomographic modalities~(e.g., X-ray) where sensing is made along fixed sets of measurements~\cite{Wiaux09,Sidky06}. 

In the early days of its development, \REPLACEDTHREE{VDS}{variable density sampling} was merely an efficient heuristic to shorten acquisition time. It has recently found a partial justification in the Compressed Sensing~(CS) literature. Even though this theory is not yet mature enough to fully explain the practical success of \REPLACEDTHREE{VDS}{variable density sampling}, CS provides good hints on how to choose the measurements~(i.e., the density), how the signal/image should be reconstructed and why it works.
Let us now recall a typical result emanating from the CS literature for orthogonal systems. 
A vector $x \in \C^n$ is said $s$-sparse if it contains at most $s$ non-zero entries. 
Denote by $a_i,\ i\in \{1,\hdots, n\}$ the sensing vectors and by $y_i=\langle a_i,x \rangle$ the possible measurements. 
Typical CS results state that if the signal~(or image) $x$ is $s$-sparse and if $\Ab=\begin{pmatrix} a^*_1 \\ \vdots \\ a^*_n \end{pmatrix}$ satisfies an incoherence property (defined in the sequel)
, then $m=O(s\log(n)^\alpha)$ measurements chosen randomly among the elements of $y=\Ab x$ are enough to ensure perfect reconstruction of $x$. The constant $\alpha>0$ depends on additional properties on $x$ and $\Ab$.
The set of actual measurements is denoted $\Omega\subseteq \{1,\hdots, n\}$ and $\Ab_\Omega$ is the matrix formed by selecting a subset of rows of $\Ab$ in $\Omega$. 
The reconstruction of $x$ knowing $y_\Omega=\Ab_\Omega x$ is guaranteed if it results from solving the following $\ell_1$ minimization problem:
\begin{align}
\label{eq:minl1}
\min_{z \in \C^n } \| z \|_1 &  \qquad \text{subject to} \qquad \Ab_\Omega z = y_\Omega.
\end{align}

Until recent works~\cite{Rauhut10,Juditsky11b,Candes11}, no general theory for selecting the rows was available. 
In the latter, the authors have proposed to construct $\Ab_\Omega$ by drawing $m$ rows of $\Ab$ at random according to a discrete probability distribution or density $p=(p_1, \hdots , p_n)$. The choice of an optimal distribution $p$ is an active field of research~(see e.g. \cite{Chauffert13,Krahmer12,Adcock13}) that remains open in many regards. 

Drawing independent rows of $\Ab$ is interesting from a theoretical perspective, however it has little practical relevance since standard acquisition devices come with acquisition constraints. 
For instance, in MRI, 
the coefficients are acquired along \textit{\ADDEDTHREE{piece-wise} continuous} curves on the $k$-space. The first paper performing \REPLACEDTHREE{VDS}{variable density sampling} in MRI~\cite{spielman1995magnetic} has fulfilled this constraint by considering spiral sampling trajectories. 
The standard reference about CS-MRI~\cite{Lustig07} has proposed to sample the MRI signal along parallel lines in the 3D $k$-space. 
Though spirals and lines can be implemented easily on a scanner, it is likely that more general trajectories could provide better reconstruction results, or save more scanning time.

The main objective of this paper is to propose new strategies to sample a signal along more general continuous curves. 
Although continuity is often not sufficient for practical implementation on actual scanner, we believe that it is a first important step towards more physically plausible compressed sampling paradigms.
As far as we know, this research avenue is relatively new. The problem was first discussed in \cite{Willett} and some heuristics were proposed. The recent contributions~\cite{polak2012performance,Bigot13} have provided theoretical guarantees when sampling is performed along fixed sets of measurements~(e.g. straight lines in the Fourier plane), but have not addressed generic continuous sampling curves yet.

The contributions of this paper are threefold. 
First, we bring a well mathematically grounded definition of \emph{variable density samplers} and provide various examples. Second, we discuss how the sampling density should be chosen in practice. 
This discussion mostly relies on variations around the theorems provided in~\cite{Rauhut10,Candes11}. In particular, we justify the deterministic sampling of a set of highly coherent vectors to overcome the so-called ``coherence barrier''. In the MRI case, this amounts to deterministically sampling the $k$-space center.
Our third and maybe most impacting contribution is to provide practical examples of variable density samplers along continuous curves and to derive some of their theoretical properties. These samplers are defined as parametrized random curves that asymptotically fit a target distribution (e.g. the one shown in Fig.~\ref{fig:Intro}~(a)). More specifically, 
we first propose a \emph{local} sampler based on random walks over the acquisition space~(see Fig.~\ref{fig:Intro}~(b)). Second, we introduce a \emph{global} sampler based on the solution of a Travelling Salesman Problem amongst randomly drawn ``cities''~(see Fig.~\ref{fig:Intro}~(c)). In both cases, we investigate the resulting density.
To finish with, we illustrate the proposed sampling schemes on 2D and 3D MRI simulations. The reconstruction results provided by the proposed techniques show that the PSNR can be substantially improved compared to existing strategies proposed e.g. in~\cite{Lustig07}.
\ADDED{Our theoretical results and numerical experiments on retrospective CS show that two key features of variable density samplers are the \textbf{limit of their empirical measure} and their \textbf{mixing properties}.}

\begin{figure}[!h]
\begin{center}
\begin{tabular}{ccc}
(a)&\hspace{-.03\linewidth}(b)&(c) \\
\includegraphics[width=.35\linewidth]{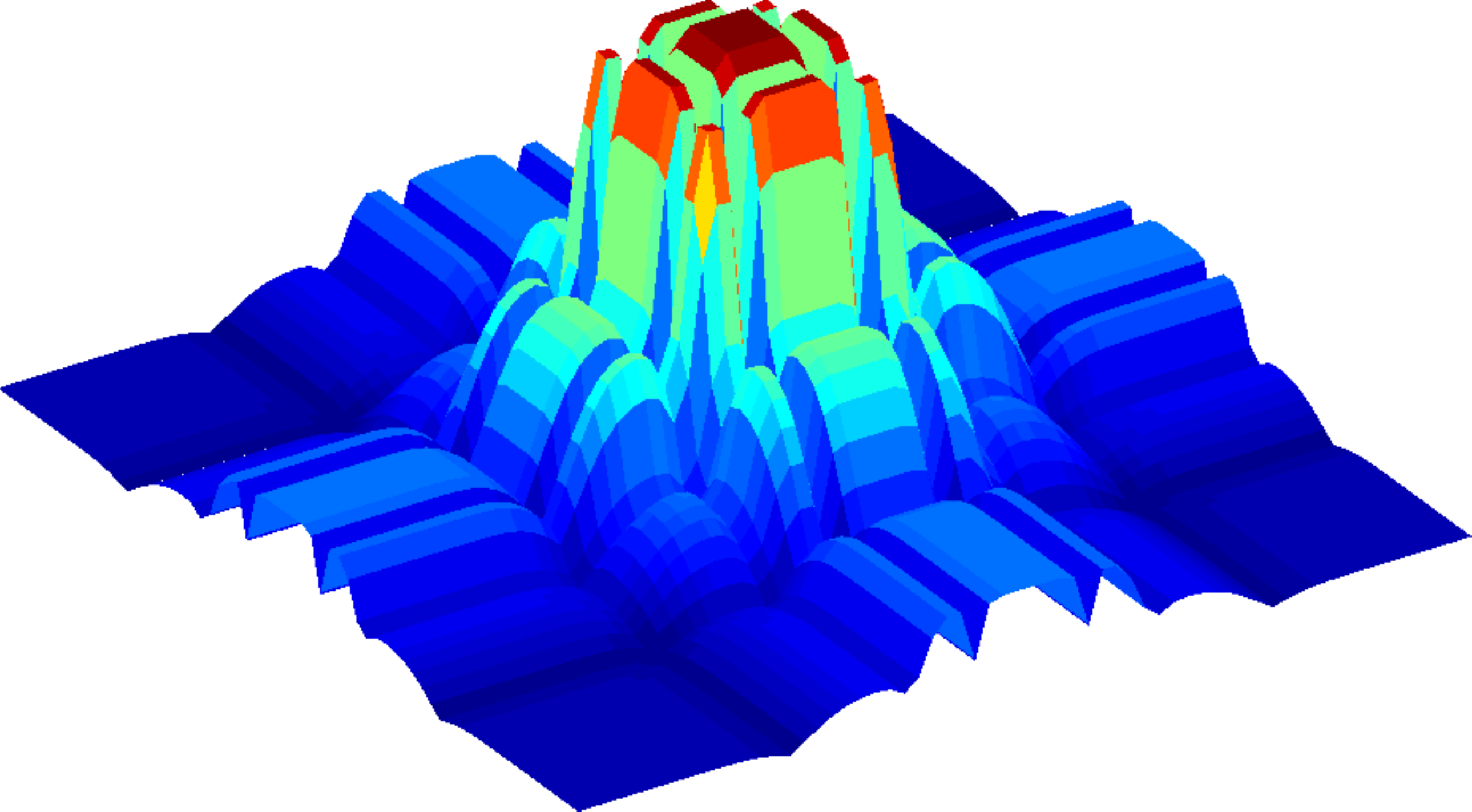}&\hspace{-.03\linewidth}
\includegraphics[width=.28\linewidth]{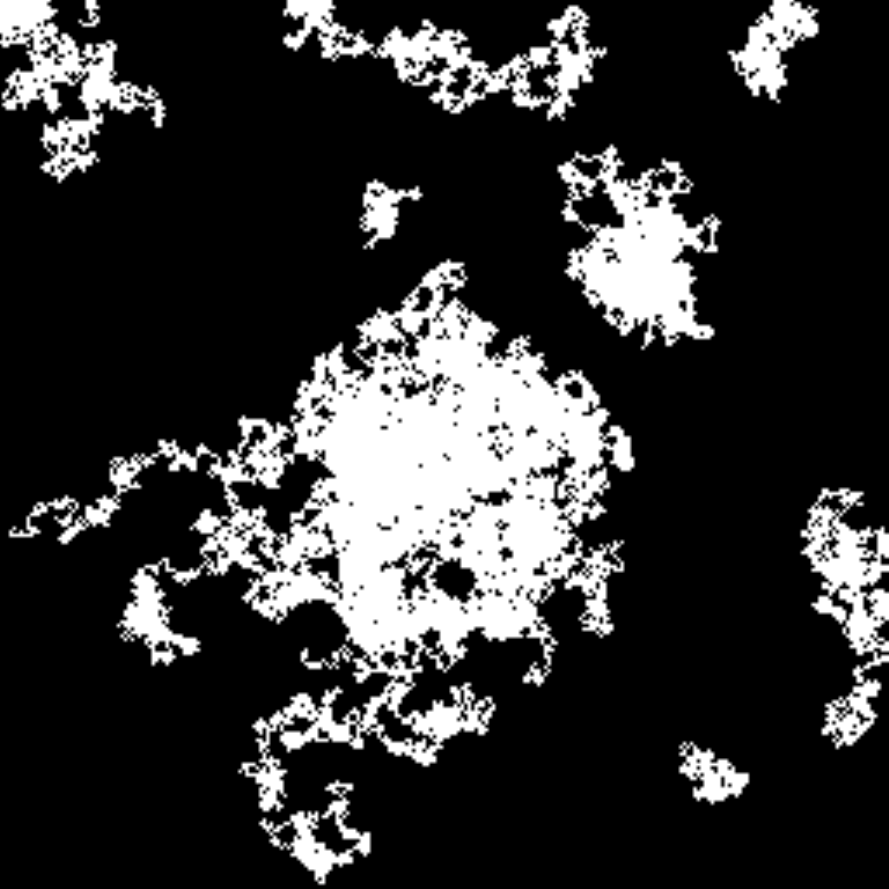}&
\includegraphics[width=.28\linewidth]{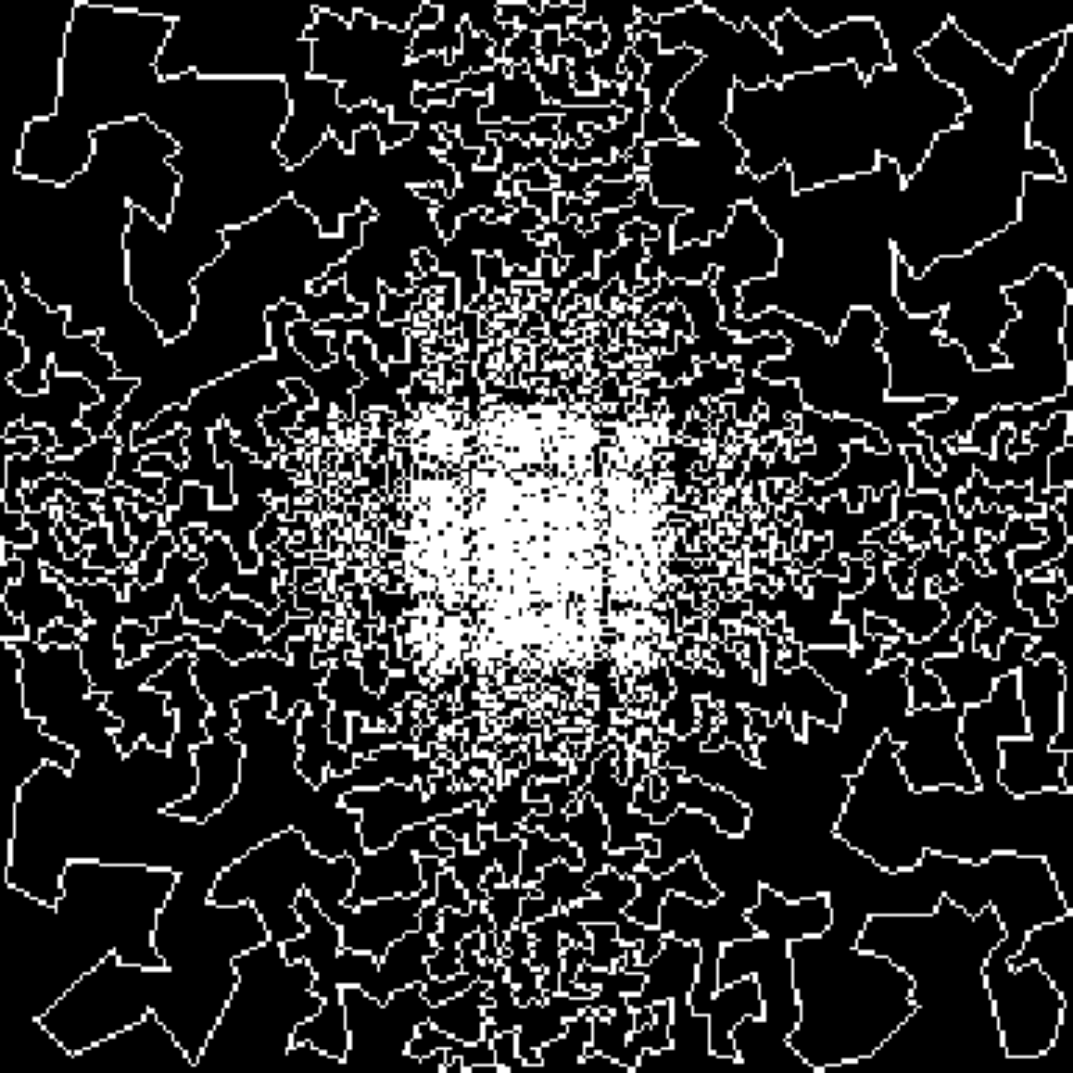}\\
\end{tabular}
\caption{\label{fig:Intro} (a): Target distribution $\pi$. Continuous random trajectories reaching distribution $\pi$ based on Markov chains (b) and on a TSP solution (c).}
\end{center}
\end{figure}

The rest of this paper is organized as follows. First, we introduce a precise definition of \REPLACEDTHREE{VDS}{a variable density sampler (VDS)} and recall CS results in the special case of independent drawings. Then, we give a closed form expression for the optimal distribution depending on the sensing matrix $\Ab$, and justify that a partial deterministic sampling may provide better reconstruction guarantees. Hereafter, in Sections~\ref{sec:cont} and~\ref{sec:TSP}, we introduce two strategies to design continuous trajectories over the acquisition space. We show that the corresponding sampling distributions converge to a target distribution when the curve length tends to infinity. 
Finally, we demonstrate on simulation results that our TSP-based approach is promising in the MRI context~(Section~\ref{sec:results}) since it outperforms its competing alternatives either in terms of PSNR at fixed sampling rate, or in terms of acceleration factor at fixed PSNR.


\section*{Notation}
\ADDED{
The main definitions used throughout the paper are defined in Tab.~\ref{Tab:not}.
}
\begin{table}[h]
 \caption{General notation used in the paper. \label{Tab:not}}
\centering
\begin{tabular}{ccll}
&Notation&Definition& Domain \\
\cline{2-4}
 \multirow{11}{*}{\rotatebox{90}{\scriptsize Compressed Sensing}}
 &\scriptsize $n$ & \scriptsize Acquisition and signal space dimensions & \scriptsize$\N$\\
 &\scriptsize $m$ & \scriptsize Number of measurements & \scriptsize$\N$\\
 &\scriptsize $R=n/m$ & \scriptsize Sampling ratio &\scriptsize $\Q$\\
 &\scriptsize $\Ab$ &\scriptsize Full orthogonal acquisition matrix & \scriptsize$\C^{n \times n}$\\
 &\scriptsize $\Omega$ &\scriptsize Set of measurements & \scriptsize$\{1,\hdots,n\}^{m}$\\
 &\scriptsize $\Ab_\Omega$ &\scriptsize Matrix formed with the rows of $\Ab$ corresponding to indexes belonging to $\Omega$ & \scriptsize$\C^{m \times n}$\\
 &\scriptsize $x$ &\scriptsize Sparse signal & \scriptsize$ \C^{n}$\\
 &\scriptsize $s$ & \scriptsize Number of non zero coefficients of $x$ & \scriptsize$\N$\\
 &\scriptsize $\Delta_n$ & \scriptsize $\Bigl\{p=\begin{pmatrix} p_1 \\ \vdots \\ p_n \end{pmatrix}, 0\leqslant p_i \leqslant 1, \sum_{i=1}^n p_i=1\Bigr\}$ &\scriptsize $\R^n$\\
 &\scriptsize $\|~\|_1$ & \scriptsize $\ell_1$ norm defined for $z \in \C^n$ by $\|z\|_1=\sum_{i=1}^n|z_i|$ &  \\
 &\scriptsize $\|~\|_\infty$ & \scriptsize $\ell_\infty$ norm defined for $z \in \C^n$ by $\|z\|_\infty=\max_{1\leqslant i \leqslant n} |z_i|$ &  \\
 \cline{2-4}
 \multirow{4}{*}{\rotatebox{90}{\scriptsize MRI application}}
 &\scriptsize $k=\begin{pmatrix} k_x \\ k_y \end{pmatrix}$ or $\begin{pmatrix} k_x \\ k_y \\k_z \end{pmatrix}$ &\scriptsize Fourier frequencies & \scriptsize$\R^{2}$ or $\R^3$\\
 &\scriptsize $\Fb^*_n$ &\scriptsize $d$-dimensional discrete Fourier transform on an $n$ pixels image & \scriptsize$\C^{n \times n}$\\
 &\scriptsize $\Psib_n$ &\scriptsize $d$-dimensional inverse discrete Wavelet transform on an image of $n$ pixels & \scriptsize$\C^{n \times n}$\\
 & & \scriptsize $\Fb^*_n$ and $\Psib_n$ are denoted $\Fb^*$ and $\Psib$ if no ambiguity &  \\
 \cline{2-4}
 \multirow{13}{*}{\rotatebox{90}{\scriptsize VDS}}
 &\scriptsize $\Xi$ &\scriptsize A measurable space which is typically $\{1, \hdots, n\}$ or $[0,1]^d$ & \\
 &\scriptsize $\setOmega$ &\scriptsize The unit cube $[0,1]^d$ & \\
 &\scriptsize $p$ &\scriptsize A probability measure defined on $\Xi$ & \\
 &\scriptsize $p(f)$ &\scriptsize$=\displaystyle\int_{x\in \Xi} f(x)\ \dd p(x)$, for $f$ continuous and bounded & \scriptsize$\R$\\
 &\scriptsize $\lambda_{[0,1]}$ & \scriptsize The Lebesgue measure on the interval $[0,1]$ & \\
 \cline{3-4}
 &\scriptsize $X=(X_n)_{n\in \mathbb{N}^*}$ &\scriptsize A time-homogeneous Markov chain on the state space $\{1, \hdots,n \}$  & \scriptsize$\{1, \hdots n\}^{\N^*}$\\
 &\scriptsize $\Pb$ & \scriptsize $:=(P_{ij})_{1\leqslant i,j \leqslant n}$ the transition matrix: $P_{ij}:=\Prob(X_k=j|X_{k-1}=i), \forall k>1$  & \scriptsize$\R^{n \times n}$\\ 
 &\scriptsize $\lambda_i(\Pb)$ & \scriptsize The ordered eigenvalues of $\Pb$: $1=\lambda_1(\Pb)\geqslant \hdots \geqslant \lambda_n(\Pb)\geqslant -1$  & \scriptsize$[-1,1]$\\
 &\scriptsize $\epsilon(\Pb)$ & \scriptsize $=1-\lambda_2(\Pb)$, the spectral gap of $\Pb$ & \scriptsize$[-1,1]$\\
 \cline{3-4}
 &\scriptsize $F$ &\scriptsize A set of points $\subset \setOmega$ & \scriptsize $\setOmega^N$\\
 &\scriptsize $C(F)$ &\scriptsize The shortest Hamiltonian path (TSP) amongst points of set $F$ & \scriptsize $\subset \setOmega$\\ 
 &\scriptsize $T(F,\setOmega)$ &\scriptsize The length of $C(F)$   & \scriptsize $\R_+$\\
 &\scriptsize $T(F,R)$ &\scriptsize For any set $R\subseteq \setOmega $, $T(F, R) := T(F \cap R,\setOmega)$  & \scriptsize $\R_+$\\
\hline
\end{tabular}
\end{table}


\section{Variable density sampling and its theoretical foundations}
\label{sec:VDS}
To the best of our knowledge, there is currently no rigorous definition of variable density sampling. 
Hence, to fill this gap, we provide a precise definition below.

\begin{definition}
\label{def:VDS}
Let $p$ be a probability measure defined on a \ADDED{measurable }space $\Xi$. A stochastic process $X=\{X_i\}_{i\in \mathbb{N}}$ or $X=\{X_t\}_{t\in \mathbb{R_+}}$ \ADDED{on state space $\Xi$} is called a $p$-\emph{variable density sampler} if its empirical measure (or occupation measure) weakly converges to $p$ almost surely, that is:
\begin{align*}
\frac{1}{N} \sum_{i=1}^N f(X_i) & \to p(f)   \qquad a.s.\\
& \textrm{or} \\
\frac{1}{T} \int_{t=0}^T f(X_t) \dd t & \to p(f)  \qquad a.s.
\end{align*}
for all continuous bounded $f$.
\end{definition}

\begin{example}
In the case where $X=(X_i)_{i\in \mathbb{N}}$ is a discrete time stochastic process with discrete state space $\Xi=\{1,\hdots, n\}$, definition~\ref{def:VDS} can be slightly simplified. 
Let us set $\displaystyle Z_j^N=\frac{1}{N}\sum_{i=1}^N\one_{X_i=j}$. The random variable $Z_j^N$ represents the proportion of points that fall on position $j$.
Let $p$ denote a discrete probability distribution function. Using these notations, $X$ is a $p$-variable density sampler if:
$$
\lim_{N\to +\infty} Z_j^N = p_j \qquad a.s.
$$

In particular, if $(X_i)_{i\in \mathbb{N}}$ are i.i.d. samples drawn from $p$, then $X$ is a $p$-variable density sampler. This simple example is the most commonly encountered in the compressed sensing literature and we will review its properties in paragraph \ref{ssec:theoretical_motiv}.
\end{example}

\begin{example}
More generally, drawing independent random variables according to distribution $p$ is a VDS if the space $\Xi$ is second countable, owing to the strong law of large numbers.
\end{example}

\begin{example}
An irreducible aperiodic Markov chain on a finite sample space is
a VDS for its stationary distribution (or invariant measure); see Section~\ref{sec:markov}.
\end{example}

\begin{example}
In the deterministic case, for a dynamical system, definition~\ref{def:VDS} closely
corresponds to the ergodic hypothesis, that is time averages are equal to expectations over space. 
We discuss an example that makes use of the TSP solution in section \ref{sec:TSP}. 
\end{example}

The following proposition directly relates the VDS concept to the time spent by the process in a part of the space, as an immediate consequence of the porte-manteau lemma~(see e.g.~\cite{Billingsley09}).

\begin{proposition}
Let $p$ denote a Borel measure defined on a set $\Xi$. Let $B\subseteq \Xi$ be a measurable set. Let $X: \R_+ \to \Xi$ (resp. $X: \N \to \Xi$) be a \DELETED{deterministic or }stochastic process. Let $\mu$ denote the Lebesgue measure on $\mathbb{R}$. Define $\mu_X^t(B)= \frac{1}{t}\mu(\{s\in [0,t],\ X(s)\in B\})$ (resp. $\mu_X^n(B)= \frac{1}{n}\sum_{i=1}^n\one_{X(i) \in B}$). Then, the following two propositions are equivalent:
\begin{align*}
(i) & \qquad X \textrm{ is a $p$-VDS}   \\
(ii) & \qquad  \textrm{Almost surely, }\forall B\subseteq \Xi \textrm{ a Borel set with $p(\partial B)=0$, } \\
 & \qquad\lim_{t\to +\infty} \mu_X^t(B) = p(B)  \qquad \textrm{a.s.}\\
\qquad \textrm{(resp.)} & \qquad \lim_{n\to +\infty} \mu_X^n(B) = p(B)  \qquad \textrm{a.s.}
\end{align*}
\end{proposition}

\ADDED{\begin{remark}
Definition~\ref{def:VDS} is a generic definition that encompasses both discrete and continuous time and discrete and continuous state space since $\Xi$ can be any measurable space.  In particular, the recent CS framework on orthogonal systems~\cite{Rauhut10,Candes11} falls within this definition.
\end{remark}}

\ADDED{
Definition~\ref{def:VDS} does not encompass some useful sampling strategies. We propose a definition of a generalized VDS, which encompasses stochastic processes indexed over a bounded time set.
}
\begin{definition}
\label{def:gVDS}
\ADDED{
A sequence $\{\{X_t^{(n)}\}_{0\leqslant t \leqslant T_n}\}_{n\in\N}$ is a \emph{generalized $p$-VDS} if the sequence of occupation measures converges to $p$ almost surely, that is:
\begin{align*}
\frac{1}{T_n} \int_{t=0}^{T_n} f(X_t^{(n)}) \dd t & \to p(f)   \qquad a.s.
\end{align*}
}
\end{definition}
\ADDED{\begin{remark}
Let $(X_t)_{t\in\R}$ be a VDS, and $(T_n)_{n\in \N}$ be any positive sequence such that $T_n\to \infty$. Then the sequence defined by $X_t^{(n)}=X_t$ for $0\leqslant t \leqslant T_n$ is a generalized VDS.
\end{remark}}
\ADDED{\begin{example}
\label{ex:spiral}
Let $\Xi=\R^2$, and consider $r:[0,1]\mapsto \R^+$ a strictly increasing smooth function. We denote by $r^{-1} : [r(0),r(1)]\to \R$ its inverse function and by $\dot{r^{-1}}$ the derivative of $r^{-1}$. Consider a sequence of spiral trajectories $s_N:[0,N]\to \R^2$ defined by $\displaystyle s_N(t)=r\Bigl(\frac{t}{N}\Bigr)\begin{pmatrix}\cos(2\pi t)\\ \sin(2\pi t) \end{pmatrix}$. Then $s_N$ is a generalized VDS for the distribution $p$ defined by:
\begin{align*}
p(x,y)=\left\{\begin{array}{cl}
\frac{\dot{r^{-1}}\left( \sqrt{x^2+y^2}\right)}{2 \pi \int_{\rho=r(0)}^{r(1)} \dot{r^{-1}}(\rho)\rho \dd \rho} & \mbox{if} \quad r(0)\leqslant \sqrt{x^2+y^2}\leqslant r(1) \\
0 & \mbox{otherwise} \\
\end{array}
\right.
\end{align*}
A simple justification is that the time spent by the spiral in the infinitesimal ring $\{(x,y)\in \R^2, \rho \leqslant \sqrt{x^2+y^2} \leqslant \rho+\dd \rho\} $  is $\int_{r^{-1}(\rho)}^{r^{-1}(\rho+\dd\rho)}\dd t\propto \dot{r^{-1}}(\rho)$.
\end{example}
}

\subsection{Theoretical foundations - Independent VDS}
\label{ssec:theoretical_motiv}

CS theories provide strong theoretical foundations of VDS based on independent drawings. 
In this paragraph, we recall a typical result that motivates independent drawing in the $\ell_1$ recovery context~\cite{Rauhut10,foucart2013mathematical,Candes11,Krahmer12,Chauffert13,Bigot13,Adcock13}. Using the notation defined in the introduction, let us give a slightly modified version of~\cite[Theorem 4.2]{Rauhut10}.
\begin{theorem}
\label{thm:rauhut}
Let $p=(p_1, \hdots , p_n)$ denote a probability distribution on $\{ 1 , \hdots , n \}$ and $\Omega \subset \{ 1 , \hdots , n \}$ denote a random set obtained by $m$ independent drawings with respect to distribution $p$. 
Let $S \in \{1, \hdots , n\}$ be an arbitrary set of cardinality $s$. 
Let $x$ be an $s$-sparse vector with support $S$ such that the signs of its non-zero entries is a Rademacher or Steinhaus sequence\footnote{A Rademacher (resp. Steinhaus) random variable is uniformly distributed on $\{-1;1\}$ (resp. on the torus $\{z\in \C; |z|=1\}$).}. Define:
\begin{align}
\label{eq:K}
K(\Ab,p):=\max_{k\in \{1\hdots n\}} \frac{\|a_k\|_\infty^2}{p_k} 
\end{align}
Assume that:
\begin{align}
\label{eq:meas_rauhut}
m\geqslant C K(\Ab,p) s \ln^2\left(\frac{6n}{\eta}\right)
\end{align}
where $C\approx 26.25$ is a constant. 
Then, with probability $1-\eta$, vector $x$ is the unique solution of the $\ell_1$ minimization problem~\eqref{eq:minl1}.
\end{theorem}

\begin{remark}
Cand\`es and Plan have stated stronger results in the case of real matrices in~\cite{Candes11}. 
Namely, the number of necessary measurements was decreased to $O(s\log(n))$, with lower constants and without any assumption on the vector signs. 
Their results have been derived using the so-called ``golfing scheme'' proposed in~\cite{Gross11}. 
It is likely that these results could be extended to the complex case, however it would not change the optimal distribution which is the main point of this paper. 
We thus decided to stick to Theorem~\ref{thm:rauhut}.
\end{remark}

The choice of an accurate distribution $p$ is crucial since it directly impacts the number of measurements required. In the MRI community, a lot of heuristics have been proposed so far to identify the \textit{best} sampling distribution. 
In the seminal paper on CS-MRI~\cite{Lustig07}, Lustig \textit{et al} have proposed to sample the $k$-space using a density that polynomially decays towards high frequencies. More recently, Knoll \textit{et al} have generalized this approach by inferring the best exponent from MRI image databases~\cite{Knoll11}.
It is actually easy to derive the theoretically \textit{optimal distribution}, i.e. the one that minimizes the right hand-side in~\eqref{eq:meas_rauhut} as shown in Proposition~\ref{prop:theoopt}\ADDED{, introduced in~\cite{Chauffert13}}.
\begin{proposition}
\label{prop:theoopt}
Denote by $\displaystyle K^*(\Ab):=\min_{p \in \Delta_n} K(\Ab,p)$\DELETED{, where $\Delta_n$ is the simplex in $\R^n$}. 
\begin{enumerate}[(i)]
\item the optimal distribution $\pi \in \Delta_n$ that minimizes $K(\Ab,p)$ is:
\begin{align}
\label{eq:opt_distrib}
\pi_i=\frac{\|a_i\|_\infty^2}{\sum_{i=1}^n \|a_i\|_\infty^2}
\end{align}
\item $K^*(\Ab)=K(\Ab,\pi)=\sum_{i=1}^n \|a_i\|_\infty^2$.\\
\end{enumerate}
\end{proposition}
\begin{proof}
\begin{enumerate}[(i)]
\item Taking $p=\pi$, we get $K(\Ab,\pi)=\sum_{i=1}^n \|a_i\|_\infty^2$. Now assume that $q\neq \pi$, since $\sum_{k=1}^n q_k = \sum_{k=1}^n \pi_k=1$, $\exists j \in \{1,\hdots,n\}$ such that $q_j<\pi_j$. Then $K(\Ab,q)\geqslant \|a_j\|_\infty^2/q_j > \|a_j\|_\infty^2/\pi_j = \sum_{i=1}^n \|a_i\|_\infty^2 = K(\Ab,\pi)$. So, $\pi$ is the distribution that minimizes $K(\Ab,p)$.
\item This equality is a consequence of $\pi$'s definition.
\end{enumerate}
\end{proof}

The theoretical optimal distribution only depends on the acquisition matrix, i.e. on the acquisition and sparsifying bases. For instance, if we measure some Fourier frequencies of a sparse signal in the time domain~(a sum of diracs), we should sample the frequencies according to a uniform distribution, since $\|a_i\|_\infty = 1/\sqrt{n}$ for all $1 \leqslant i \leqslant n$. In this case, $K^*(\Fb)=1$ and the number of measurements $m$ is proportional to $s$, which is in accordance with the seminal paper by Cand\`es~\textit{et al.}~\cite{Candes06}.

\subsubsection*{Independent drawings in MRI} 

In the MRI case, the images are usually assumed sparse~(or at least compressible) in a wavelet basis, while the acquisition is performed in the Fourier space. In this setting, the acquisition matrix can be written as $\Ab = \Fb^* \Psib$\DELETED{ where $\Fb^*$ and $\Psib$ denote the Fourier and inverse wavelet transform matrices, respectively}. 
In that case, the \textit{optimal distribution} only depends on the choice of the wavelet basis. The optimal distributions in 2D and 3D are depicted in Fig.~\ref{fig:pi_opt}(a)-(b), respectively if we assume that the MR images are sparse in the Symmlet basis with 3 decomposition levels in the wavelet transform.

\begin{figure}[!h]
\begin{center}
\begin{tabular}{cc}
(a)&(b)\\
\includegraphics[width=.4 \linewidth]{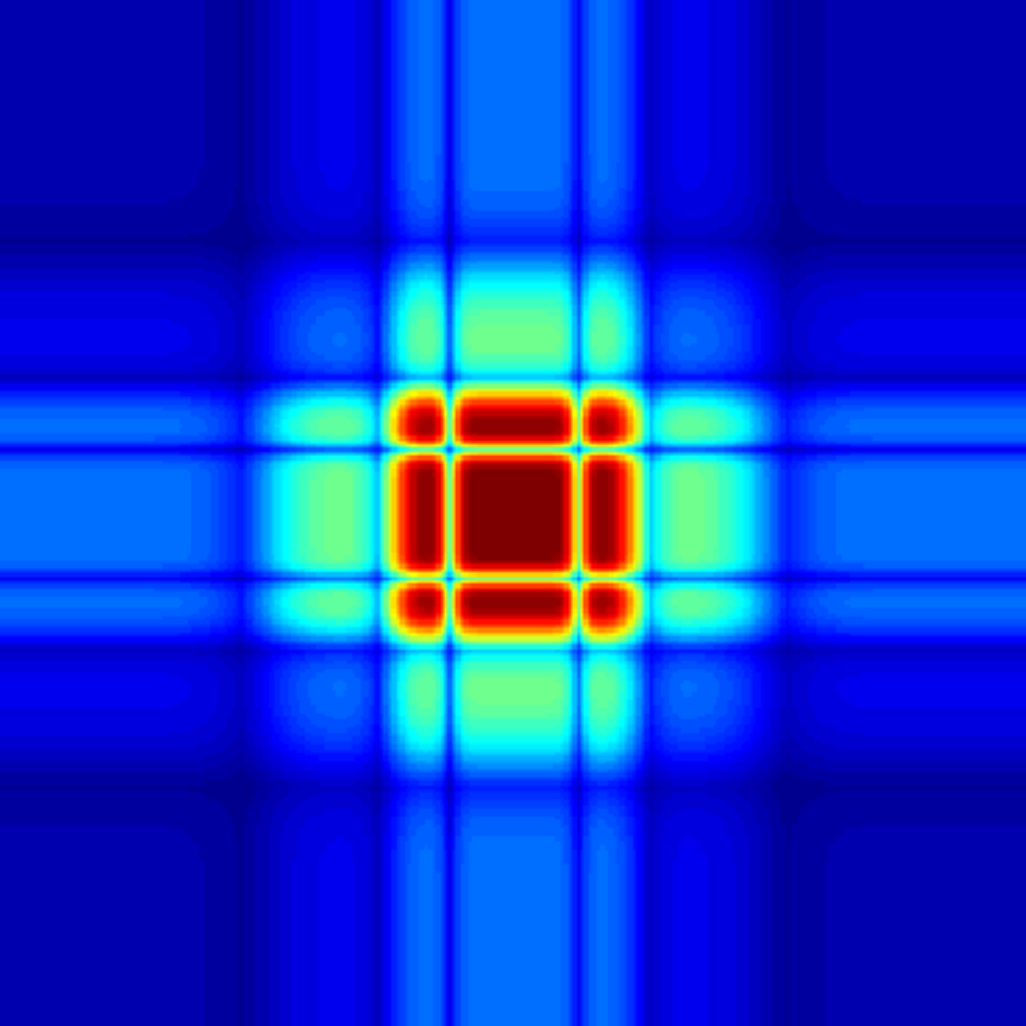}&
\includegraphics[width=.4 \linewidth]{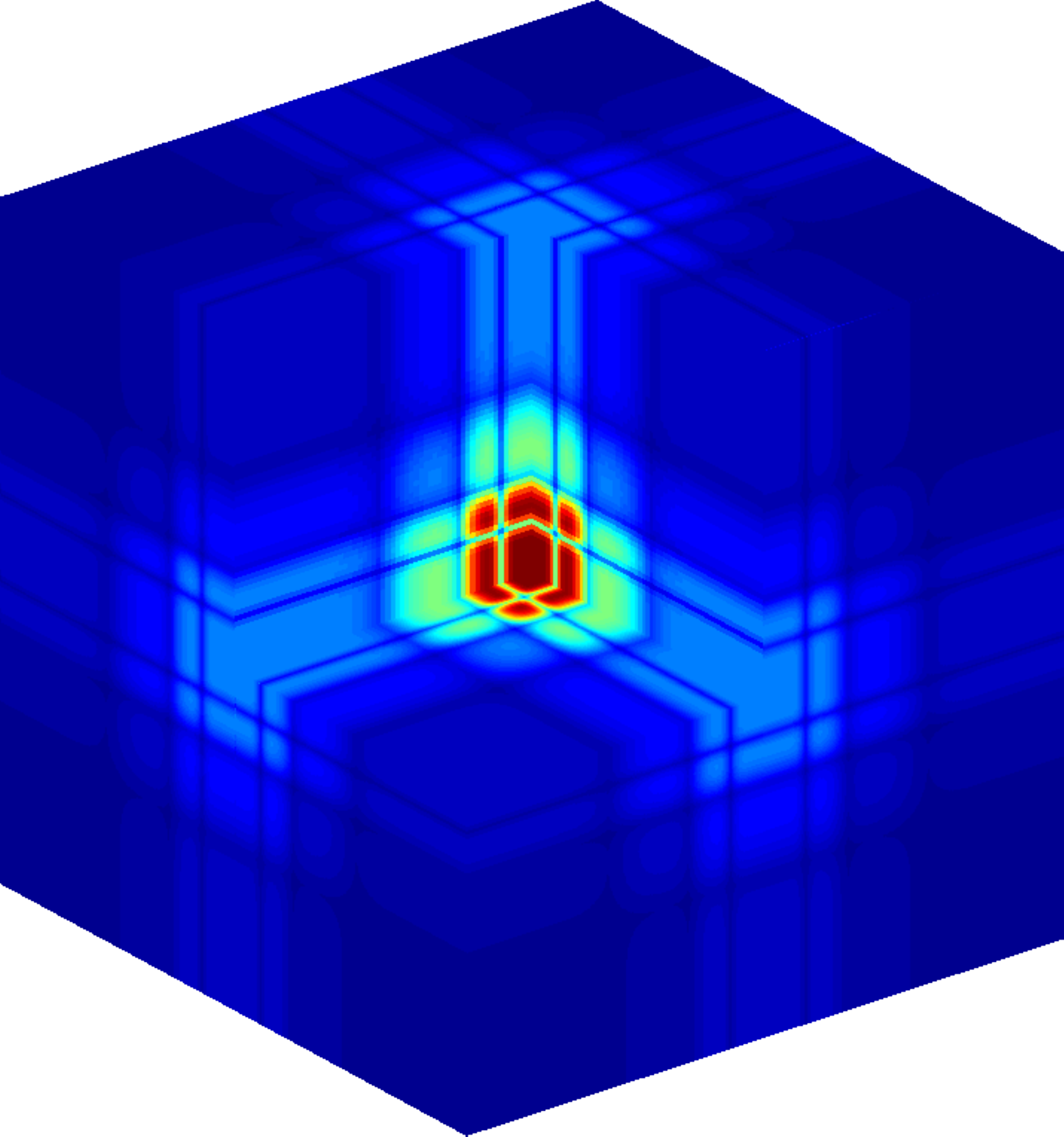}\\
\end{tabular}
\caption{\label{fig:pi_opt} Optimal distribution $\pi$ for a Symmlet-10 tranform in 2D (a) and a maximal projection of the optimal distribution in 3D (b).}
\end{center}
\end{figure}

Let us mention that similar distributions have been proposed in the literature. 
First, an alternative to independent drawing was proposed by Puy \textit{et al.}~\cite{Puy11}. 
Their approach consists in selecting or not a frequency by drawing a Bernoulli random variable. Its parameter is determined by minimizing a quantity that slightly differs from $K(\Ab,p)$. Second, Krahmer and Ward~\cite{Krahmer12} tried to unify theoretical results and empirical observations in the MRI framework. 
For Haar wavelets, they have shown that 
a polynomial distribution on the 2D \REPLACED{$(k_x,k_y)$}{$k$}-space which varies as $1/(k_x^2+k_y^2)$ is close to the optimal solution since it verifies $K(\Ab,p)=O(\log (n))$. Our numerical experiments have confirmed that a decay as a power of 2 is near optimal in 2D. 

In the next section, we improve the existing theories by showing that a deterministic sampling of highly coherent vectors~(i.e. those satisfying $\|a_i\|_\infty^2\gg \frac{1}{n}$) may decrease the total number of required measurements. In MRI, this amounts to fully sampling the low frequencies, which exactly matches what has been done heuristically hitherto.

\subsection{Mixing deterministic and independent samplings}

In a recent work~\cite{Chauffert13}, we observed and partially justified the fact that a deterministic sampling of the low frequencies in MRI could drastically improve reconstruction quality. The following theorem proven in Appendix~1 provides a theoretical justification to this approach.

\begin{theorem}
\label{thm:rauhut2}
Let $S \in \{1, \hdots , n\}$ be a set of cardinality $s$. 
Let $x$ be an $s$-sparse vector with support $S$ such that the signs of its non-zero entries is a Rademacher or Steinhaus sequence.
Define the acquisition set $\Omega \subseteq \{1, \hdots n\}$ as the union of:\\
\begin{enumerate}[(i)]
\item a deterministic set $\Omega_1$ of cardinality $m_1$. 
\item a random set $\Omega_2$ obtained by $m_2$ independent drawing according to distribution $p$ defined on $\{1 \hdots n \} \setminus \Omega_1$.
\end{enumerate}
Denote $m=m_1+m_2$, $\Omega_1^c=\{1,\hdots, n\} \setminus \Omega_1$ and let \DELETED{$\Ab_\Omega\in \C^{m\times n}$ be the sensing matrix built by selecting the lines of $\Ab$ corresponding to the indexes in }$\Omega=\Omega_1 \cup \Omega_2$.
Assume that:
\begin{align}
\label{eq:meas_det}
m\geqslant m_1 + C  K(\Ab_{\Omega_1^c},p) s \ln^2\left(\frac{6n}{\eta}\right)
\end{align}
where $C=7/3$ is a constant, and $\displaystyle  K(\Ab_{\Omega_1^c},p)=\max_{i\in \{1,\hdots, n\} \setminus \Omega_1} \frac{\|a_i\|_\infty^2}{p_i}$. 
Then, with probability $1-\eta$, vector $x$ is the unique solution of the $\ell_1$ minimization problem~\eqref{eq:minl1}. \\
\end{theorem}

This result implies that there exists an optimal partition between deterministically and randomly selected samples, which is moreover easy to compute. 
For example, consider the optimal distribution $p_i \propto \|a_i\|_\infty^2$, then $\displaystyle K^*(\Ab_{\Omega_1^c})=\sum_{i\in \{1,\hdots, n\} \setminus \Omega_1} \|a_i\|_\infty^2$. If the measurement matrix contains rows with large values of $\|a_i\|_\infty$, we notice from inequality~\eqref{eq:meas_det} that these frequencies should be sampled deterministically, whereas the rest of the measurements should be obtained from independent drawings. This simple idea is another way of overcoming the so-called coherence barrier~\cite{Krahmer12,Adcock13}. 

A striking example raised in~\cite{Bigot13} is the following. Assume that $\Ab=\begin{pmatrix} 1 & 0 \\ 0 & \Fb^*_{n-1} \end{pmatrix}$\DELETED{, where $\Fb^*_{n-1}\in \C^{(n-1)\times (n-1)}$ denotes the 1D discrete Fourier matrix}. The assumed optimal independent sampling strategy would consist in independently drawing the rows with distribution $p_1=1/2$ and $p_k=1/\sqrt{n-1}$ for $k\geqslant 2$. According to Theorem~\ref{thm:rauhut}, the number of required measurements is $2Cs \ln^2\left(\frac{6n}{\eta}\right)$. The alternative approach proposed in Theorem~\ref{thm:rauhut2} basically performs a deterministic drawing of the first row combined with an independent uniform drawing over the remaining rows. In total, this scheme requires $1+ Cs \ln^2\left(\frac{6n}{\eta}\right)$ measurements and thus reduces the number of measurements by almost a factor 2. Note that the same gain would be obtained by using independent drawings with rejection.

\subsubsection*{Mixed deterministic and independent sampling in MRI}

\DELETED{Recall that in MRI, the acquisition matrix is $\Ab = \Fb^* \Psib$ where $\Fb^*$ is the Fourier transform and $\Psib$ an inverse wavelet transform. }In our experiments, we will consider wavelet transforms with three decomposition levels and the Symmlet basis with 10 vanishing moments. Fig.~\ref{fig:omega1}(a)-(b) shows the modulus of $\Ab$'s entries with a specific reordering in~(b) according to decaying values of $\|a_i\|_\infty$. This decay is illustrated in Fig.~\ref{fig:omega1}(c). We observe that a typical acquisition matrix in MRI shows large differences between its $\|a_i\|_\infty$ values. More Precisely, there is a small number of rows with a large infinite norm, sticking perfectly to the framework of Theorem~\ref{thm:rauhut2}. This observation justifies the use of a partial deterministic $k$-space sampling, which had already been used in ~\cite{Lustig07,Chauffert13}. In Fig.~\ref{fig:omega1}(d), the set $\Omega_1$ is depicted for a fixed number of deterministic samples $m_1$, by selecting the rows with the largest infinite norms.

\begin{figure}[!h]
\begin{center}
\begin{tabular}{cccc}
(a) &(b)&\hspace{-.025\linewidth}(c)&(d)\\
\includegraphics[width=.27\linewidth]{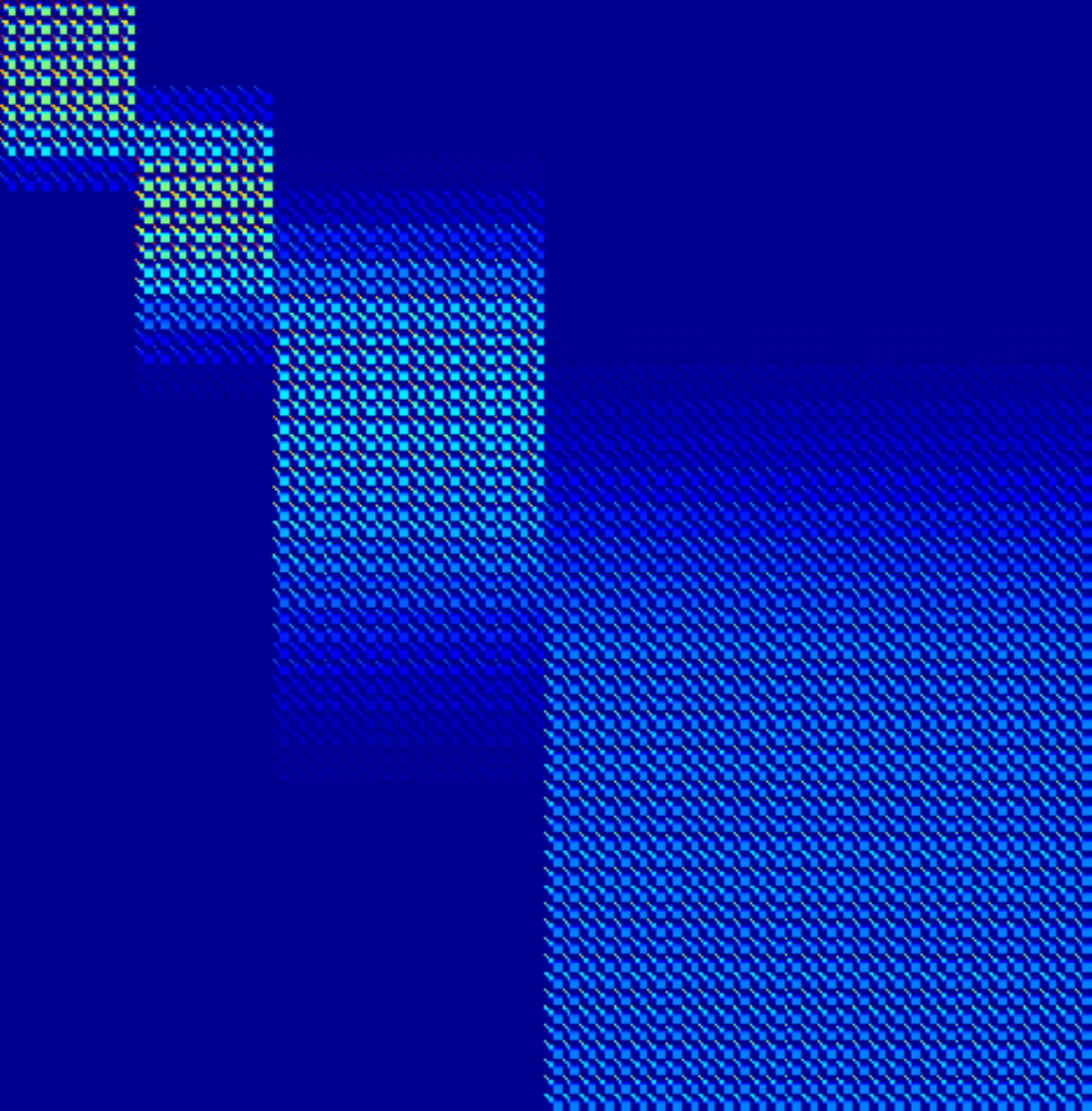}&
\includegraphics[width=.27\linewidth]{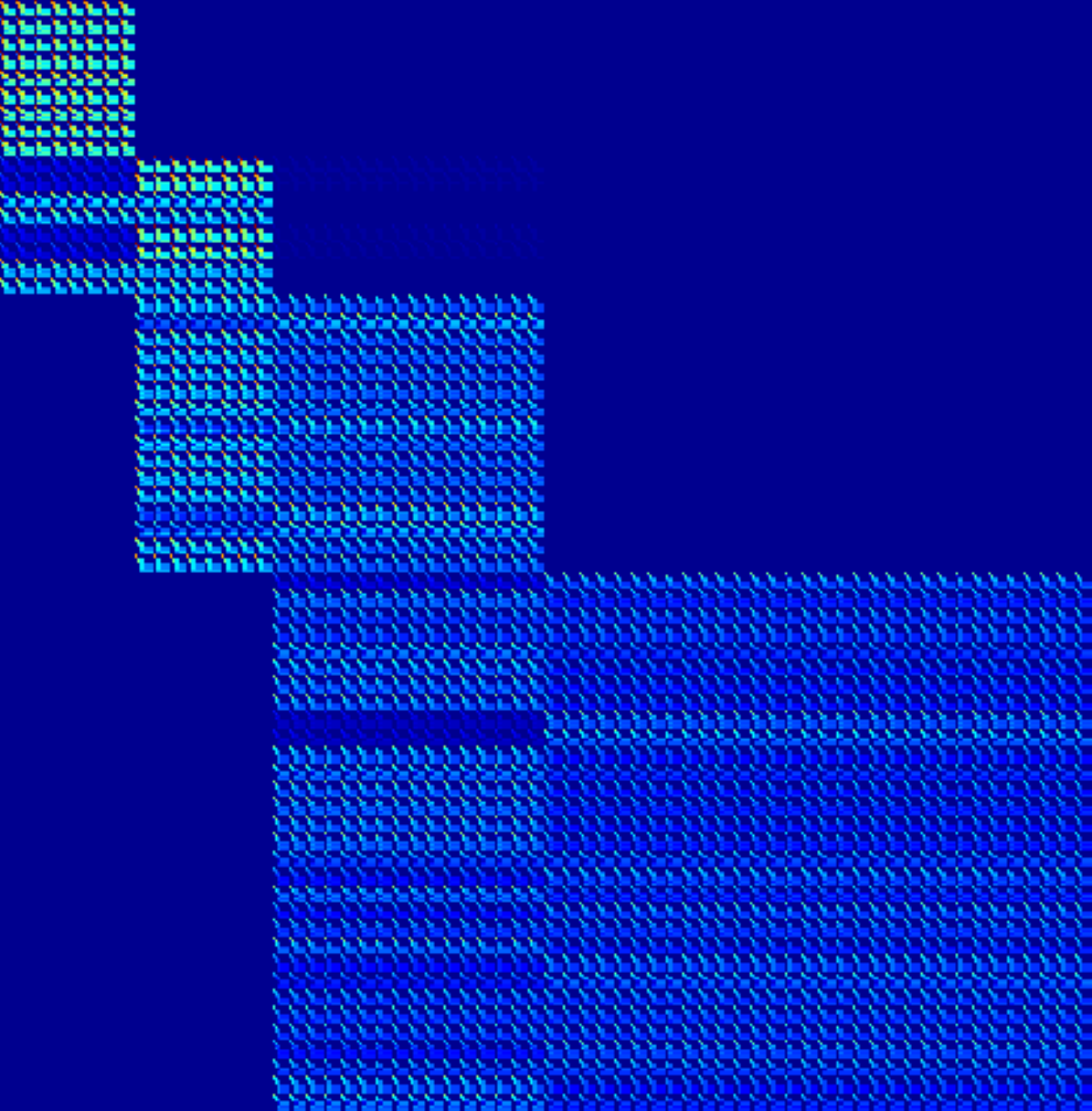}&
\hspace{-.04\linewidth}
\includegraphics[height=.27\linewidth]{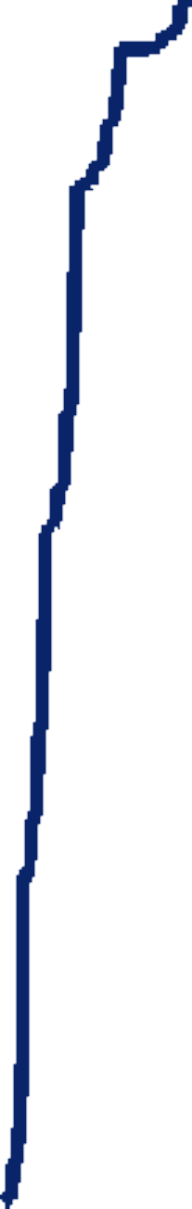}&
\includegraphics[width=.27\linewidth]{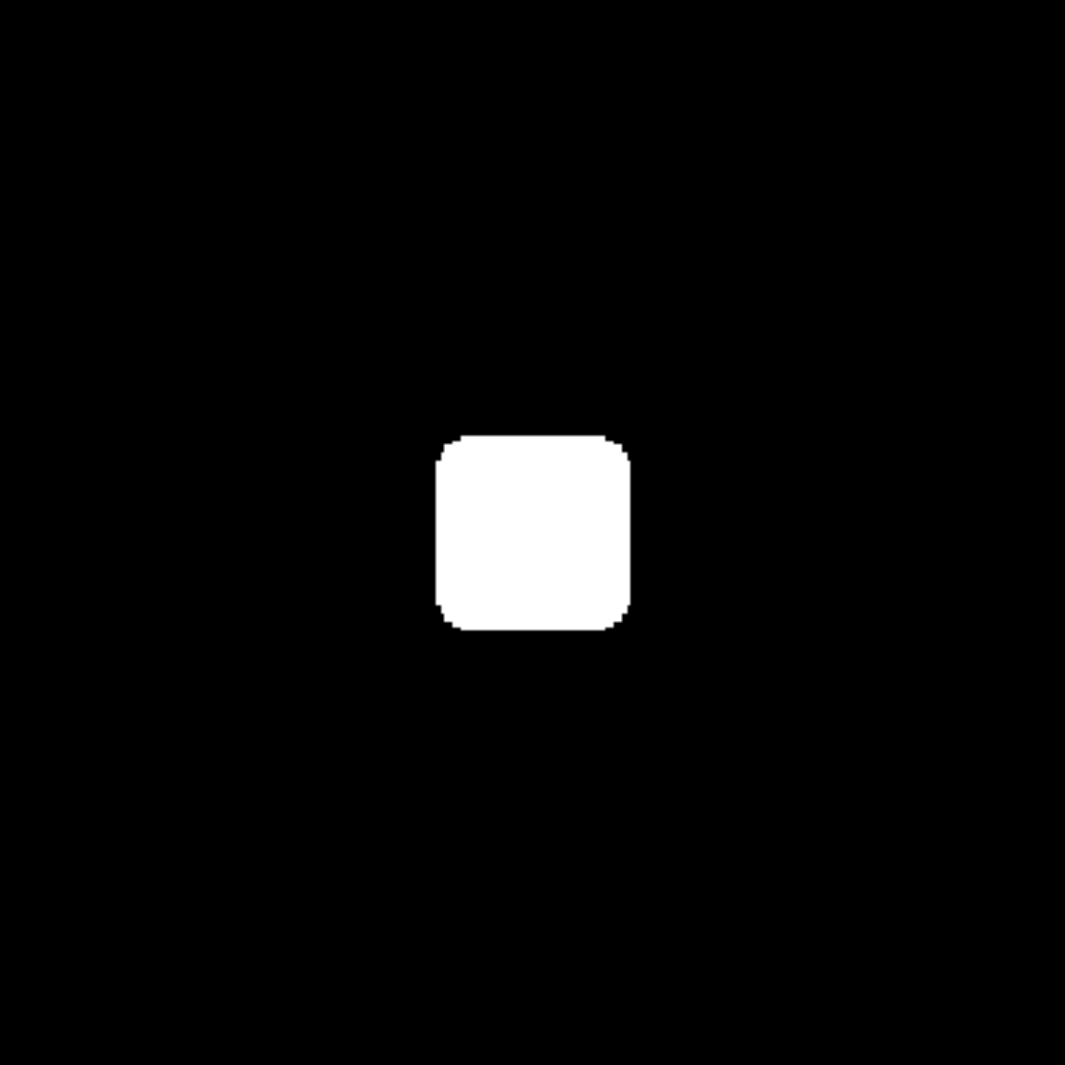}\\
\end{tabular}
\caption{\label{fig:omega1} (a): Absolute magnitudes of $\Ab$ for a 2D Symmlet basis with $10$ vanishing moments and 3 levels of decomposition. (b): same quantities as in~(a) but sorted by decaying $\|a_i\|_\infty$~(i.e. by decreasing order). (c): decay of $\|a_i\|_\infty$. (d): Set $\Omega_1$ depicted in the 2D $k$-space.}
\end{center}
\end{figure}

Hereafter, the strategy we adopt is driven by the previous remarks. All our sampling schemes are performed according to Theorem.~\ref{thm:rauhut2}: a deterministic part is sampled, and \ADDEDTHREE{a }VDS is performed on the rest of the acquisition space (\textit{e.g.} the high frequencies in MRI).

\section{Variable density samplers along continuous curves}
\label{sec:cont}

\subsection{Why independent drawing can be irrelevant}

In many imaging applications, the number of samples is of secondary importance compared to the time spent to collect the samples.
A typical example is MRI, where the important variable to control is the scanning time. 
It depends on the total length of the pathway used to visit the $k$-space rather than the number of collected samples.
MRI is not an exception and many other acquisition devices have to meet such physical constraints amongst which are
scanning probe microscopes, ultrasound imaging, ecosystem monitoring, radio-interferometry or sampling using vehicles subject to kinematic constraints~\cite{Willett}.
In these conditions, measuring isolated points is not relevant and existing practical CS approaches consist in designing parametrized curves performing a variable density sampling.
In what follows, we first review existing variable density sampling approaches based on continuous curves. 
Then, we propose two original contributions and analyze some of their theoretical properties. 
\ADDED{We mostly concentrate on continuity of the trajectory which is not sufficient for implementability in many applications. For instance, in MRI the actual requirement for a trajectory to be implementable is piecewise smoothness. More realistic constraints are discussed in Section~\ref{sec:discussion}.}

\subsection{A short review of samplers along continuous trajectories}

The prototypical variable density samplers in MRI were based on spiral trajectories~\cite{spielman1995magnetic}. Similar works investigating different shapes and densities from a heuristic point of view were proposed in~\cite{tsai2000reduced,kim2003simple,park2005artifact}.
The first reference to compressed sensing appeared in the seminal paper~\cite{Lustig07}. In this work, Lustig~\textit{et al} have proposed to perform independent drawings in a 2D plane~(defined by the partition and phase encoding directions) and sample continuously along the orthogonal direction to design piecewise continuous schemes in the 3D $k$-space~(see Fig.~\ref{fig:schemes_lustig}). These authors have also suggested to make use of randomly perturbed spirals. The main advantage of these schemes lies in their simplicity of practical implementation since they only require minor modifications of classical MRI acquisition sequences.

\begin{figure}[!h]
\begin{center}
\begin{tabular}{cc}
(a)&(b)\\
\includegraphics[height=.38 \linewidth]{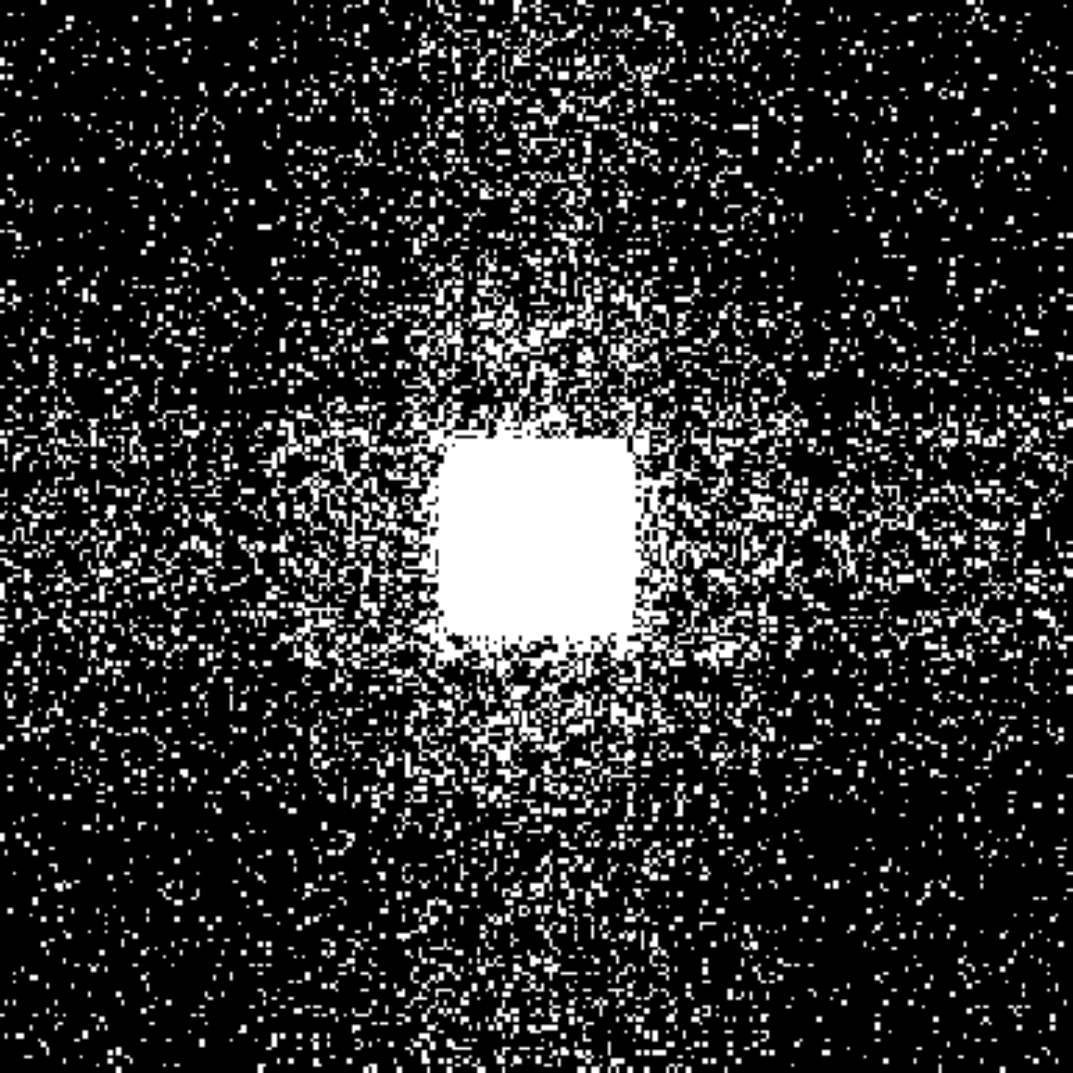}&
\includegraphics[height=.38 \linewidth]{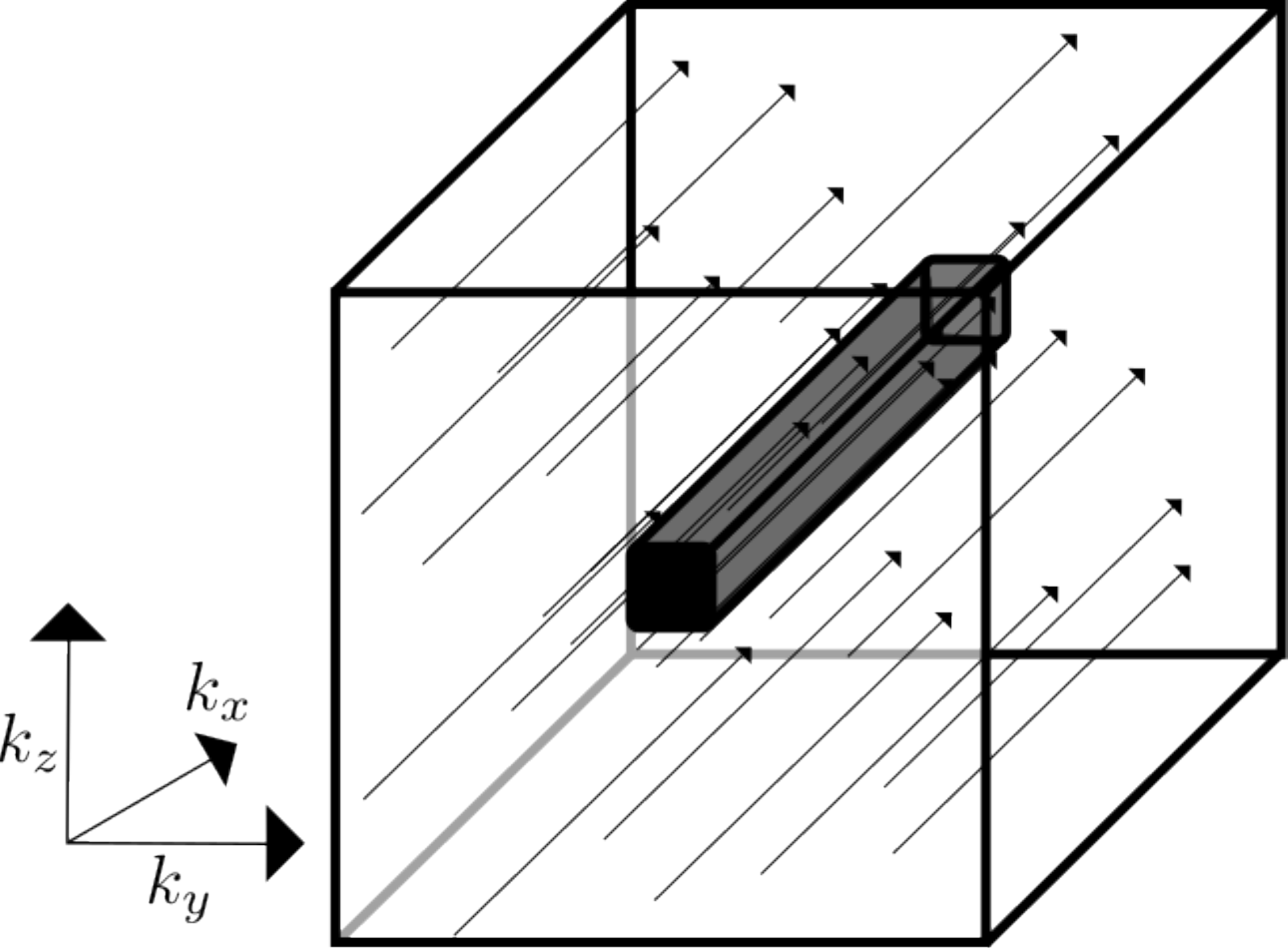}\\
\end{tabular}
\caption{\label{fig:schemes_lustig} Classical CS-MRI strategy. (a): 2D independent sampling according to a distribution $\pi$.  (b): measurements performed in the orthogonal readout direction.}
\end{center}
\end{figure}

Recent papers~\cite{polak2012grouped,Bigot13,boyer2013algorithm} have generalized CS results from independent drawing of isolated measurements to independent drawings of blocks of measurements. In these contributions, the blocks can be chosen arbitrarily and may thus represent continuous trajectories. Interestingly, these authors have provided closed form expressions for the optimal distribution on the block set. Nevertheless, this distribution is very challenging to compute in large scale problems. Moreover, the restriction to sets of admissible blocks reduces the versatility of many devices such as MRI and can therefore impact the image reconstruction quality. 

In many applications the length of the sampling trajectory is more critical than the number of acquired samples, therefore, finding the shortest pathway amongst random points drawn independently has been studied as a way of designing continuous trajectories~\cite{Willett,Wang12}. Since this problem is NP-hard, one usually resorts to a TSP solver to get a reasonable suboptimal trajectory. To the best of our knowledge, the only practical results obtained using the TSP were given by Wang~\textit{et al}~\cite{Wang12}. In this work, the authors did not investigate the relationship between the initial sample locations and the empirical measure of the TSP curve. In Section~\ref{sec:TSP}, it is shown that this relationship is crucial to make efficient TSP-based sampling schemes.

In what follows, we first introduce an original sampler based on random walks on the acquisition space and then analyse its asymptotic properties. Our theoretical investigations together with practical experiments allows us to show that the VDS mixing properties play a central role to control its efficiency. This then motivates the need for more global VDS schemes.

\subsection{Random walks on the acquisition space}
\label{sec:markov}

Perhaps the simplest way to transform independent random drawings into continuous random curves consists in performing random walks on the acquisition space.
Here, we discuss this approach and provide a brief analysis of its practical performance in the discrete setting. 
Through both experimental and theoretical results, we show that this technique is doomed to fail. However, we believe that this theoretical analysis provides a deep insight on what VDS properties characterize its performance. 

\DELETED{Let us introduce a time-homogeneous Markov chain $X=(X_n)_{n\in \mathbb{N}}$ on the set $\{1, \dots, n\}$ that represents the locations of possible measurements. 
The transition matrix of the Markov chain is denoted $\Pb\in \mathbb{R}^{n\times n}$.}
\ADDED{Let us consider a time-homogeneous Markov chain $X=(X_n)_{n\in \mathbb{N}}$ on the set $\{1, \dots, n\}$ and its transition matrix denoted $\Pb\in \mathbb{R}^{n\times n}$.}
If $X$ possesses a stationary distribution, i.e. a row vector $p\in \mathbb{R}^n$ such that $p = p \Pb$ then, by definition, $X$ is a $p$-variable density sampler. 

\subsubsection{Construction of the transition matrix $\Pb$}
A classical way to design a transition kernel ensuring that (i) $p$ is the stationary distribution of the chain and (ii) the trajectory defined by the chain is continuous, is the Metropolis algorithm~\cite{hastings1970montecarlo}. For a pixel/voxel position $i$ in the 2D/3D acquisition space, let us define by $\mathcal{N}(i)\subseteq \{1,\hdots,n\}$ its neighbourhood, \textit{i.e.} the set of possible measurement locations allowed when staying on position $i$. Let $|\mathcal{N}(i)|$ denote the cardinal of $\mathcal{N}(i)$ and define the \textit{proposal kernel} $\Pb^*$ as $\Pb^*_{i,j}=|\mathcal{N}(i)|^{-1}\delta_{j\in \mathcal{N}(i)}$. The Metropolis algorithm proceeds as follows:
\begin{enumerate}
\item from state $i$, draw a state $i^*$ with respect to the distribution $\Pb^*_{i,:}$.
\item accept the new state $i^*$ with probability:
\begin{equation}
q(i,i^*)=\min\left(1,\frac{p(i^*)\Pb^*_{i^*,i}}{p(i)\Pb^*_{i,i^*}} \right).
\end{equation}
Otherwise stay in state $i$.
\end{enumerate}
The transition matrix $\Pb$ can then be defined by $\Pb_{i,j}=q(i,j)\Pb^*_{i,j}$ for $i\neq j$. The diagonal is defined in a such a way that $\Pb$ is a stochastic matrix. It is easy to check that $p$ is an invariant distribution for this chain\footnote{If the neighboring system is such that the corresponding graph is connected, then the invariant distribution is unique.}.
It is worth noticing that if the chain is irreducible positive recurrent~(which is fulfilled if the graph is connected and the distribution $p$ positive), the ergodic theorem ensures that $X$ is a $p$-VDS.

Unfortunately, trajectories designed by this technique leave huge parts of the acquisition space unexplored~(see Fig.~\ref{fig:example_chains}~(a)). To circumvent this problem, we may allow the chain to \textit{jump} to independent locations over the acquisition space. Let $\tilde{\Pb}$ be the Markov kernel corresponding to independent drawing with respect to $p$, \textit{i.e.} $\tilde{\Pb}_{i,j}=p_j$ for all $1\leqslant i , j \leqslant n$. Define:
\begin{equation}
\label{eq:p_alpha}
\Pb^{(\alpha)}=(1-\alpha)\Pb + \alpha\tilde{\Pb} \qquad \forall \,\, 0 \leqslant \alpha \leqslant 1.
\end{equation}
Then the Markov chain associated with $\Pb^{(0)}$ corresponds to a continuous random walk, while the Markov chain associated with $\Pb^{(\alpha)}$, $\alpha>0$ has a nonzero \textit{jump} probability. This means that the trajectory is composed of continuous parts of average length $1/\alpha$.

\subsubsection{Example}
In Fig.~\ref{fig:example_chains}, we show illustrations in the 2D MRI context where the discrete $k$-space is of size $64\times 64$. On this domain, we set a distribution $p$ which  matches distribution $\pi$ in Fig.~\ref{fig:pi_opt}~(a). We perform a random walk on the acquisition space until $10\%$ of the coefficients are selected. In Fig.~\ref{fig:example_chains}(a), we set $\alpha=0$ whereas $\alpha=0.1$ in Fig.~\ref{fig:example_chains}(b). As expected, $\alpha=0$ leads to a sampling pattern where large parts of the $k$-space are left unvisited. The phenomenon is partially corrected using a nonzero value of $\alpha$.
\begin{figure}[!h]
\centering
\begin{tabular}{cc}
(a)&(b)\\
\includegraphics[width=.35\linewidth]{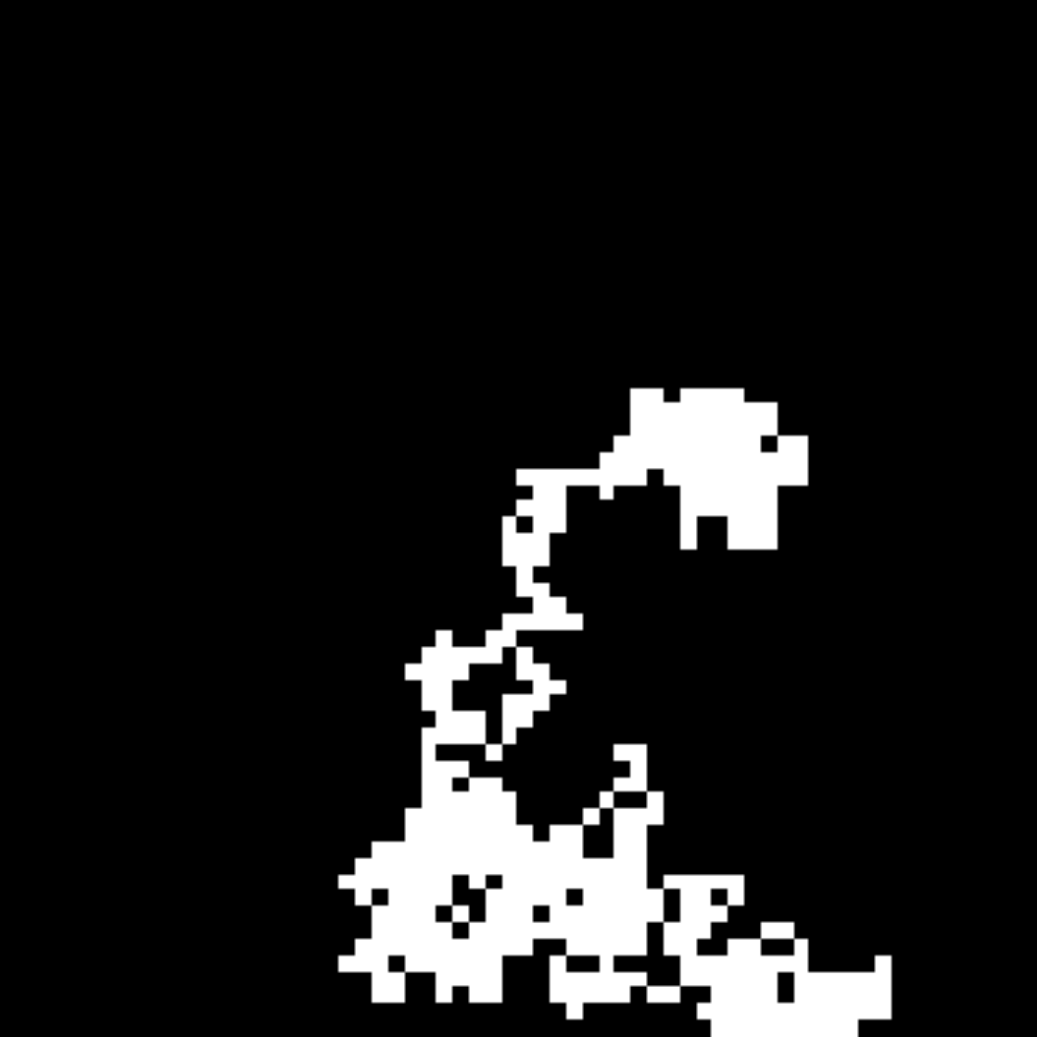}&
\includegraphics[width=.35\linewidth]{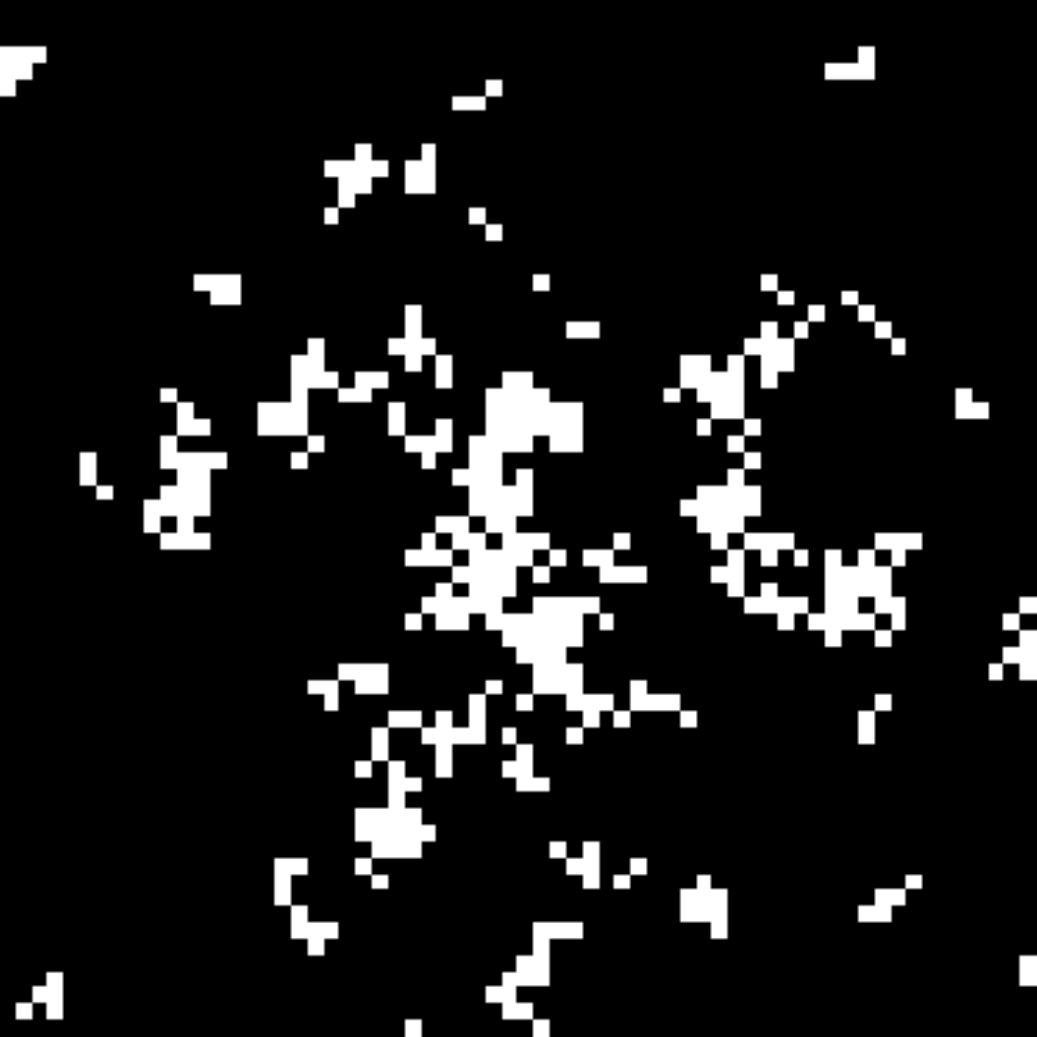}
\end{tabular}
\caption{\label{fig:example_chains} Example of sampling trajectories in 2D MRI. (a) (resp. (b)): 2D sampling scheme of the $k$-space with $\alpha=0$ (resp. $\alpha=0.1$). Drawings are performed until $10\%$ of the coefficients are selected~($m=0.1 n$).}
\end{figure}

\ADDEDONE{
\begin{remark}
Performing $N$ iterations of the Metropolis algorithm requires $O(N)$ computations leading to a fast sampling scheme design procedure. In our experiments, we iterate the algorithm until $m$ \emph{different} measurements are probed. Therefore, the number of iterations $N$ required increases non linearly with respect to $m$, and can be time consuming especially when $R=m/n$ is close to $1$. This is not a tough limitation of the method since the sampling scheme is computed off-line.
\end{remark}
}

\subsubsection{Compressed sensing results}
\REPLACED{To derive theoretical reconstruction guarantees in the context of continuous sampling, let us define the spectral gap of the Markov transition kernel $\Pb$ by $\epsilon(\Pb)= 1- \lambda_2(\Pb)$ where $\lambda_2(\Pb)$ is the second largest eigenvalue of $\Pb$. We further}{Let us} assume\footnote{By making this assumption, there is no burn-in period and the chain $X$ converges more rapidly to its stationary distribution $p$.} that $\Prob(X_1=i)=p_i$ and that $X_i$ is drawn using $\Pb$ as a transition matrix. 
The following result provides theoretical guarantees about the performance of the VDS $X$:
\begin{proposition}[see~{\cite{Chauffert13c}}] 
\label{prop:measurements_needed}
Let $\Omega\ADDED{:= X_1, \hdots, X_m}\subset\{1,\hdots,n\}$ denote a set of $m$ indexes selected using the Markov chain $X$.
\DELETED{Let $\Ab_\Omega$ be a measurement matrix designed by selecting the lines of $\Ab$ in $\Omega$:}

 Then, with probability $1-\eta$, if
\begin{equation}
\label{eq:measurements}
m \geqslant \frac{12}{\epsilon(\Pb)} K^2(\Ab,p) s^2 \log(2n^2/\eta),
\end{equation}
every  $s$-sparse signal $x$ is the unique solution of the $\ell_1$ minimization problem.
\end{proposition}

The proof of this proposition is given in Appendix~2.
Before going further, some remarks may be useful to explain this theoretical result.

\begin{remark}
Since the constant $ K^2(\Ab,p)$ appears in Eq.~\eqref{eq:measurements}, the optimal sampling distribution using Markov chains is also distribution $\pi$, as proven in Proposition~\ref{prop:theoopt}.
\end{remark}

\begin{remark}
In contrast to Theorem~\ref{thm:rauhut}, Proposition \ref{prop:measurements_needed} provides uniform results, i.e. results that hold for \textit{all} $s$-sparse vectors. 
\end{remark}

\begin{remark}
Ineq.~\ref{eq:measurements} suffers from the so-called \textit{quadratick bottleneck} (i.e. an $O(s^2\log(n))$ bound). 
It is likely that this bound can be improved to $O(s\log(n))$ by developing new concentration inequalities on matrix-valued Markov chains.
\end{remark}

\begin{remark}
\label{rmk:SpectralGap}
More importantly, it seems however unlikely to avoid the spectral gap  $O(1/\epsilon(P))$ using the standard mechanisms for proving compressed sensing results. Indeed, all concentration inequalities obtained so far on Markov chains~(see e.g. \cite{Lezaud98,Kargin07,paulin2012concentration}) depend on $1/\epsilon(P)$. 
The spectral gap satisfies $0 < \epsilon(P) \leqslant 1$ and corresponds to mixing properties of the chain. 
The closer the spectral gap to 1, the fastest ergodicity is achieved.
Roughly speaking, if $|i-j|>1/\epsilon(P)$ then $X_i$ and $X_j$ are almost independent random variables. 
Unfortunately, the spectral gap usually depends on the dimension $n$~\cite{diaconis1991geometric}.
In our example, it can be shown using Cheeger's inequality that $\epsilon(P)=O\left( n^{-\frac{1}{d}} \right)$ if the stationary distribution $\pi$ is uniform (see Appendix~3). 
This basically means that the number of measurements necessary to accurately reconstruct $x$ could be as large as $O(s n^{1/d} \log(n))$, which strongly limits the interest of this CS approach. 
The only way to lower this number consists in frequently jumping since Weyl's theorem~\cite{Horn91} ensures that $\epsilon(P^{(\alpha)})>\alpha$.
\end{remark}

To sum up, the main drawback of random walks lies in their inability to cover the acquisition space quickly since they are based on local considerations. 
Keeping this in mind, it makes sense to focus on more global displacement strategies that allow a faster exploration of the whole acquisition domain.
In the next section, we thus introduce this global sampling alternative based on TSP-solver.
Our main contribution is the derivation of the link between a prescribed a priori sampling density and the distribution of samples located on the TSP solution so as to eventually get a VDS.

\section{Travelling salesman-based VDS}
\label{sec:TSP}
In order to design continuous trajectories, we may think of picking points at random and join them using a travelling salesman problem~(TSP) solver. Hereafter, we show how to draw the initial points in order to reach a target distribution $p$. In this section, the probability distribution $p$ is assumed to be a density. 

\subsection{Introduction}
The naive idea would consist in drawing some points according to the distribution $p$ and joining them using a TSP solver. Unfortunately, the trajectory which results from joining all samples does not fit the distribution $p$, as shown in Fig.~\ref{fig:TSP2}(b)-(d). To bring evidence to this observation, we performed a Monte Carlo study, where we drew one thousand sampling schemes, each one designed by solving the TSP on a set of independent random samples. We notice in Fig.~\ref{fig:TSP2}~(d) that the empirical distribution of the points along the TSP curve, hereafter termed the final distribution, departs from the original distribution $p$.
A simple intuition can be given to explain this discrepancy between the initial and final distributions in a $d$-dimensional acquisition space. Consider a small subset of the acquisition space $\omega$. In $\omega$, the number of points is proportional to $p$. The typical distance between two neighbors in $\omega$ is then proportional to $p^{-1/d}$. Therefore, the local length of the trajectory in $\omega$ is proportional to $p p^{-1/d} = p^{1-1/d} \neq p$. In what follows, we will show that the empirical measure of the TSP solution converges to a measure proportional to $p^{1-1/d}$.

\subsection{\REPLACED{Notation and d}{D}efinitions}

We shall work on the hypercube $\setOmega = [0,1]^d$ with $d \geqslant 2$. \DELETED{Let $h\in \N$. The set $\setOmega$ will be partitioned in $h^d$ congruent hypercubes $(\omega_i)_{i\in I}$ of edge length $1/h$.}
In what follows, $\left\{ x_i \right\} _{i\in \mathbb{N}^* }$ denotes a sequence of points in the hypercube $\setOmega$, independently drawn from a density $p:\setOmega \mapsto \R_+$. The set of the first $N$ points is denoted $X_N = \left\{ x_i \right\} _{i\leqslant N}$. 
\DELETED{For a set of points $F$, we consider the solution to the TSP, that is the shortest Hamiltonian path between those points.\\}

\DELETED{We denote $T(F)$ its length. For any set $R\subseteq \setOmega $ we define $T(F, R) = T(F \cap R)$.}

\DELETED{We also introduce $C(X_N, \setOmega )$ to denote the optimal curve itself, and $\gamma _N: [0,1] \to \setOmega $ the function that parameterizes $C(X_N, \setOmega )$ by moving along it at constant speed $T(X_N, \setOmega )$.}
\ADDED{Using the definitions introduced in Tab.~\ref{Tab:not}, we introduce $\gamma _N: [0,1] \to \setOmega $ the function that parameterizes $C(X_N)$ by moving along it at constant speed $T(X_N, \setOmega )$.}
\DELETED{The Lebesgue measure on an interval $[0,1]$ is denoted $\lambda _{[0,1]}$.} Then, the \textit{distribution of the TSP solution} reads as follows:

\begin{definition}
    The distribution of the TSP solution is denoted $\tilde{P}_N$ and defined, for any Borelian $B$ in $\setOmega $ by:
    \begin{align*}
        \tilde{P}_N(B) & = \lambda _{[0,1]} \left( \gamma _N^{-1} (B) \right).
    \end{align*}
\end{definition}

\begin{remark}
    The distribution $\tilde{P}_N$ is defined for fixed $X_N$. It makes no reference to the stochastic component of $X_N$.
\end{remark}
\begin{remark}
A more intuitive definition of $\tilde{P}_N$ can be given if we introduce other tools. 
For a subset $\omega \subseteq \setOmega $, we denote the length of $\REPLACED{C(X_N, \setOmega )}{C(X_N)}\cap \omega$ as $T_{|\omega }(X_N, \setOmega )=T(X_N, \setOmega ) \tilde{P}_N(\omega )$. Using this definition, it follows that:
\begin{equation}
\label{eq:defalternative}
\tilde{P}_N(\omega) = \frac{T_{|\omega}(X_N, \setOmega )}{T(X_N,\setOmega)}, \ \forall \omega. 
\end{equation}
Then $\tilde{P}_N(\omega)$ is the relative length of the curve inside $\omega$.
\end{remark}

\subsection{Main results}
\label{part:thmTSP}
Our main theoretical result introduced in~\cite{Chauffert13b} reads as follows: \\

\begin{theorem}
    \label{thm:convergence_proba}
Define the density $\displaystyle \tilde p = \frac{p^{(d-1)/d}}{\int_{\setOmega} p^{(d-1)/d}(x) dx}$ where $p$ is a density defined on $\setOmega$. Then almost surely with respect to the law $p^{\otimes \mathbb{N} }$ of the random points sequence $\{x_i\}_{i\in \mathbb{N} ^*}$ in $\setOmega$, the distribution $\tilde{P}_N$ converges in distribution to $\tilde p$:
    \begin{align}
        \label{convForm}
        \tilde P_N & \stackrel{(d)}{\rightarrow} \tilde p & \mbox{$p^{\otimes \mathbb{N}}$-a.s.}
    \end{align}
\end{theorem}

The proof of the theorem is given in Appendix~4. 
\begin{remark} The $TSP$ solution does not define as such a VDS, since the underlying process is finite in time. \DELETED{However, the sequence of random measures $\{\tilde P_N\}_{N\in \mathbb{N}}$ almost surely converges weakly to $\tilde p$ as $N\to \infty$.} 
\ADDED{Nevertheless, since $\tilde P_N$ is the occupation measure of $\gamma_N$, the following result holds:
\begin{corollary}
$(\gamma_N)_{N\in\N}$ is a generalized $\tilde p$ VDS.
\end{corollary}}
\end{remark}

\begin{remark} The theorem indicates that if we want to reach distribution $p$ in 2D, we have to draw the initial points with respect to a distribution proportional to $p^2$, and to $p^{3/2}$ in 3D. Akin to the previous Monte Carlo study illustrating the behavior of the naive approach in Fig.~\ref{fig:TSP2} (top row), we repeated the same procedure after having taken this result into account. The results are presented in Fig.~\ref{fig:TSP2}(e)-(g), in which it is shown that the final distribution now closely matches the original one~(compare Fig.~\ref{fig:TSP2}(g) with Fig.~\ref{fig:TSP2}(a)).

\begin{figure}[!h]
\centering
\begin{tabular}{cccc}
&  (b)&(c)&(d) \\
\multirow{4}*{\begin{minipage}{.3 \linewidth} \vspace{-.4 \linewidth}\begin{tabular}{cc}(a) & \multirow{2}*{\includegraphics[height=.83\linewidth]{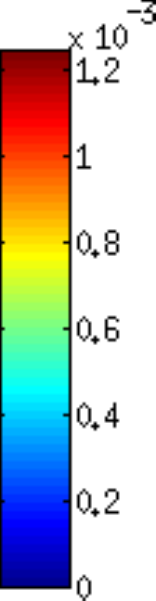}} \\ \includegraphics[width=.66\linewidth]{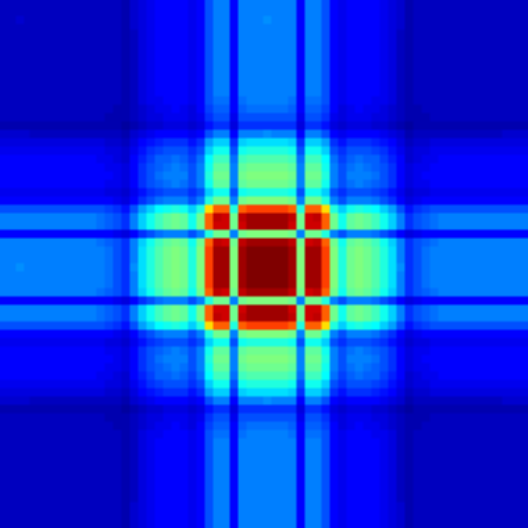} & \end{tabular}\end{minipage}}&
\includegraphics[width=.2\linewidth]{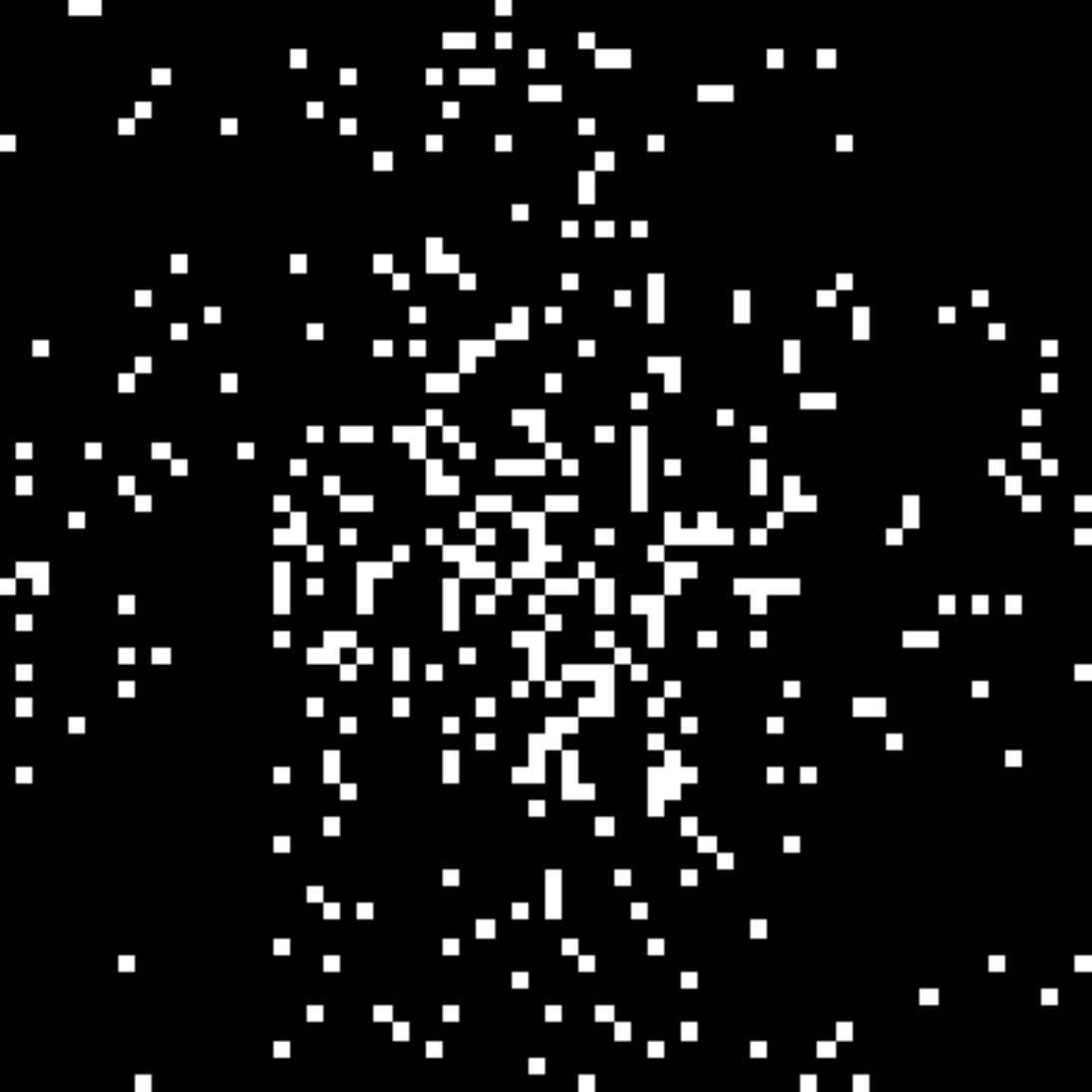}&
\includegraphics[width=.2\linewidth]{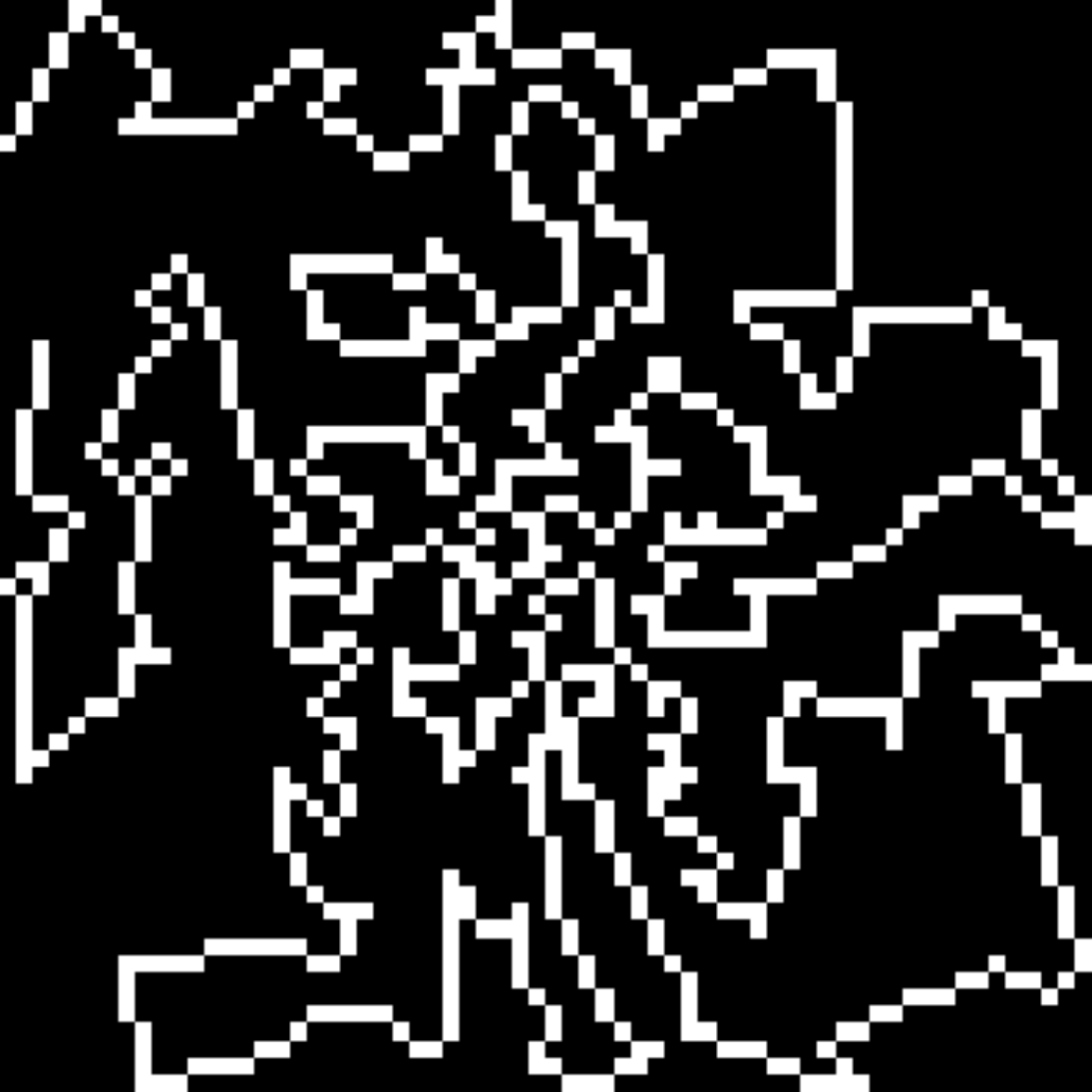}&
\includegraphics[width=.2\linewidth]{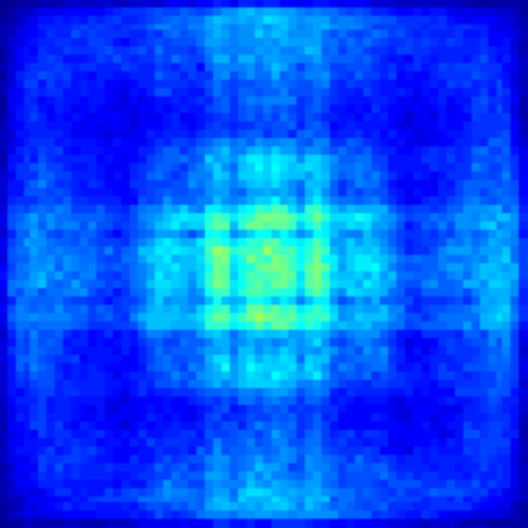}\\
& (e)&(f)&(g) \\
& 
\includegraphics[width=.2\linewidth]{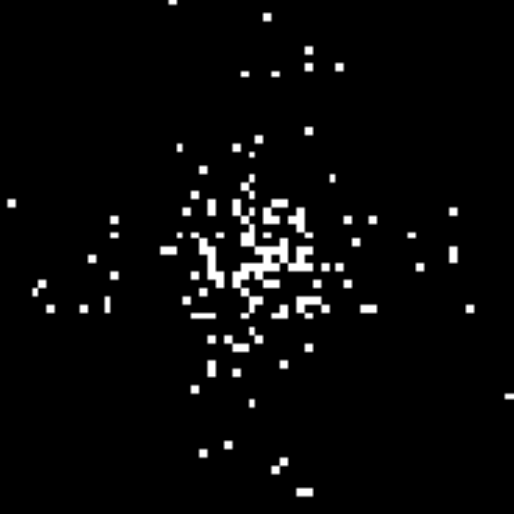}&
\includegraphics[width=.2\linewidth]{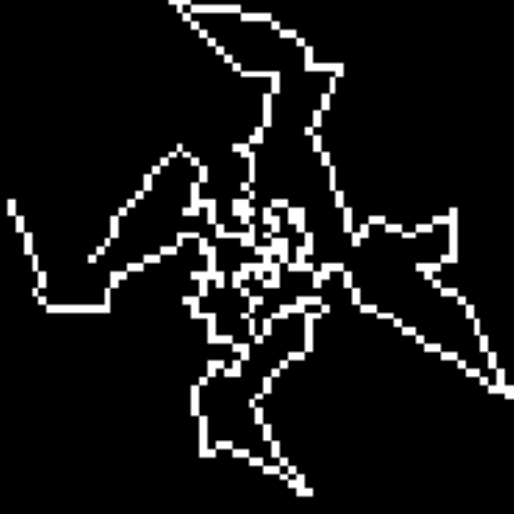}&
\includegraphics[width=.2\linewidth]{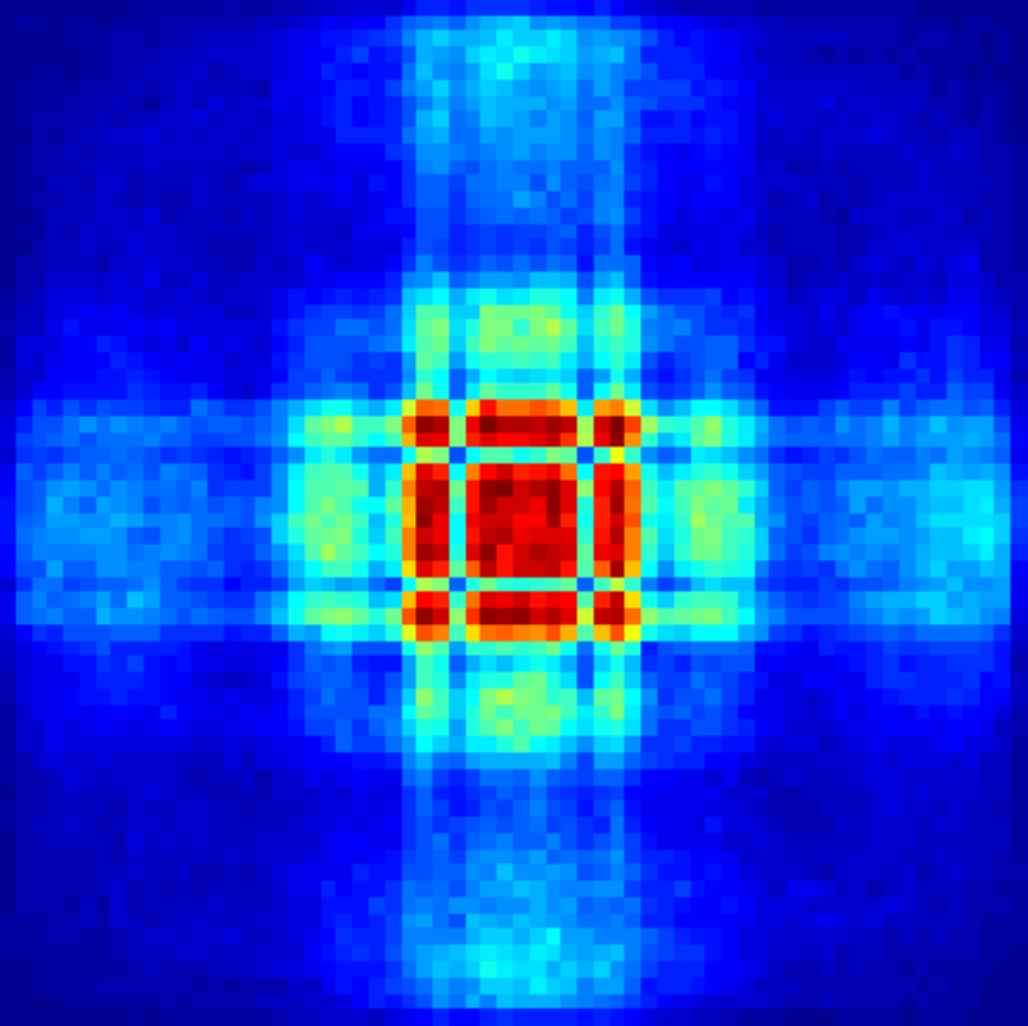}\\
\end{tabular}
\caption{\label{fig:TSP2} Illustration of the TSP-based sampling scheme to reach distribution $\pi$. (a): distribution $\pi$. (b) (resp. (e)): independent drawing of points from distribution $\pi$ (resp. $\propto \pi^2$). (c) (resp (f)): solution of the TSP amongst points of (b) (resp. (e)) . (d) and (g): Monte Carlo study: average scheme over one thousand drawings of sampling schemes, with the same color scale as in~(a).} 
\end{figure}
\end{remark}

\begin{remark}
Contrarily to the Markov chain approach for which we derived compressed sensing results in Proposition \ref{prop:measurements_needed}, the TSP approach proposed here is mostly heuristic and based on the idea that the TSP solution curve covers the space rapidly. An argument supporting this idea is the fact that in 2D, the TSP curve $\REPLACED{C(X_N, \setOmega )}{C(X_N)}$ does not self-intersect. This property is clearly lacking for random walks.
\end{remark}

\begin{remark}
One of the drawback of this approach is the TSP's NP-hardness. We believe that this is not a real problem. \REPLACED{First, the computation is done off-line from the acquisition procedure. The computation time is thus not a real limitation. Moreover, }{Indeed, }there now exist very efficient approximate solvers such as the Concorde solver~\cite{Concorde06}. It finds an approximate solution with $10^5$ cities from a few seconds to a few hours depending on the required accuracy of the solution. \ADDED{The computational time of the approximate solution is not a real limitation since the computation is done off-line from the acquisition procedure.}
Moreover, many solvers are actually designed in such a way that their solution also fulfil Theorem.~\ref{thm:convergence_proba}. \ADDEDONE{For example, in 2D, to reach a sampling factor of $R=5$ on a $256\times 256$ image, one need $N\simeq10^4$ cities, and an approximate solution is obtained in $142$s. In 3D, for a $256\times 256 \times 256$ image, $N\simeq9 \ 10^5$ and an approximate solution is obtained in about 4 hours. In each case the solutions seem to be correctly approximated. In particular they do not self-intersect in 2D.}
\end{remark}

\section{Experimental results in MRI}
\label{sec:results}
In this section, we focus on the reconstruction results by minimizing the $\ell_1$ problem~\eqref{eq:minl1} with a simple MRI model: $\Ab=\Fb^* \Psib$, where $\Psib$ denote the inverse Symmlet-10 transform\footnote{We focused on $\ell_1$ reconstruction since it is central in the CS theory. The reconstruction quality can be improved by considering more \textit{a priori} knowledge on the image. Moreover we considered a simple MRI model, but our method can be extended to parallel MRI~\cite{Pruessmann99}, or spread-spectrum techniques~\cite{Haldar11,Puy12}.}. The solution is computed using Douglas-Rachford's algorithm~\cite{Combettes11b}. We consider an MR image of size $256\times 256 \times 256$ as a reference, and perform reconstruction for different discrete sampling strategies. Every sampling scheme was regridded using a nearest neighbour approach to avoid data interpolation.\footnote{\ADDEDTHREE{We provide Matlab codes to reproduce the proposed experiments here: http://chauffertn.free.fr/codes.html}} \\

\subsection{2D-MRI}

In 2D, we focused on a single slice of the MR image and considered its discrete Fourier transform as the set of possible measurements. First, we found the best made a comparison of independent drawings with respect to various distributions in order to find heuristically the best sampling density. Then we explored the performance of the two proposed methods to design continuous schemes: random walks and Travelling Salesman Problem. We also compared our solution to classical MRI sampling schemes. In every sampling schemes, the number of measurements is the same and equals $20\%$ of the number of pixels in the image, so that the \REPLACEDTHREE{\emph{acceleration factor}}{\emph{sampling factor}}
~$R$ is equal to $5$.
In cases where the sampling strategy is based on randomness (VDS, random walks, TSP...), we performed a Monte Carlo study by generating 100 sampling patterns for each variable density sampler.

\subsubsection{Variable density sampling using independent drawings}
Here, we assessed the impact of changing the sampling distribution using independent drawings. In all experiments, we sampled the Fourier space center deterministically as shown on Figure \ref{fig:cont2D}.

\begin{table}[!h]
\begin{center}
\caption{\label{tab:res2D_indep} Quality of reconstruction results in terms of PSNR for 2D sampling with variable density independent drawings.}
\begin{tabular}{l|c|c|c|c|c|c|c|}
\cline{2-8}
	  & \multirow{2}*{$\pi$} & \multicolumn{6}{|c|}{polynomial decay: $(k_x^2+k_y^2)^{-d/2}$}  \\
 \cline{3-8}
  	  &	  & $d=1$ & $\mathbf{d=2}$ & $d=3$ & $d=4$ & $d=5$ & $d=6$ \\ 
\hline
\multicolumn{1}{|l|}{mean PSNR (dB)} & 35.6  & \textbf{36.4} & \textbf{36.4} & 36.3 & 36.0 & 35.5 & 35.2 \\ \hline
\multicolumn{1}{|l|}{\ADDEDTHREE{std dev.}}       &  $<0.1$    &  $<0.1$   & $<0.1$    &  $<0.1$   &  $<0.1$   &  $<0.1$  & $<0.1$   \\ \hline
\end{tabular}
\end{center}
\end{table}

Table \ref{tab:res2D_indep} shows that the theoretically-driven optimal distribution $\pi$ is outperformed by the best heuristics. Amongst the latter, the distribution leading to the best reconstruction quality decays as $1/|k|^2$, which is the distribution used by Krahmer and Ward~\cite{Krahmer12} as an approximation of $\pi$ for Haar wavelets. 
The standard deviation of the PSNR is negligible compared to the mean values and for a given distribution, each reconstrucion PSNR equals its average value at the precision used in Tab.~\ref{tab:res2D_indep}.

\subsubsection{Continuous VDS}
In this part we compared various variable density samplers:
\begin{itemize}
 \item Random walks with a stationary distribution proportional to $1/|k|^2$ and different average chain lengths of $1/\alpha$,
 \item TSP-based sampling with distributions proportional to $1/|k|^2$ and $\pi$,
 \item Classical MRI sampling strategies such as spiral, radial and radial with random angles. \ADDED{The choice of the spiral follows Example~\ref{ex:spiral}: the spiral is parameterized by $s: [0,T] \rightarrow \R^2$, $\theta  \mapsto r(\theta/T)\begin{pmatrix}\cos \theta \\ \sin \theta \end{pmatrix}$ where $r(t):=\frac{r(0)r(1)}{r(1)-t(r(1)-r(0))}$, so as the spiral density decays as $1/|k|^2$.}
\end{itemize}
The sampling schemes are presented in Fig.~\ref{fig:cont2D} and the reconstruction results in Tab.~\ref{tab:cont2D}.

\begin{figure}[!h]
\begin{center}
\begin{tabular}{cccc}
(a)&(b)&(c)&(d) \\
\includegraphics[width=.2\linewidth]{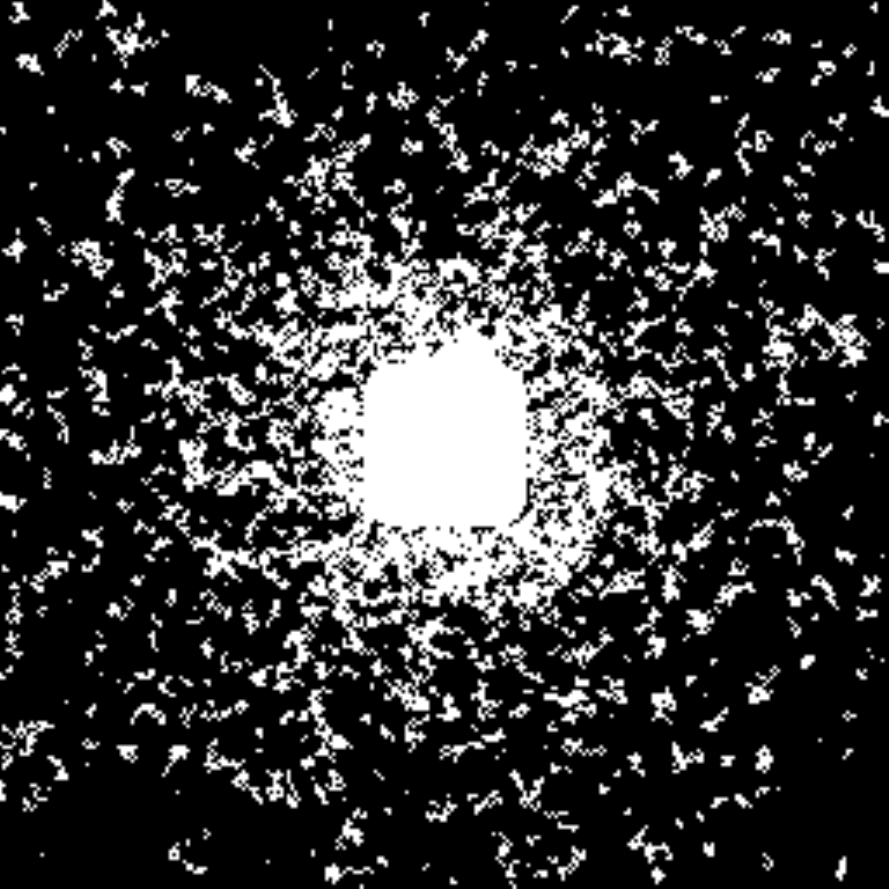}&
\includegraphics[width=.2\linewidth]{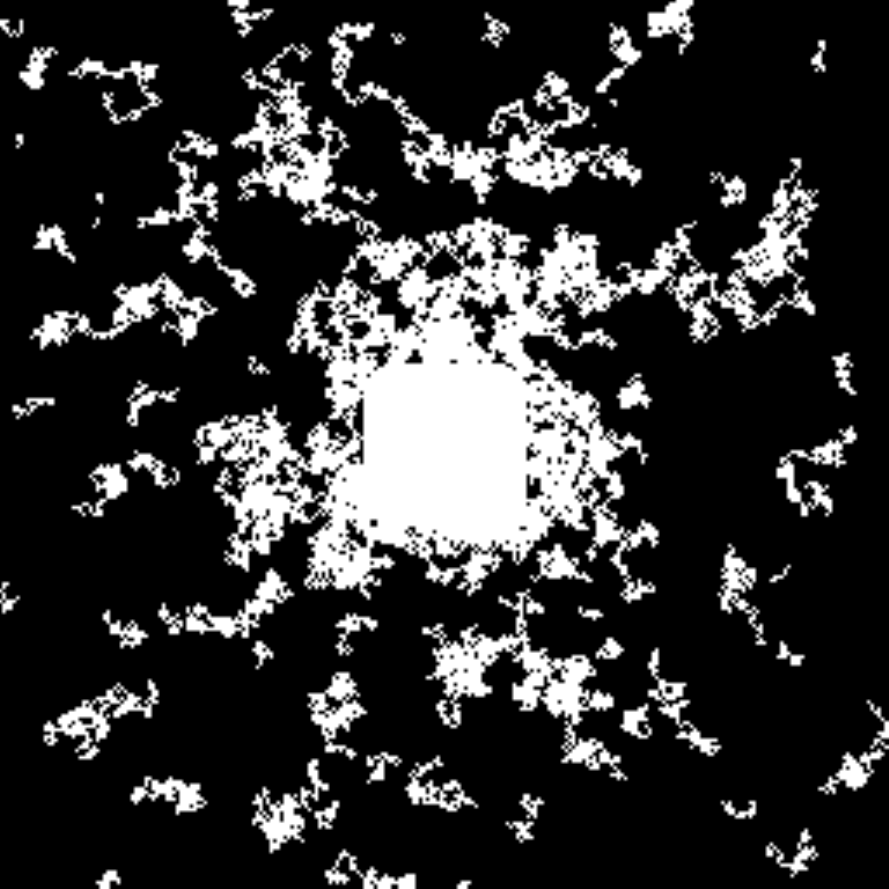}&
\includegraphics[width=.2\linewidth]{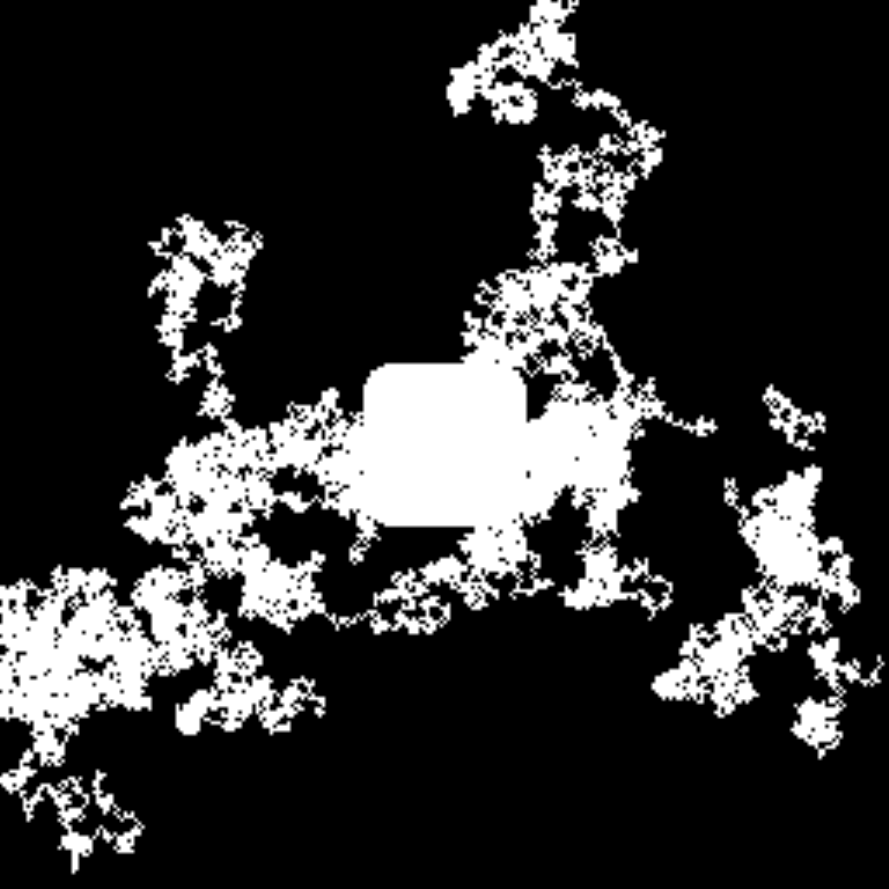}&
\includegraphics[width=.2\linewidth]{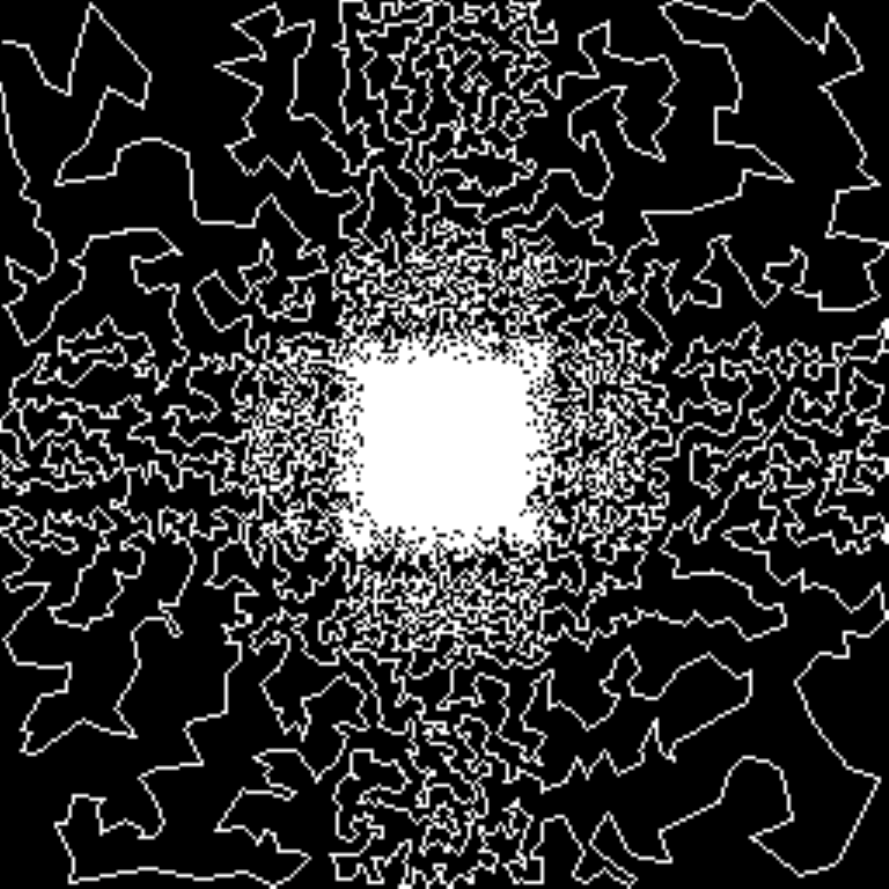}\\
(e)&(f)&(g)&(h)\\ 
\includegraphics[width=.2\linewidth]{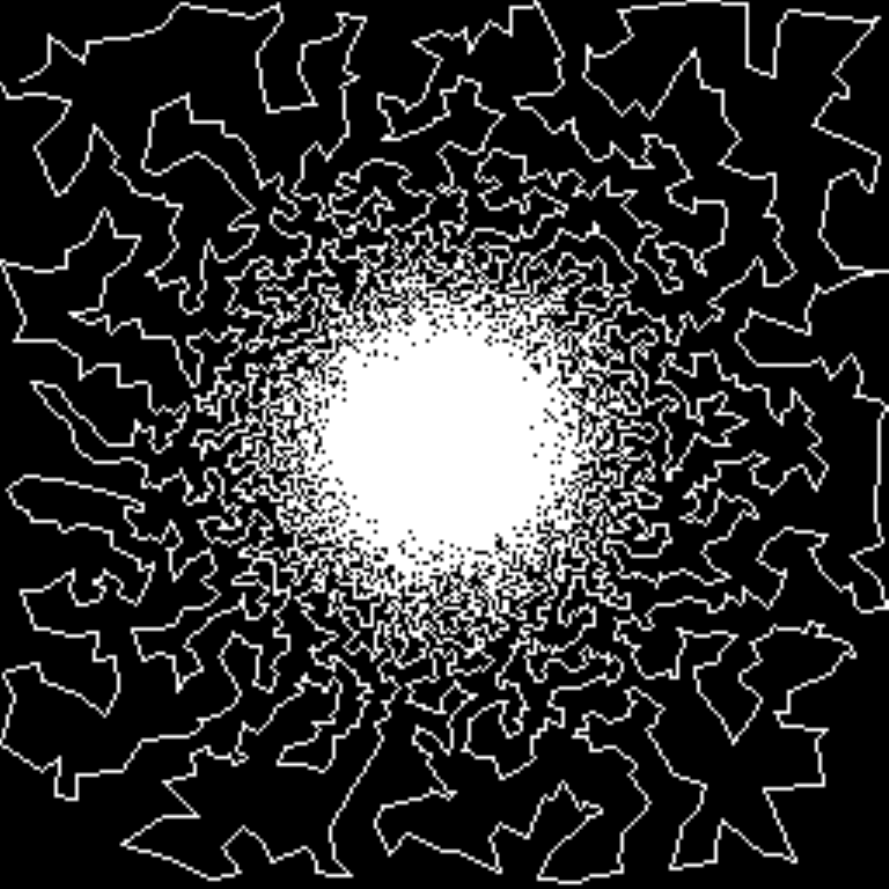}&
\includegraphics[width=.2\linewidth]{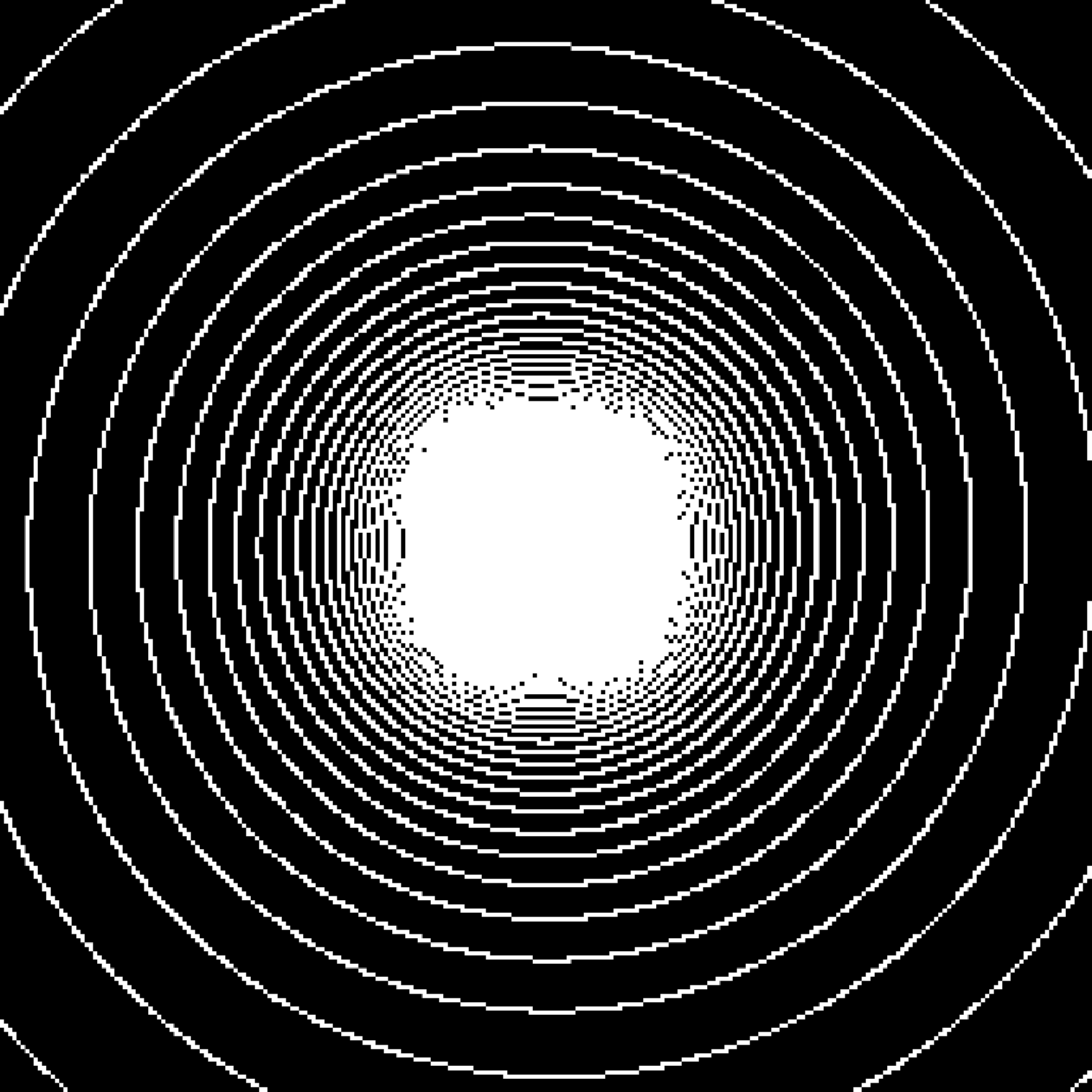}&
\includegraphics[width=.2\linewidth]{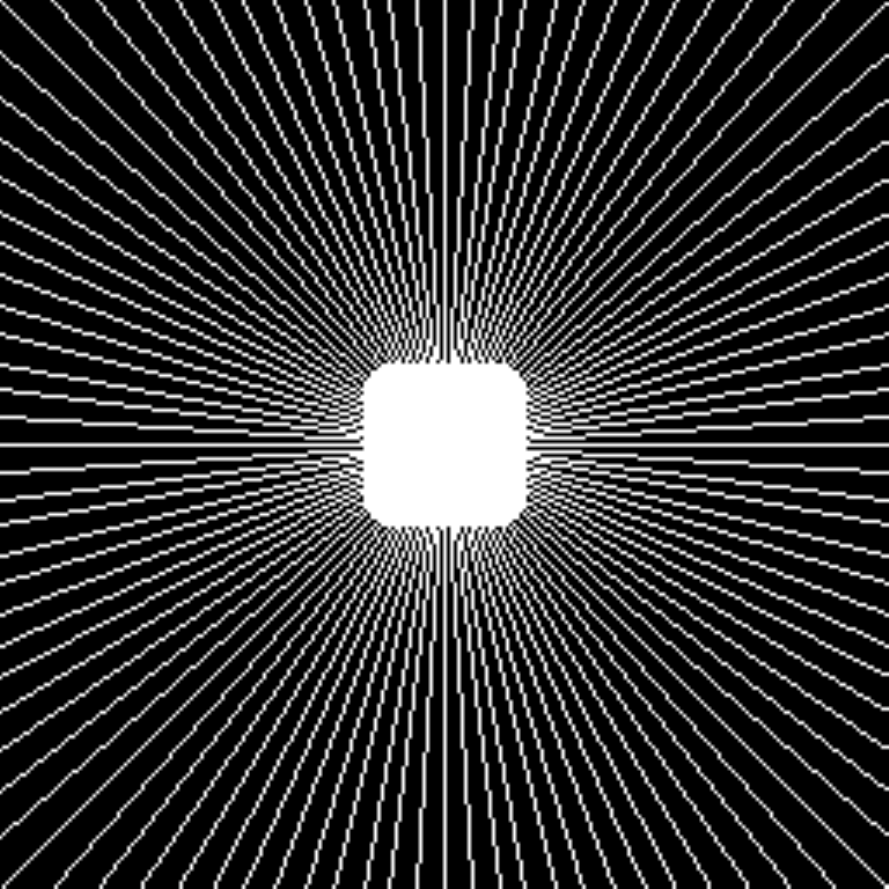}&
\includegraphics[width=.2\linewidth]{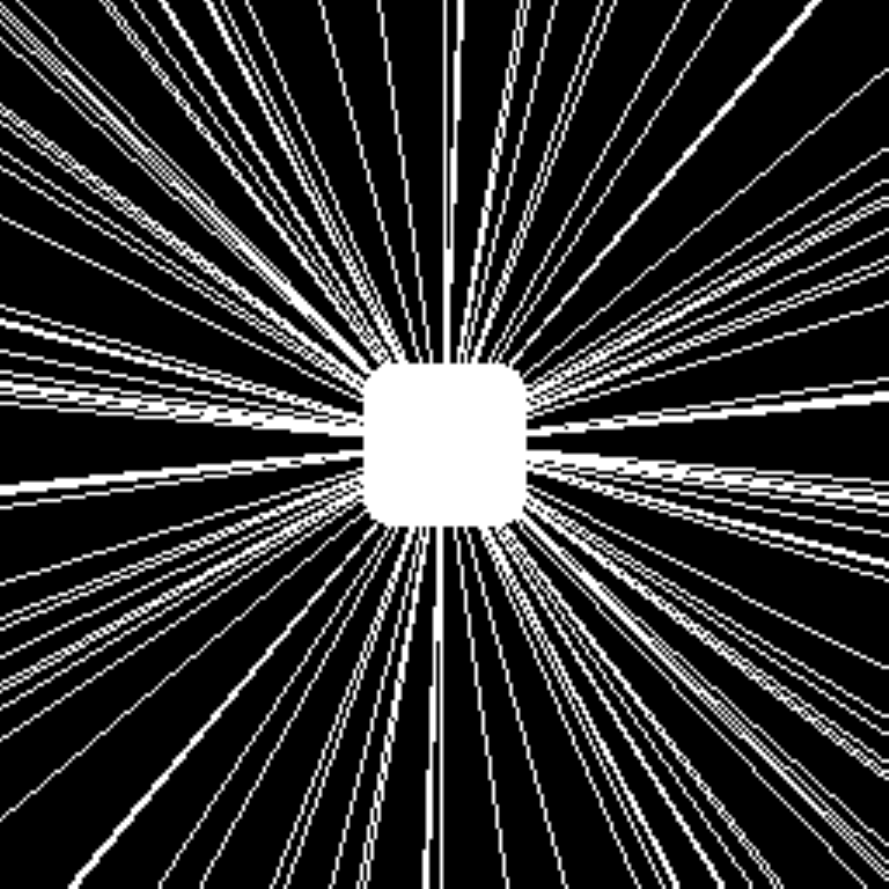}\\
\end{tabular}
\end{center}
\caption{\label{fig:cont2D} 2D continuous sampling schemes based on random walks with $\alpha=.1$~(a), $\alpha=.01$~(b), $\alpha=.001$~(c), and based on TSP solutions with distributions proportional to $\pi$~(d) and to $1/|k|^2$~(e). Classical sampling schemes: spiral~(f), radial~(g) and radial with random angles~(h).}
\end{figure}

\begin{table}[!h]
\caption{\label{tab:cont2D} Quality of reconstruction results in terms of PSNR for continuous sampling trajectories.}
\begin{center}
\begin{tabular}{l|c|c|c||c|c||c|c|c|}
\cline{2-9}
 & \multicolumn{3}{|c||}{Markovian drawing ($\alpha$)} & \multicolumn{2}{|c||}{TSP sampling}  & \multirow{2}*{spiral} & \multirow{2}*{radial} & \multirow{1}*{radial} \\
\cline{2-6}
 &$0.1 $&$0.01  $&$0.001$& $\propto\pi$ &$\mathbf{\propto 1/|\boldsymbol k|^2}$& & & \multirow{1}*{random} \\\hline
\multicolumn{1}{|l|}{mean PSNR} & 35.7 & 34.6 & 33.5 & 35.6 & \textbf{36.1} & 35.6 & 34.1 & 33.1 \\
\multicolumn{1}{|l|}{std dev.}  & 0.1  & 0.3  & 0.6  & 0.1  & 0.1  &      &      &  0.4   \\
\multicolumn{1}{|l|}{max value} & 36.0 & 35.1 & 34.8 & 35.9 & 36.2 &      &      & 34.0   \\\hline
\multicolumn{1}{|l|}{\ADDEDTHREE{in Fig.~\ref{fig:cont2D}: }}      &  (a) &  (b) &  (c) &  (d) &  (e) &  (f) &  (g) &  (h) \\
\hline
\end{tabular}
\end{center}
\end{table}

As predicted by the theory, the shorter the chains the better the reconstructions. The optimal case corresponds to chains of length 1 ($\alpha=1$) i.e. corresponding to independent VDS. When the chain is too long, large $k$-space areas are left unexplored, and the reconstruction quality decreases. 

Besides, the use of a target distribution proportional to $ 1/|k|^2$ instead of $\pi$ for TSP-based schemes provides slightly better reconstruction results. 

We also considered more classical sampling scheme. \DELETED{First, we designed a general spiral trajectory parametrized by $s: \R_+^* \rightarrow \R^2$, $\theta  \mapsto r(\theta)(\cos \theta; \sin \theta)$, which deviates from the classical case (e.g., Archimedes' spiral where $r(\theta)=a+b\theta$). We notice that the time spent by the spiral in an infinitesimal ring of radius $|k|$ is proportional to $r^{-1}(|k|)$. 
In our experiments, to achieve a radial distribution decaying as $1/|k|^2$, we thus set $r(\theta) \propto 1/\sqrt{\theta}$.
Second, w}We observe that the spiral scheme and the proposed ones provide more accurate reconstruction results than radial schemes. 
We believe that the main reason underlying these different behaviors is closely related to the sampling rate decay from low to high frequencies\ADDED{, which is proportional to $1/|k|$ for radial schemes}.

\subsection{3D-MRI}

Since VDS\DELETEDTHREE{ processes} based on Markov chains have shown rather poor reconstruction results compared to the TSP-based sampling schemes in 2D simulations, we only focus on comparing TSP-based sampling schemes to classical CS sampling schemes. \ADDEDTHREE{Moreover, the computational load to treat 3D images being much higher than in 2D, we only perform one drawing per sampling scheme in the following experiments. Experiments in 2D suggest that the reconstruction quality is not really impacted by the realization of a particular sampling scheme, except for drawing with Markov chains or with radial with random angles, which are not considered in our 3D experiments.}

\subsubsection{Variable density sampling using independent drawings}
The first step of the TSP-based approach is to identify a relevant target distribution. For doing so, we consider independent drawings as already done in 2D. \DELETEDTHREE{Since the computational load to treat 3D images is much higher than in 2D, we only perform one drawing per sampling scheme for different decays of the density in the Fourier domain. }The results are summarized in Tab.~\ref{tab:res_indep3D}. In this experiment, we still use a number of measurements equal to $20\%$ of the total amount ($R=5$).

\begin{table}[!h]
\begin{center}
\caption{\label{tab:res_indep3D} Quality of reconstruction results in terms of PSNR for sampling schemes based on 3D variable density independent drawings, with densities $\propto 1/k^d$ and $\pi$, and with $20\%$ of measured samples.}
\begin{tabular}{|c|c|c|c|c||c|}
\hline
$d$ & 1 & 2 & 3 & 4 & $\pi$ \\
\hline
PSNR (dB) & 44.78 & \textbf{45.01} & 44.56 & 44.03 & 42.94 \\ \hline
\end{tabular}
\end{center}
\end{table}

\REPLACED{The best reconstruction result is achieved with $d=2$. The best heuristical distributions radially decay from low to high frequencies. Nevertheless, because of the tensor product structure of $\Ab$ matrix, it seems unlikely that distributions with radial decays are optimal in theory. A good approximation of $\pi$ with radial decay has been proposed in~\cite{Krahmer12} for 2D Haar wavelets, but it seems difficult to extend their result to trivariate Haar wavelets. We believe that the reason for the good practical behavior of radial distributions is that no assumptions on the sparsity structure is made to derive the optimal one in the literature. Wavelet coefficients tend to become sparser as the resolution levels increase, and this feature should be accounted for to derive optimal sampling densities for natural images (see Section~\ref{sec:discussion}).}
{The best reconstruction result is achieved with $d=2$ and not the theoretically optimal distribution $\pi$. This illustrates the importance of departing from the sole sparsity hypothesis under which we constructed $\pi$. Natural signals have a much richer structure. For instance wavelet coefficients tend to become sparser as the resolution levels increase, and this feature should be accounted for to derive optimal sampling densities for natural images (see Section~\ref{sec:discussion}.)}

\subsubsection{Efficiency of the TSP sampling based strategy}
Let us now compare the reconstruction results using the TSP based method and the method proposed in the original CS-MRI paper \cite{Lustig07}. These two sampling strategies are depicted in Fig.~\ref{fig:3dkspace}. For 2D independent drawings, we used the distribution providing the best reconstruction results in 2D, \textit{i.e.} proportional to $1/|k|^2$. The TSP-based schemes were designed by drawing city locations independently with respect to a distribution proportional to $p^{\frac{3}{2}}$.  According to Theorem~\ref{thm:convergence_proba} this is the correct way to reach distribution $p$ after joining the cities with constant speed along the TSP solution path. The experiments were performed with $p=\pi$ (see Fig.~\ref{fig:pi_opt}~(b)), and $p\propto 1/|k|^2$, since the latter yielded the best reconstruction results in the 3D independent VDS framework. 
We also compared these two continous schemes to 3D independent drawings with respect to a distribution proportional to $1/|k|^2$. 

\begin{figure}[!h]
\begin{center}
\begin{tabular}{cc}
(a)&(b)\\[-.02\linewidth]
\includegraphics[width=.48 \linewidth]{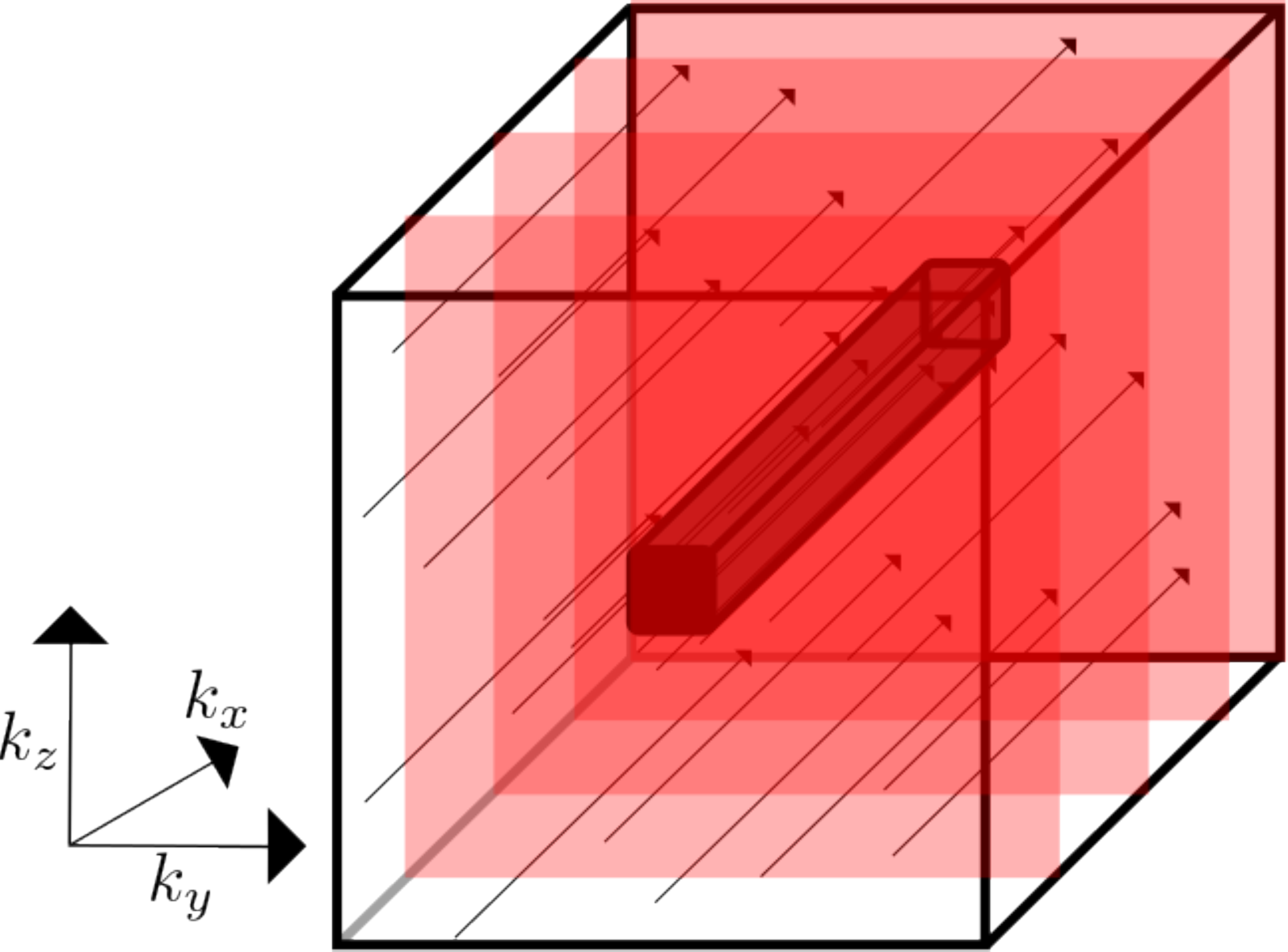}&
\includegraphics[width=.5 \linewidth]{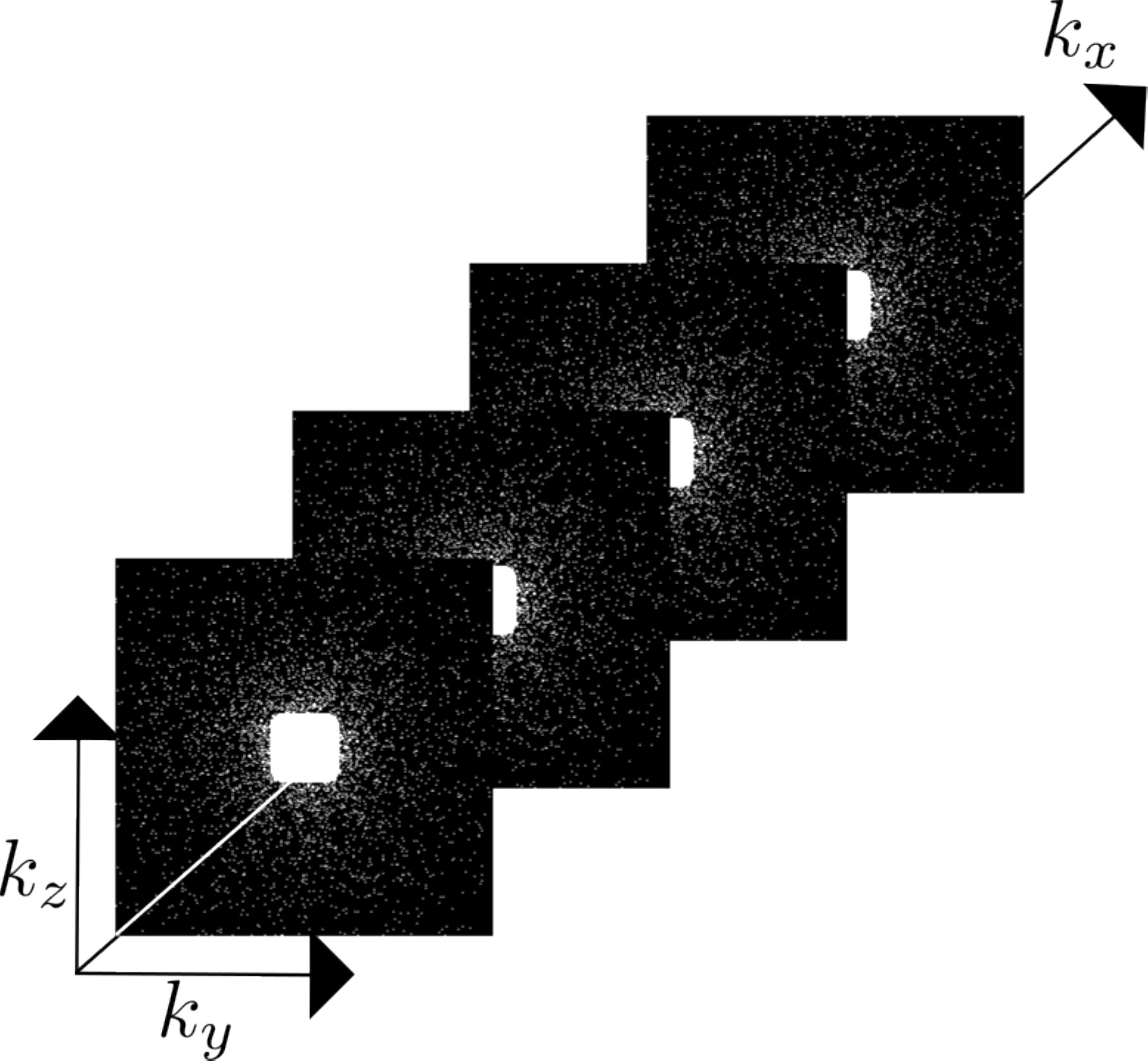}\\[.02 \linewidth]
(c)&(d)\\[-.05\linewidth]
\includegraphics[width=.48 \linewidth]{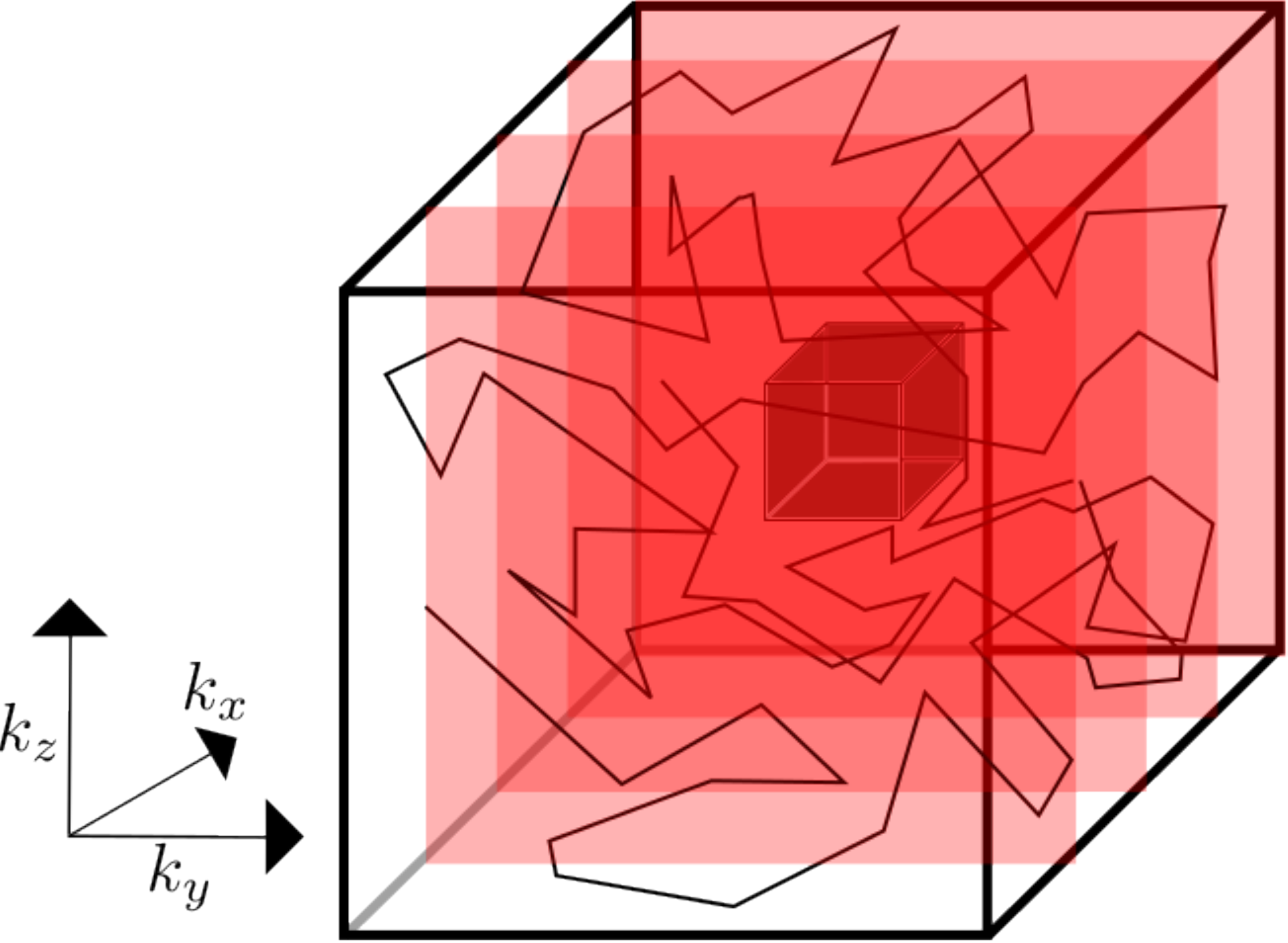}&
\includegraphics[width=.5 \linewidth]{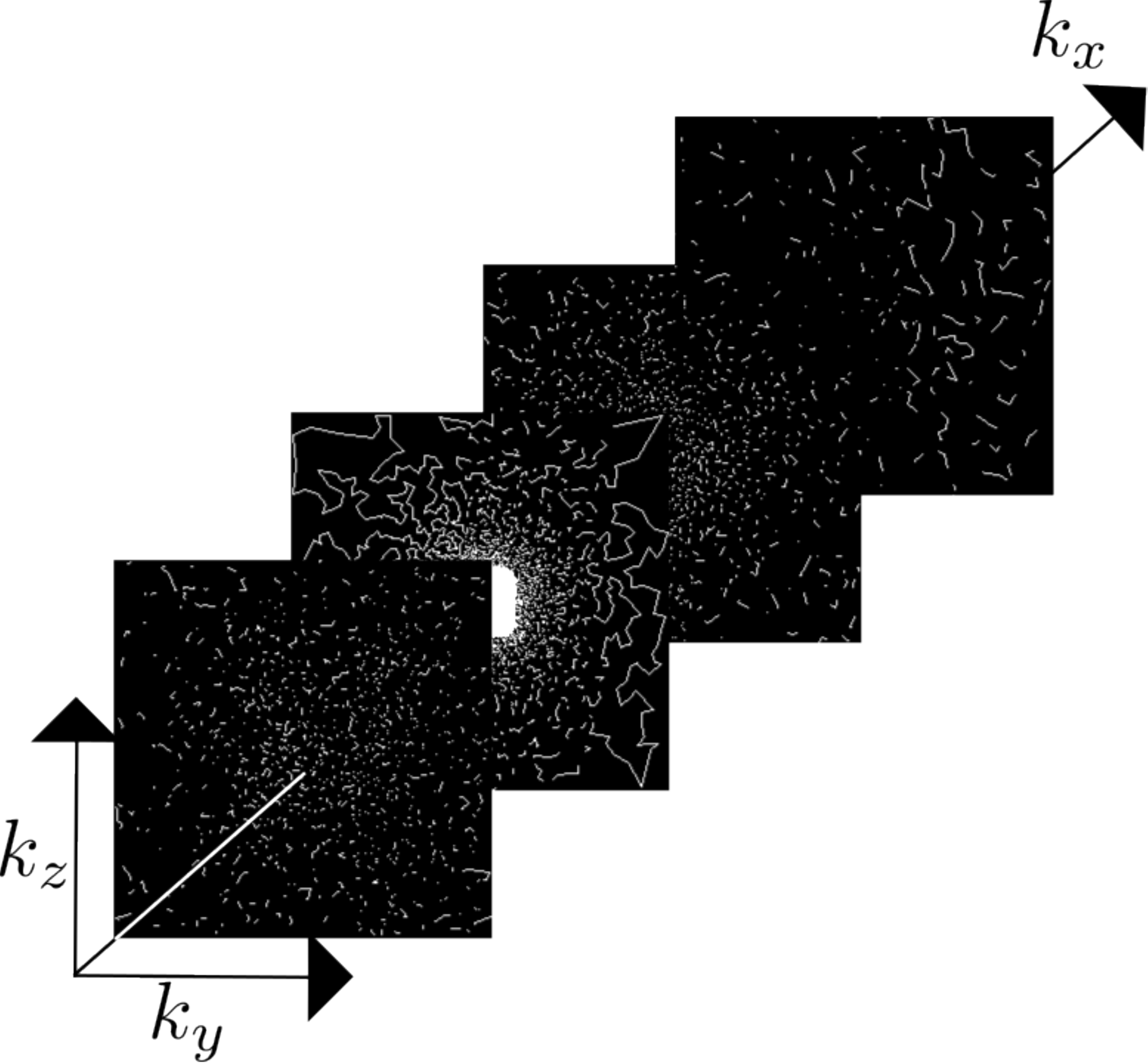}\\
\end{tabular}
\caption{\label{fig:3dkspace} Compared sampling strategies in 3D-MRI. \textbf{Top:} 2D independent drawing sampling schemes designed by a planar independent drawing and measurements in the orthogonal readout direction. \textbf{Bottom:} 3D TSP-based sampling scheme. \textbf{Left:} Schematic representation of the 3D sampling scheme. \textbf{Right:} Representations of 4 parallel slices.}
\end{center}
\end{figure}

\begin{figure}[!h]
\begin{center}
\rotatebox{90}{\hspace{2.6cm} PSNR (dB)}  \includegraphics[width=.8\linewidth]{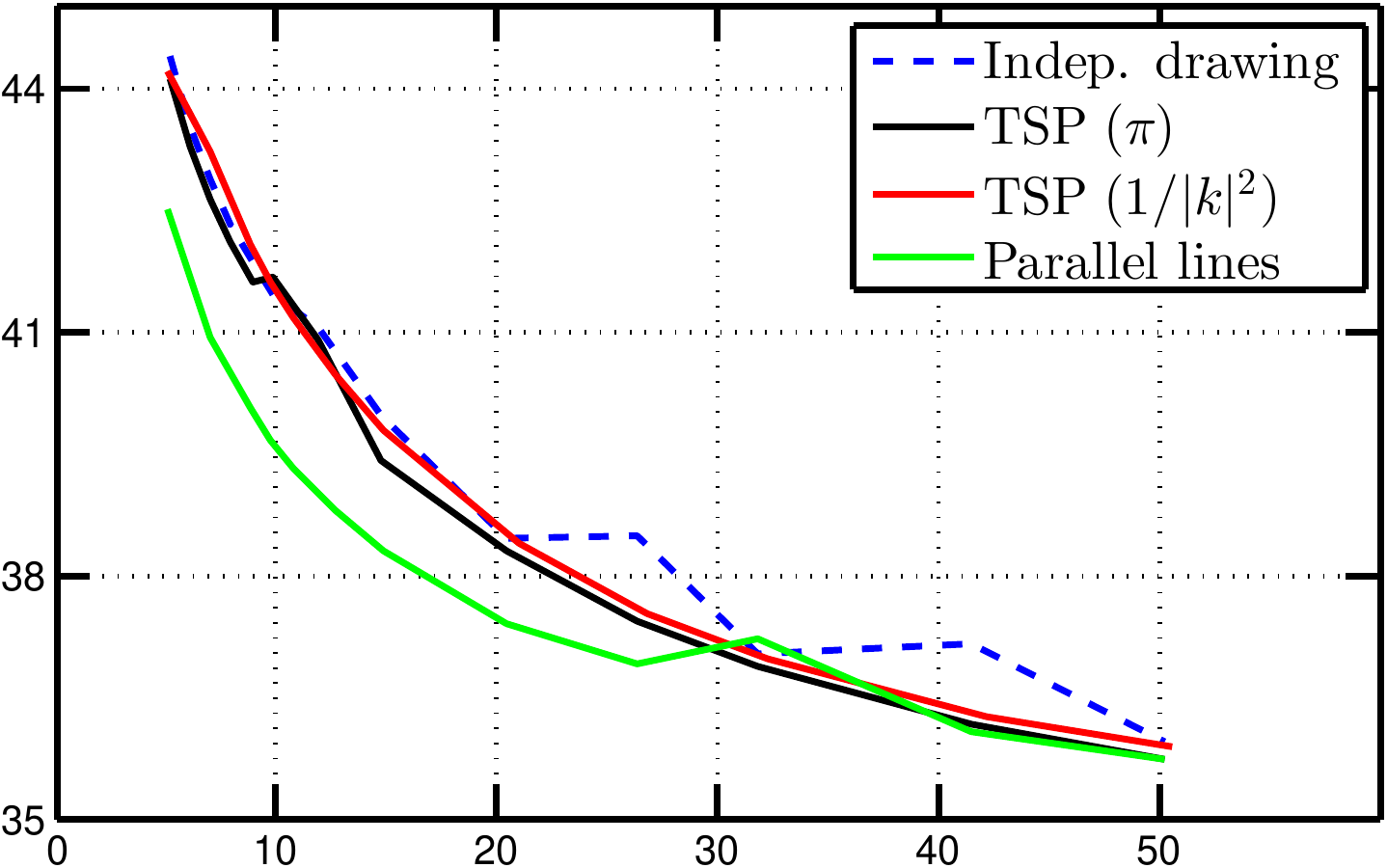}\\
$R$
\caption{\label{fig:res_3D} Quality of 3D reconstructed images in terms of PSNR as a function of \REPLACED{acceleration}{sampling} rates $R$ for various sampling strategies: independent drawings with respect to distribution $\propto 1/|k|^2$ (dashed blue line), TSP-based sampling with target densities $\pi$ (black line) and $\propto 1/|k|^2$ (red line), and parallel lines with 2D independent drawing with respect to $\propto 1/|k|^2$ distribution (green line) as depicted in Fig.~\ref{fig:3dkspace}[Top row].}
\end{center}
\end{figure}

\begin{figure}[!h]
\begin{center}
\begin{tabular}{c}
\hspace{.17\linewidth} (a)  \hfill (b)\hspace{.18\linewidth}  \\
\includegraphics[height=.4\linewidth]{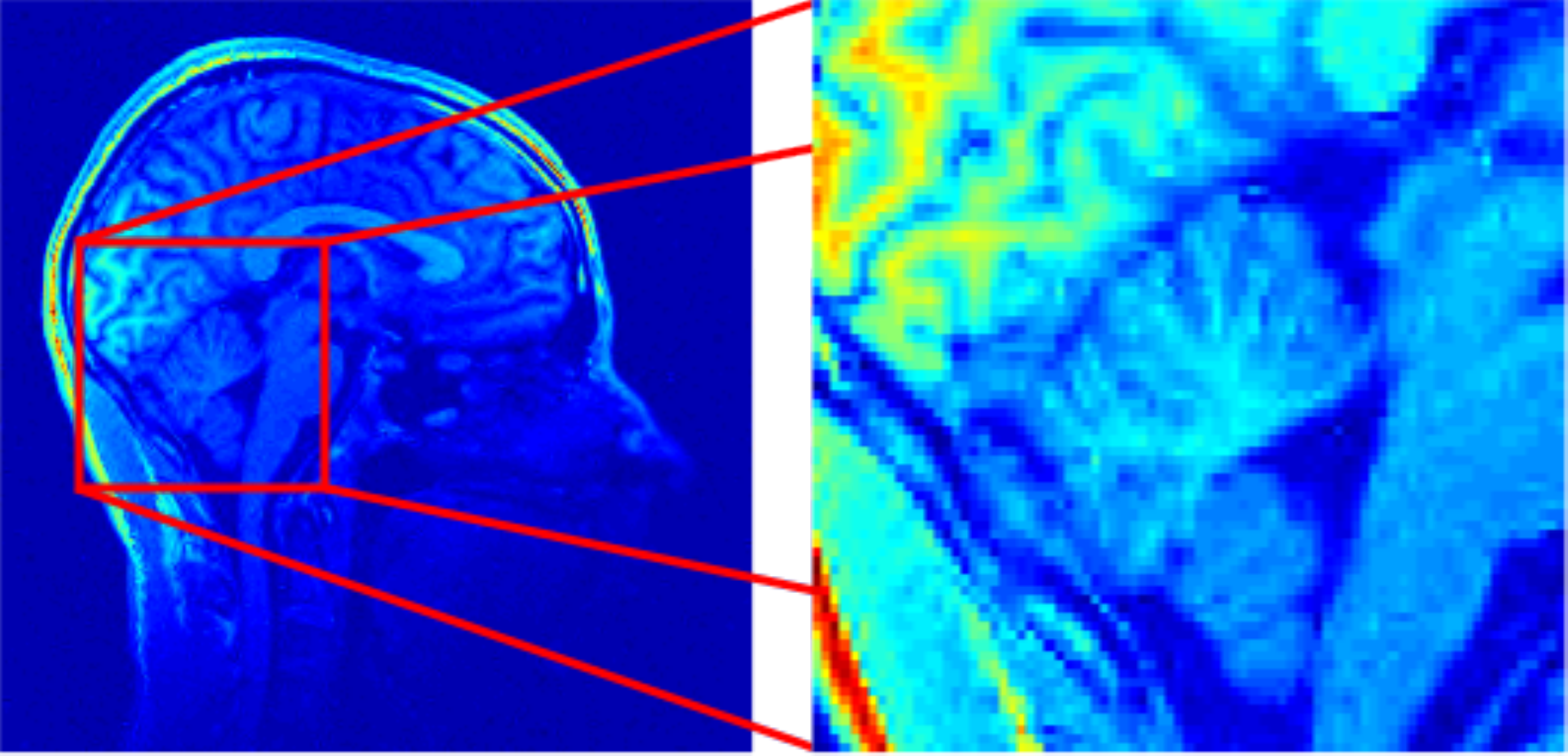}
\\
\hspace{.17\linewidth} (c)  \hfill (d)\hspace{.18\linewidth}  \\
\includegraphics[height=.4\linewidth]{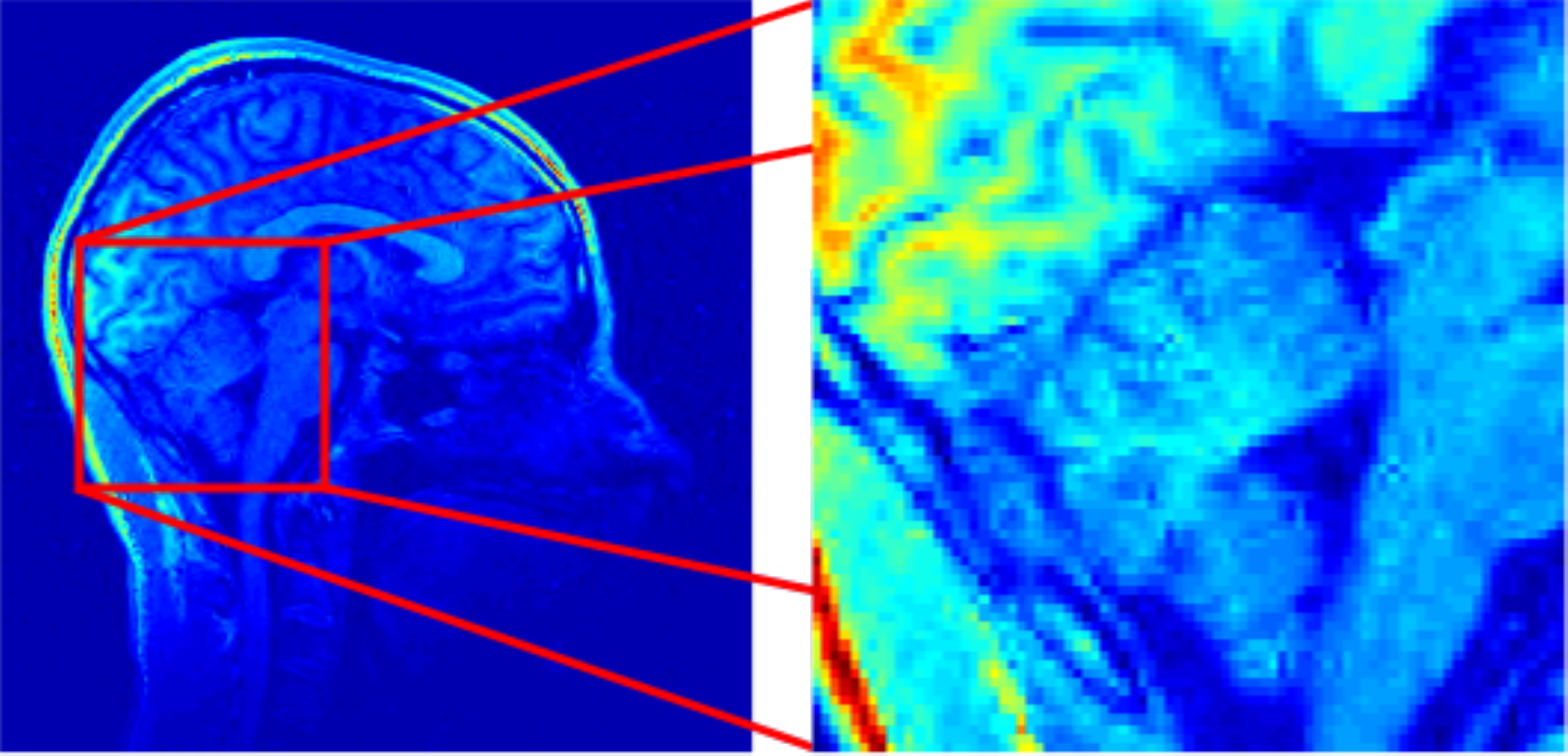}
\end{tabular}
\caption{\label{fig:res_R9} Reconstruction results for $R= 8.8$ for various sampling strategies. \textbf{Top row:} TSP-based sampling schemes (PSNR=$42.1$ dB). \textbf{Bottom row:}  2D random drawing and acquisitions along parallel lines (PSNR=$40.1$ dB). Sagital view (\textbf{left}) and zoom on the cerebellum (\textbf{right}).}
\end{center}
\end{figure}

\begin{figure}[!h]
\begin{center}
\begin{tabular}{cc}
(a) & (b) \\
\includegraphics[width=.4\linewidth]{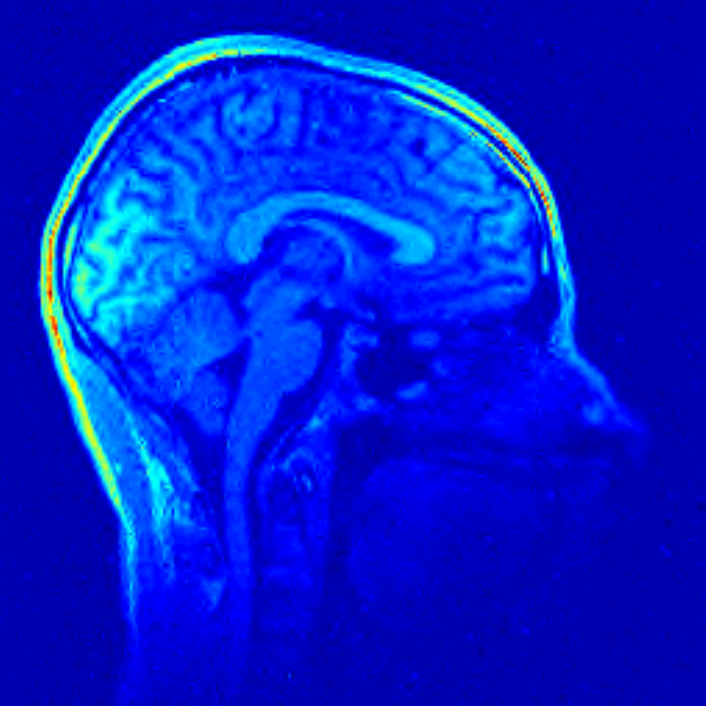}&
\includegraphics[width=.4\linewidth]{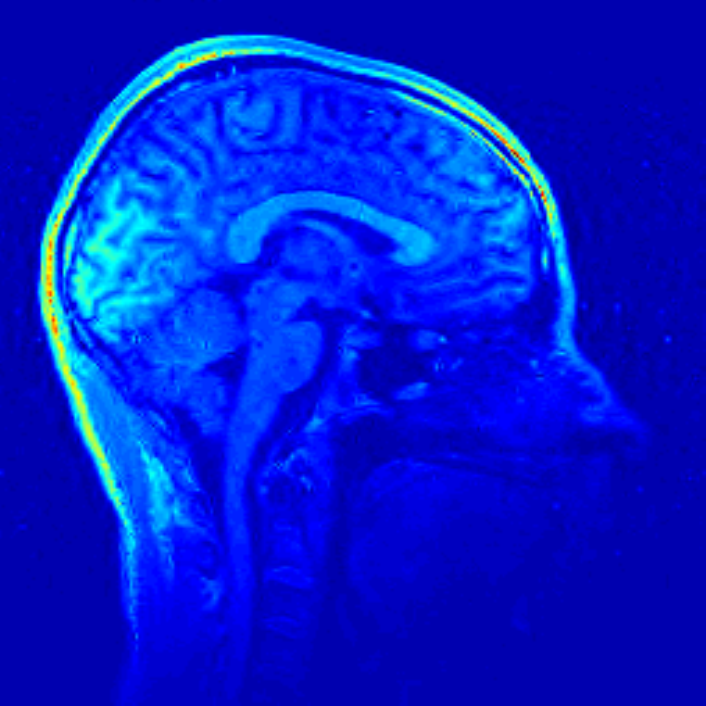}\\
\end{tabular}
\caption{\label{fig:res_R15} Reconstruction results for $R= 14.9$ for various sampling strategies. \textbf{Left:} TSP-based sampling schemes (PSNR=$39.8$ dB). \textbf{Right:}  2D random drawing and acquisitions along parallel lines (PSNR=$38.3$ dB).}
\end{center}
\end{figure}

Reconstruction results with an \REPLACED{acceleration factor}{sampling rate} $R=8.8$ are presented in Fig.~\ref{fig:res_R9}, with a zoom on the cerebellum. The reconstruction quality using the proposed sampling scheme is better than the one obtained from classical CS acquisition and contains less artifacts. In particular, the branches of the cerebellum are observable with our proposed sampling scheme only. At higher \REPLACED{acceleration}{sampling} rate, we still observe less artifacts with the proposed schemes, as depicted in Fig.~\ref{fig:res_R15} with \REPLACED{an acceleration factor}{a sampling rate} $R=14.9$. Moreover, Fig.~\ref{fig:res_3D} shows that our proposed method outperforms the method proposed in \cite{Lustig07} by up to 2dB. If one aims at reaching a fixed PSNR, we can increase $r$ by more than $50\%$ using the TSP based strategy. In other words, we could expect a substantial decrease of scanning time by using more advanced sampling strategies than those proposed until to now.

The two different choices of the target density $\pi$ and $\propto 1/|k|^2$ provide similar results. This is a bit surprising since 3D independent VDS with these two probability distributions provide very different reconstruction results (see Tab.~\ref{tab:res_indep3D}). A potential explanation for that behavior is that the TSP tends to ``smooth out'' the target distribution. An independent drawing would collect very few Fourier coefficients in the blue zones of Fig.~\ref{fig:pi_opt}, notably the vertical and horizontal lines crossing the Fourier plane center. Sampling these zones seems to be of utmost importance since they contain high energy coefficients. The TSP approach tends to sample these zones by crossing the lines.

Perhaps the most interesting fact is that Fig.~\ref{fig:res_3D} shows that the TSP based sampling schemes provide results that are similar to independent drawings up to important \REPLACED{acceleration}{sampling} rates such as $20$. We thus believe that the TSP solution proposed in this paper is near optimal since it provides results similar to unconstrained acquisition schemes. The price to be paid by integrating continuity constraints is thus almost null. 

\section{Discussion and perspectives}
\label{sec:discussion}

In this paper, we investigated the use of variable density sampling along continuous trajectories. 
Our first contribution was to provide a well-grounded mathematical definition of $p$-variable density samplers~(VDS) as stochastic processes with a prescribed limit empirical measure $p$. We identified through both theoretical and experimental results two key features characterizing their efficiency: their \textbf{empirical measure} as well as their  \textbf{mixing properties}. We showed that VDS based on random walks were doomed to fail since they were unable to quickly cover the whole acquisition space. This led us to propose a two-step alternative that consists first of drawing random points independently and then joining them using a Travelling Salesman Problem solver. In contrast to what has been proposed in the literature so far, we paid attention to the manner the points have to be drawn so as to reach a prescribed empirical measure. Strikingly, our numerical results \REPLACEDTwo{demonstrate}{suggest} that the proposed approach yields reconstruction results that are nearly equivalent to independent drawings. This suggests that adding continuity constraints to the sampling schemes might not be so harmful to derive CS results.

We believe that the proposed work opens many perspectives as outlined in what follows.

\textbf{How to select the target density?}
We recalled existing theoretical results to address this point in Section~\ref{sec:VDS} and showed that \textit{deterministic} sampling could reduce the total number of required measurements. The analysis we performed closely followed the proofs proposed in~\cite{Rauhut10,Candes11} and was based solely on sparsity hypotheses on the signal/image to be reconstructed. 
The numerical experiments we performed indicate that heuristic densities still outperform the theoretical optimal ones. This suggests that the optimality critera used so far to derive target sampling densities does not account for the whole structure of the sought signal/image.
Although sparsity is a key feature that characterizes natural signals/images, we believe that introducing stronger knowledge like \emph{structured sparsity} might contribute to derive a new class of optimal densities that would compete with heuristic densities.

To the best of our knowledge, the recent paper~\cite{Adcock13} is the first contribution that addresses the design of sampling schemes by accounting for a simple structured sparsity hypothesis. The latter assumes that wavelet coefficients become sparser as the resolution increases. The main conclusion of the authors is the same as that of Theorem~\ref{thm:rauhut2} even though it is based on different arguments: the low frequencies of a signal should be sampled deterministically.

Finally, let us notice that the best empirical convex reconstruction techniques do not rely on the resolution of a simple $\ell^1$ problem such as \eqref{eq:minl1}. They are based on regularization with redundant frames and total variation for instance \cite{Boyer12}. The signal model, the target density and the reconstruction algorithm should clearly be considered simultaneously to make a substantial leap on reconstruction guarantees. 

\textbf{What VDS properties govern their practical efficiency?}
In Section~\ref{sec:cont}, it was shown that the key feature characterizing random walks efficiency was the mixing properties of the associated stochastic transition matrix.
In order to derive CS results using generic random sets rather than point processes or random walks, it seems important to us to find an equivalent notion of mixing properties.

\textbf{How to generate VDS with higher degrees of regularity?}
This is probably the most important question from a practical point of view. We showed that the TSP based VDS outperformed more conventional sampling strategies by substantial acceleration factors for a given PSNR value or recovers 3D images at an improved PSNR for a given acceleration factor. However, this approach may not really be appealing for many applications: continuity is actually not a sufficient condition for making acquisition sequences implementable on devices like MRI scanners or robot motion where additional kinematic constraints such as bounded first \ADDEDTHREE{(gradients) }and second \ADDEDTHREE{(slew rate) }derivatives should be taken into account. Papers such as~\cite{Lustig08} derive time-optimal waveforms to cross a given curve using optimal control. 
By using this approach, it can be shown that the angular points on the TSP trajectory have to be visited with a zero speed. This strongly impacts the scanning time and the distribution of the parametrized curve. The simplest strategy to reduce scanning time would thus consist in smoothing the TSP trajectory, however this approach dramatically changes the target distribution which was shown to be a key feature of the method. The key element to prove our TSP Theorem~\ref{thm:convergence_proba} was the famous Beardwood, Halton and Hammersley theorem~\cite{beardwood1959shortest}. To the best of our knowledge, extending this result to smooth trajectories remains an open question\footnote{To be precise, many crucial properties of the length of the shortest path used to derive asymptotic results are lost. The most important one is subadditivity  \cite{steele1981subadditive}.}. Recent progresses in that direction were obtained in recent papers such as \cite{le2012dubins}, but they do not provide sufficient guarantees to extend Theorem~\ref{thm:convergence_proba}. Answering this question is beyond the scope of this paper. We believe that the work~\cite{Teuber11} based on attraction and repulsion potentials opens an appealing research avenue for solving this issue.

\section*{Appendix 1 - proof of Theorem~\ref{thm:rauhut2}}
\ADDED{For a symmetric matrix $M$, we denote by $\lambda_{\max}(M)$ its largest eigenvalue and by $\|M\|$ the largest eigenvalue modulus.} The crucial step to obtain Theorem~\ref{thm:rauhut2} is Proposition \ref{prop:propoconcentr} below.  
The rest of the proof is the same as the one proposed in~\cite{Rauhut10} and we refer the interested reader to \cite[Section 7.3]{Rauhut10} for further details.

\begin{proposition}
 \label{prop:propoconcentr}
Let $\Omega =\Omega_1\cup \Omega_2 \subseteq \{1,\hdots, n\}$ be a set constructed as in Theorem~\ref{thm:rauhut2}.
Define
$$
\tilde a_i =  \left\{\begin{array}{ll}
		    a_i &\textrm{ if } i\in \Omega_1 \\
		    a_i/\sqrt{p_i}  &\textrm{ if }  i\in \{1 \hdots n \} \setminus \Omega_1.
                   \end{array}\right.
$$
and 
\begin{align}
\label{eq::defAtilde}
\tilde \Ab &=\begin{pmatrix}
                  \tilde a_{\Omega_1(1)} \\ 
                  \vdots \\
                  \tilde a_{\Omega_1(m_1)} \\ 
                  \frac{1}{\sqrt{m_2}}\tilde a_{\Omega_2(1)}\\ 
                  \vdots \\ 
                  \frac{1}{\sqrt{m_2}}\tilde a_{\Omega_2(m_2)}
                 \end{pmatrix}\in \C^{m\times n}.
\end{align}                 
Then for all $\delta \in [0,\frac{1}{2}]$:
\begin{align*}
\Prob \Bigl( \Bigl\| \tilde \Ab^{S*}\tilde \Ab^S - I_s \Bigr\| \geqslant \delta \Bigr) \leqslant 2s \exp \left(- \frac{m_2 \delta^2}{C K_2^2 s} \right) .
\end{align*} 
where $\tilde \Ab^S\in \C^{m\times s}$ is the matrix composed of the $s$ columns of $\tilde \Ab$ belonging to $S$. $C=7/3$ is a constant. 
\end{proposition}

The proof of this proposition relies heavily on the matrix Bernstein inequality below \cite{Tropp12}.
\begin{proposition}[Matrix Bernstein inequality]
Let $\Zb_k$ be a finite sequence of independent, random, self-adjoint matrices in $\C^{d\times d}$. Assume that each random matrix satisfies
\begin{align*}
\Exp( \Zb_k)=0 \qquad \textrm{and} \qquad \lambda_{\max} (\Zb_k) \leqslant R \qquad a.s.
\end{align*}
Denote $\displaystyle \sigma^2= \Bigl\|\sum_k \Exp (\Zb_k^2) \Bigr\|$.
Then, for all $t\geqslant 0$:
\begin{align*}
\Prob\Bigl(  \Bigl\| \sum_k \Zb_k \Bigr\| \geqslant t \Bigr) \leqslant 2d \exp \Bigr( -\frac{t^2/2}{\sigma^2+Rt/3} \Bigr).
\end{align*}
\end{proposition}{

We are now ready to prove Proposition~\ref{prop:propoconcentr}.
\begin{proof}
For any vector $v \in \C^n$, denote by $v^S \in \C^s$ the vector composed of the entries of $v$ belonging to $S\subseteq \{1,\hdots,n\}$. 
Consider the random sequence $X_1 , \hdots, X_{m_2}$ where $X_i=j \in \{1 \hdots n \} \setminus \Omega_1$ with probability $p_j$, and denote by $\Omega_2$ the set $\{X_1, \hdots X_{m_2}\}$.
Denote by $\Mb_1:=\sum_{i\in \Omega_1} a_i^S {a_i^S}^*$.
Consider the matrices $\Zb_j:=\Mb_1 + \tilde a_j^S  \tilde a_j^{S*}-I_s$. According to Eq.~\eqref{eq::defAtilde}, we get by construction: 
\begin{align*}
\tilde \Ab^{S*}\tilde \Ab^S - I_s= \frac{1}{m_2}\sum_{j\in \Omega_2} \Zb_j.
\end{align*}
Since $I_s=\sum_{i=1}^n a_i^S {a_i^S}^*$, we notice that $\forall i\in\{1,\hdots, m_2\}$ (i) $\Exp (\Zb_{X_i})=0$, (ii) $\Exp(\tilde a_{X_i}^S \tilde a_{X_i}^{S*})=I_s-\Mb_1$. Moreover, we have (iii) $0 \preceq I_s-\Mb_1 \preceq I_s$ and (iv) $0 \preceq \Mb_1 \preceq I_s $. 

Using the identity $({\tilde a_j}^S {\tilde a_j}^{S*})^2=\|\tilde a_j^S\|^2{\tilde a_j}^S {\tilde a_j}^{S*}$ and the fact that $\|\tilde a_i^S\| \leqslant \sqrt{s} \|\tilde a_i^S\|_\infty $, we get $\Exp((\tilde a_{X_i}^S \tilde a_{X_i}^{S*})^2)\preceq K_2^2 s (I_s-\Mb_1)$ using (ii). 
We can then proceed as follows using points (iii) to (iv):
\begin{eqnarray*}
\Exp(\Zb_{X_i}^2)&=& \Mb_1^2 -2\Mb_1+I_s + \Exp((\tilde a_{X_i}^S \tilde a_{X_i}^{S*})^2)+ 2 \Mb_1 \Exp(\tilde a_{X_i}^S \tilde a_{X_i}^{S*})- 2 \Exp(\tilde a_{X_i}^S \tilde a_{X_i}^{S*})\\
&\leq & \Mb_1^2 -2 \Mb_1+I_s + K_2^2 s (I_s-\Mb_1) + 2\Mb_1 (I_s-\Mb_1) -2 (I_s-\Mb_1) \\
&= & -(I_s-\Mb_1)^2 + K_2^2 s (I_s-\Mb_1) \\
& \preceq &  K_2^2s I_s.
\end{eqnarray*}
Then $\displaystyle \|\sum_{i=1}^{m_2} \Exp(\Zb_{X_i}^2) \| \leqslant  m_2K_2^2s $.\\
By noticing that $ \tilde a_{X_i}^S \tilde a_{X_i}^{S*}- I_s\preceq \Zb_{X_i} \preceq \tilde a_{X_i}^S \tilde a_{X_i}^{S*}$, we obtain $\|\Zb_{X_i}\|\leqslant K_2^2s$.
Finally, by applying Bernstein inequality to the sequence of matrices $\Zb_{X_1} , \hdots \Zb_{X_{m_2}}$, we derive for all $t\geqslant 0$:
\begin{align*}
\Prob \Bigl( \Bigl\|\sum_{j\in \Omega_2} \Zb_j\Bigr\| \geqslant t \Bigr) \leqslant 2s \exp \left(- \frac{ t^2/2}{m_2 K_2^2 s + K_2^2 s t/3} \right) .
\end{align*} 
Plugging $\delta:=t/m_2$, and noticing that $\delta \leqslant 1/2 \Rightarrow 2(1+\delta/3) \leqslant 2(1+\delta/3) \leqslant 7/3$, the announced result is shown. 
\end{proof}

\section*{Appendix 2 - proof of Proposition~\ref{prop:measurements_needed}}
Our approach relies on the following perfect recovery condition introduced in~\cite{Juditsky11}:
\begin{proposition}[\cite{Juditsky11}]
\label{prop:Juditsky11}
If $\Ab_\Omega \in \R^{m\times n}$ satisfies
\begin{equation*}
\gamma(\Ab_\Omega)=\min_{\Yb\in \R^{m\times n}} \|I_n-\Yb^T \Ab_\Omega\|_\infty < \frac{1}{2 s},
\end{equation*}  
all $s$-sparse signals $x \in \R^n$ are recovered exactly by solving the $\ell_1$ minimization problem~\eqref{eq:minl1}.
\end{proposition}

\ADDED{We noted $\|A\|_\infty$ the maximal modulus of all the entries of $A$.}
This can be seen as an alternative to the \textit{mutual coherence}~\cite{Donoho06}. We limit our proof to the real case but it could be extended to the complex case using a slightly different proof.

We aim at finding $\Yb \in \R^{m \times n}$, such that $\|I_n-\Yb^T \Ab_\Omega \|_\infty < \frac{1}{2 s}$, for a given positive integer $s$, where $\Ab_\Omega$ is the sensing matrix defined in Proposition~\ref{prop:measurements_needed}. 
Following~\cite{Juditsky11b}, we set $\Thetab_i=\frac{a_i a_i^T}{p_i}$ and use the decomposition $I_n = \Ab^T\Ab = \sum_{i=1}^{n} p_i \Thetab_i$.
We consider a realization of the Markov chain $X_1 , \dots, X_m$ , with $X_1 \sim p$ and $X_i \sim \Pb_{X_{i-1},:}$ for $i>1$.
Let us denote $\Wb_m = \frac{1}{m} \sum_{l=1}^m \Thetab_{X_l}$. Then $\Wb_m$ may be written as $\Yb^T \Ab_\Omega$.

\begin{lemma}
\label{prop:markov}
$ \forall\  0<t\leqslant 1$,
\begin{equation}
\label{eq:conc_markov}
 \Prob \left(\|I_n - \Wb_m\|_\infty \! \geqslant t \right) \! \leqslant\! n (n+1) e^{\frac{\epsilon(\Pb)}{5}} \! \exp \Bigl(\!- \frac{mt^2 \epsilon(\Pb)}{12 K^2(\Ab,p)}\Bigr).
\end{equation}
\end{lemma}

Before proving the lemma, let first recall a concentration inequality for finite-state Markov chains~\cite{Lezaud98}.
\begin{proposition}
\label{prop:Lezaud}
Let $(\Pb,p)$ be an irreducible and reversible Markov chain on a finite set G of size $n$ with transition matrix $\Pb$ and stationary distribution $p$. Let $f:G \rightarrow \mathbb{R}$ be such that $\sum_{i=1}^n p_i f_i = 0 , \, \|f\|_\infty \leqslant1$ and $0< \sum_{i=1}^n f_i^2 p_i \leqslant b^2$. Then, for any initial distribution $q$, any positive integer $m$ and all $0< t\leqslant1$,
\begin{equation*}
\Prob \Bigl(\frac{1}{m} \sum_{i=1}^m f(X_i) \geqslant t \Bigr) \leqslant e^{\frac{\epsilon(\Pb)}{5}} N_q \exp \Bigl(- \frac{m t^2 \epsilon(\Pb)}{4 b^2(1+g(5 t/b^2))} \Bigr)
\end{equation*}
where $N_q=(\sum_{i=1}^n (\frac{q_i}{p_i})^2 p_i)^{1/2}$\REPLACED{, $\beta_1(\Pb)$ is the second largest eigenvalue of $\Pb$, and $\epsilon(\Pb)=1-\beta_1(\Pb)$ is the spectral gap of the chain. Finally}{ and} $g$ is given by $g(x)=\frac{1}{2}(\sqrt{1+x} - (1-x/2))$.
\end{proposition}

Now, we can prove Lemma~\ref{prop:markov}
\begin{proof}
By applying Proposition~\ref{prop:Lezaud} to a function $f$ and then to its opposite $-f$, we get:

\begin{equation*}
\label{eq:ineq_abs}
\Prob \Bigl(\Bigl|\frac{1}{m} \sum_{i=1}^m f(X_i)\Bigr| \geqslant t \Bigr) \leqslant2 e^{\frac{\epsilon(\Pb)}{5}} N_q  \exp \Bigl(- \frac{m t^2 \epsilon(\Pb)}{4 b^2(1+g(5 t/b^2))} \Bigr).
\end{equation*}

\noindent Then we set $f(X_i)=(I_n-\Thetab_{X_i})^{(a,b)}/K(\Ab,p)$ as real-valued function. 
Recall that $p$ satisfies $\sum_{i=1}^np_i f(X_i)=0$. 
Since $\|f\|_\infty \leqslant1$, $b=1$ and $t\leqslant 1$, we deduce $1+g(5t)<3$.
Moreover, since the initial distribution is $p$, $q_i=p_i, \forall i$ and thus $N_q=1$. Finally, resorting to a union bound enables us to extend our result for the $(a,b)$th entry to the whole infinite norm of the $n \times n$ matrix $I_n-\Wb_m$~\eqref{eq:conc_markov}.
\hfill\end{proof}

\noindent Finally, set $s\in \N^*$ and $\eta \in (0,1)$. If $m$ satisfies Ineq.~\eqref{eq:measurements}, then 
\begin{align*}
\Prob \left(\|I_n - \Wb_m\|_\infty  \geqslant \frac{1}{2s} \right) <\eta \; .
\end{align*}
In other words, with probability at least $1-\eta$, every $s$-sparse signal can be recovered by $\ell_1$ minimization~\eqref{eq:minl1}.

\begin{remark}
It is straightforward to derive a similar result to Theorem~\ref{thm:rauhut2} and thus to justify that a partial deterministic sampling reduces the total number of measurements required for perfect recovery.
\end{remark}

\section*{Appendix 3 - proof of Remark~\ref{rmk:SpectralGap}}
In this part, we prove that for a random walk with uniform stationary distribution $p$, $\displaystyle \epsilon(\Pb)=O(n^{-\frac{1}{d}})$. We use geometric bounds known as \emph{Cheeger's inequality} in~\cite{diaconis1991geometric} and \emph{Conductance Bounds} in~\cite{Jerrum89,Bremaud99}. Let us recall a useful result concerning finite state space irreducible reversible transition matrices $\Pb$.

The \emph{capacity} of a set $B\subset \{1,\hdots,n \}$ is defined as $ p(B):=\sum_{i \in B}p(i)$ and the \emph{ergodic flow} out of $B$ is defined by $ F(B):=\sum_{i \in B, j\in B^c }p(i)\Pb_{i,j}$. The \emph{conductance} of the pair $(\Pb,p)$ is:

\begin{align*}
\varphi(\Pb):=\inf_B \left( \frac{F(B)}{p(B)};\;\; 0 <|B|<n,\;\; p(B)\leqslant \frac{1}{2}  \right).
\end{align*}

Then the following result holds (see~\cite{Jerrum89} and~\cite[Theorem~4.3]{Bremaud99}):
\begin{proposition}
\begin{align*}
\frac{\varphi(\Pb)^2}{2} \leqslant \epsilon(\Pb) \leqslant 2\varphi(\Pb).
\end{align*}
\end{proposition}

Now, assume that $n^{1/d}\in \mathbb{N}$ is even and construct a finite graph with $n$ nodes representing a Euclidean grid of the unit hypercube of dimension $d$. Assume that the vertices of the graph at one grid point are the $2d$ nearest nodes, with periodic boundary conditions (the graph can be seen as a $d$-dimensional torus). Assume that the transition probability is uniform over the neighbors, thus the stationary distribution is also the uniform one. This graph is depicted in Fig.~\ref{fig:proofSpectralGap}[Left], with $d=2$.

Let $B$ be the halved graph defined by the hyperplane parallel to an axis of the grid and including its center, so that $p(B)=1/2$. An illustration is given in 2D in Fig.~\ref{fig:proofSpectralGap}[Right]. Since we assumed periodic boundary conditions, the number of nodes belonging to $B$ and having a neighbor in $B^c$ is $2n^{(d-1)/d}$. Each of these nodes have $2d$ neighbors, but only one belonging to $B^c$. Since the stationary distribution is equal to $1/n$ on each node, the ergodic flow is $2n^{(d-1)/d}(\frac{1}{n} \frac{1}{2d})$. It follows that $\epsilon(\Pb) \leqslant \frac{4}{d}n^{-\frac{1}{d}}$.

\begin{figure}[h!]
\begin{center}
\includegraphics[height=.35\linewidth]{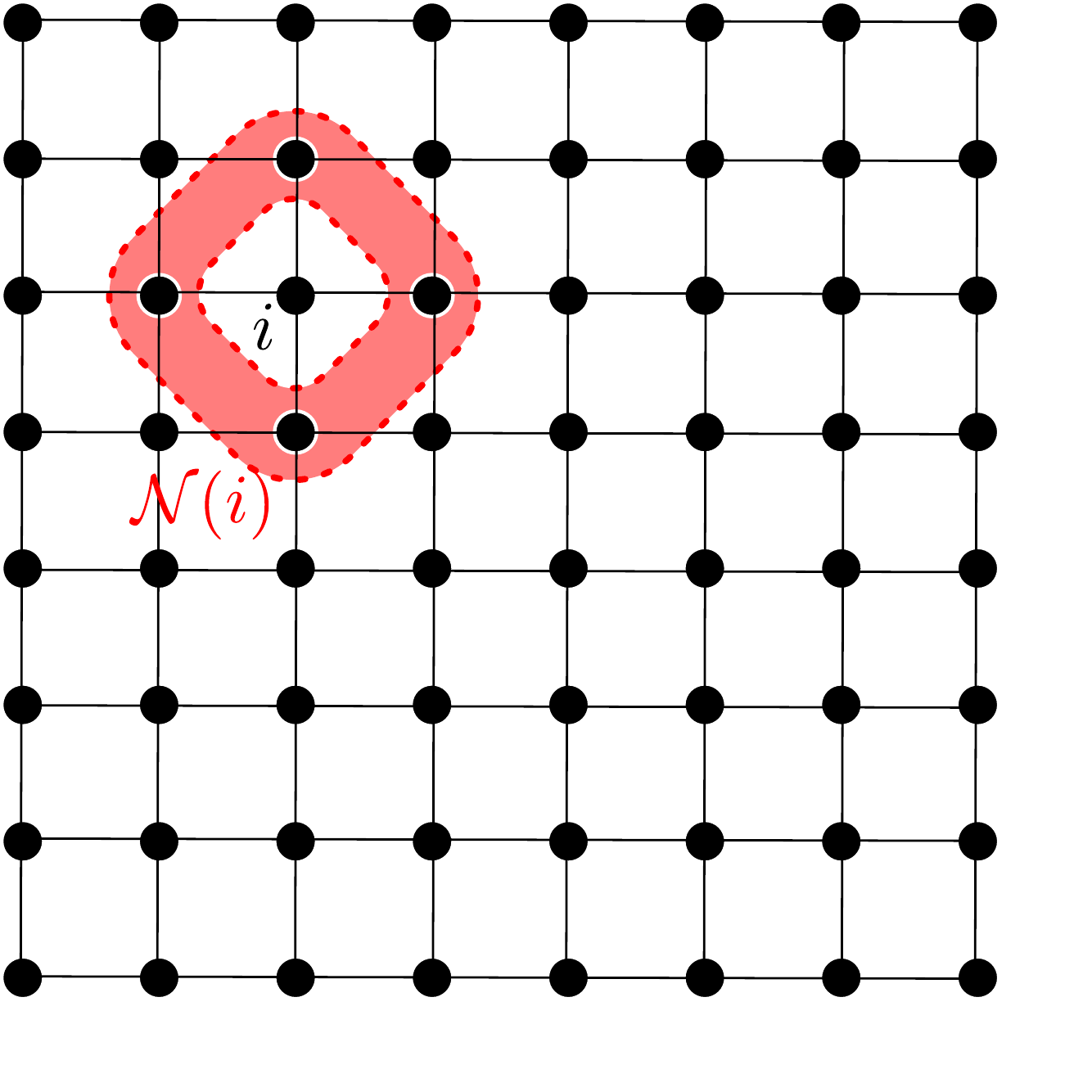}
\hspace{.04\linewidth}
\includegraphics[height=.35\linewidth]{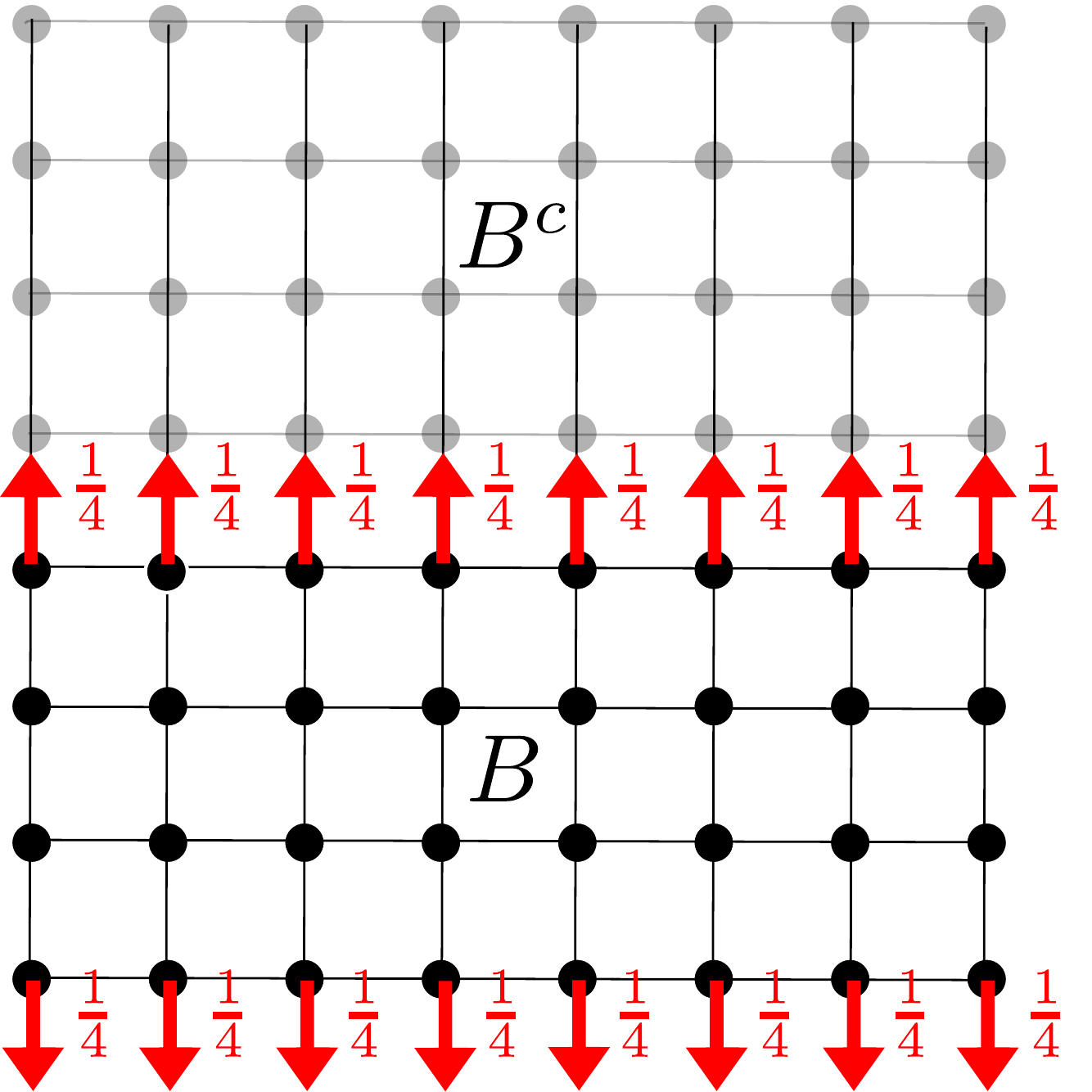}
\caption{\label{fig:proofSpectralGap} Illustration of the proof of Remark~\ref{rmk:SpectralGap} in dimension 2. \textbf{Left:} regular grid with $n=\sqrt{n}\times \sqrt{n}$ nodes. \textbf{Right:} Graph partitioning in $B$ and $B^c$ with $p(B)=1/2$.}
\end{center}
\end{figure}

\section*{Appendix 4 - proof of Theorem~\ref{thm:convergence_proba}}
\ADDED{Let $h\in \N$. The set $\setOmega=[0,1]^d$ will be partitioned in $h^d$ congruent hypercubes $(\omega_i)_{i\in I}$ of edge length $1/h$. }The following proposition is central to obtain the proof:

\begin{proposition}
    \label{cube_by_cube}
    Almost surely, for all $\omega_i$ in $\{\omega_i\}_{1\leq i \leq h^d}$:    
    \begin{align}
        \label{for_small}
        \lim_{N\to \infty} \tilde{P}_N(\omega _i) & = \tilde{p}(\omega _i) \\
                                                   & = \frac{\int_{\omega _i} p^{(d-1)/d}(x) \mathrm{d}x}{\int_{\setOmega} p^{(d-1)/d}(x) \mathrm{d}x} & \mbox{$p^{\otimes \mathbb{N}}$-a.s.}
    \end{align}
   
\end{proposition}

The strategy consists in proving that $T_{|\omega _i}(X_N, \setOmega )$ tends asymptotically to $T(X_N, \omega _i)$. The estimation of each term can then be obtained by applying the asymptotic result of Beardwood, Halton and Hammersley~\cite{beardwood1959shortest,steele1981subadditive}:\\

\begin{theorem}
    \label{BHH}
    If $R$ is a Lebesgue-measurable set in $\mathbb{R}^d$ such that the boundary $\partial R$ has zero measure, and $\{y_i\}_{i\in \mathbb{N}^* }$, with $Y_N = \left\{ y_i \right\} _{i\leqslant N}$ is a sequence of i.i.d. points from a density $p$ supported on $R$, then, almost surely,
    \begin{align}
        \label{BHHeq}
        \lim_{N \to \infty} \frac{T(Y_N, R) }{N^{(d-1)/d}} & = \beta (d) \int_{R}  p^{(d-1)/d}(x) \mathrm{d}x,
    \end{align}
    where $\beta (d)$ depends on the dimension $d$ only.\\
\end{theorem}

To show Prop.\ref{cube_by_cube}, we need to introduce the boundary TSP. For a set of points $F$ and an area $R$, we denote by $T_B(F, R)$ its length on the set $F \cap R$. 
The boundary TSP is defined as the shortest Hamiltonian tour on $F \cap R$ for the metric obtained from the Euclidean metric by the quotient of the boundary of $R$, that is $d(a,b) = 0$ if $a, b \in \partial R$. Informally, it matches the original TSP while being allowed to travel along the boundary for free. We refer to~\cite{yukich_gutin2002traveling} for a complete description of this concept.\\

We shall use a set of classical results on TSP and boundary TSP, that may be found in the survey books~\cite{yukich_gutin2002traveling} and~\cite{yukich98}. \\

\begin{lemma}
Let $F$ denote a set of $n$ points in $\setOmega $.

\begin{enumerate}
 \item The boundary TSP is superadditive, that is, if $R_1$ and $R_2$ have disjoint interiors.
\begin{align}
    \label{superadditive}
    T_B(F, R_1 \cup R_2) & \geqslant T_B(F, R_1) + T_B(F, R_2).
\end{align}
 \item The boundary TSP is a lower bound on the TSP, both globally and on subsets. If $R_2 \subset R_1$:
\begin{align}
    \label{boundT}
    T(F, R) & \geqslant T_B(F, R)    \\
    \label{boundlocal}    
    T_{|R_2}(F, R_1) & \geqslant T_B(F, R_2)
\end{align}
  \item The boundary TSP approximates well the TSP~\cite[Lemma $3.7$]{yukich98}):

\begin{align}
    \label{bound_approx}
    |T(F,\setOmega) - T_B(F,\setOmega)| =  O(n^{(d-2)/(d-1)} ).
\end{align}
  \item The TSP in $\setOmega$ is well-approximated by the sum of TSPs in a grid of $h^d$ congruent hypercubes~\cite[Eq.~(33)]{yukich_gutin2002traveling}.
\begin{align}
    \label{approx}
    \lvert T(F, \setOmega ) - \sum_{i=1}^{h^d} T(F, \omega _i) \rvert = O(n^{(d-2)/(d-1)}).
\end{align}
\end{enumerate}
\end{lemma}
We now have all the ingredients to prove the main results.\\

\begin{proof}[Proof of Proposition \ref{cube_by_cube}]

\begin{align*}
    \sum_{i\in I} T_B(X_N,\omega_i) & \stackrel{\eqref{superadditive}}{\leqslant} T_B(X_N,\setOmega)  \\
  & \stackrel{\eqref{boundT}}{\leqslant} T(X_N,\setOmega) = \sum_{i=1}^{h^d} T|_{\omega_i}(X_N,\setOmega) \\
  & \stackrel{\eqref{approx}}{\leqslant} \sum_{i=1}^{h^d} T(X_N,\omega_i) + O(N^{(d-1)/(d-2)})
\end{align*}
Let $N_i$ be the number of points of $X_N$ in $\omega _i$. 

Since $N_i \leqslant N$, we may use the bound \eqref{bound_approx} to get: 
\begin{equation}
\label{eq:9}
 \lim_{N\rightarrow \infty}\frac{T(X_N,\omega_i)}{N^{(d-1)/d}} =  \lim_{N\rightarrow \infty}\frac{T_B(X_N,\omega_i)}{N^{(d-1)/d}}.
\end{equation}
Using the fact that there are only finitely many $\omega _i$, the following equalities hold almost surely:
\begin{align*}
\lim_{N\rightarrow \infty}  \frac{\sum_{i=1}^{h^d}T_B(X_N,\omega_i)}{N^{(d-1)/d}} 
     & = \lim_{N\rightarrow \infty} \frac{\sum_{i=1}^{h^d}T(X_N,\omega_i)}{N^{(d-1)/d}} \\
     & \stackrel{\eqref{approx}}{=} \lim_{N\rightarrow \infty}  \frac{\sum_{i=1}^{h^d}T_{|\omega_i}(X_N,\setOmega )}{N^{(d-1)/d}}.
\end{align*}

Since the boundary TSP is a lower bound~(cf. Eqs.~\eqref{boundlocal}-\eqref{boundT}) to both local and global TSPs, the above equality ensures that:
\begin{align}
    \label{allequal}
    \lim_{N\rightarrow \infty}  \frac{T_B(X_N,\omega_i)}{N^{(d-1)/d}} & = \lim_{N\rightarrow \infty} \frac {T(X_N,\omega_i )}{N^{(d-1)/d}}\\
     & = \lim_{N\rightarrow \infty} \frac {T_{|\omega_i}(X_N,\setOmega )}{N^{(d-1)/d}}   & \mbox{$p^{\otimes \mathbb{N} }$-a.s, $\forall i$}.\nonumber
\end{align}
Finally, by the law of large numbers, almost surely $N_i / N \to p(\omega _i)=\int_{\omega_i} p(x)dx$. 
The law of any point $x_j$ conditioned on being in $\omega _i$ has density $p / p(\omega_i)$. By applying Theorem \ref{BHH} to the hypercubes $\omega _i$ and $\setOmega$ we thus get:
\begin{align*}
    \lim_{N\rightarrow +\infty} \frac{T(X_N,\omega_i)}{N^{(d-1)/d}} & = \beta(d) \int_{\omega_i} p(x)^{(d-1)/d}dx & \mbox{$p^{\otimes \mathbb{N} }$-a.s, $\forall i$}.
\end{align*}
and 
\begin{align*}
    \lim_{N\rightarrow +\infty} \frac{T(X_N,\setOmega)}{N^{(d-1)/d}} & = \beta(d) \int_{\setOmega} p(x)^{(d-1)/d}dx & \mbox{$p^{\otimes \mathbb{N} }$-a.s, $\forall i$}.
\end{align*}
Combining this result with Eqs.~\eqref{allequal} and \eqref{eq:defalternative} yields Proposition~\ref{cube_by_cube}.
\end{proof}

\begin{proof}[Proof of Theorem \ref{thm:convergence_proba}]
    Let $\varepsilon > 0$ and $h$ be an integer such that $\sqrt{d} h^{-d} < \varepsilon$. Then any two points in $\omega _i$ are at distance less than $\varepsilon $.
    
    Using Theorem \ref{cube_by_cube} and the fact that there is a finite number of $\omega _i$, almost surely, we get:

        $\lim_{N\rightarrow +\infty} \sum_{i=1}^{h^d} \left| \tilde{P}_N(\omega_i) - \tilde p(\omega _i) \right| = 0$. 
 Hence, for any $N$ large enough, there is a coupling $K$ of $\tilde{P}_N$ and $\tilde p$ such that both corresponding random variables are in the same $\omega _i$ with probability $1- \varepsilon $. 
Let $A\subseteq \setOmega$ be a Borelian. The coupling satisfies $\tilde{P}_N(A) = K(A \otimes \setOmega)$ and $\tilde p(A) = K(\setOmega \otimes A)$. Define the $\varepsilon$-neighborhood by $A^\varepsilon=\{X\in \setOmega \, | \, \exists Y \in A, \  \|X-Y\|<\varepsilon \}$. Then, we have:
$\tilde{P}_N(A) =K(A \otimes \setOmega)=K(\{A \otimes \setOmega\} \cap \{ |X - Y| < \varepsilon\})+ K(\{A \otimes \setOmega\} \cap \{ |X - Y| \geqslant \varepsilon \})$. It follows that:
\begin{align*}
       \tilde{P}_N(A)&\leqslant  K({A \otimes A^{\epsilon}}) + K(|X - Y| \geqslant \varepsilon) \\
       &\leqslant  K(\setOmega \otimes A^{\varepsilon})  + \varepsilon 
=    \tilde p(A^{\varepsilon}) + \varepsilon.
\end{align*}

This exactly matches the definition of convergence in the Prokhorov metric, which implies convergence in distribution.
\end{proof}

\section*{Acknowledgments}

The authors wish to thanks Yves Wiaux, Fabrice Gamboa, J\'er\'emie Bigot, Laurent Miclo, Alexandre Vignaud and Claire Boyer for fruitful discussions and feedback. 
This research was supported by the Labex CIMI through a 3 months invitation of Philippe Ciuciu.
This work was partially supported by ANR SPH-IM-3D (ANR-12-BSV5-0008), by the FMJH Program Gaspard Monge in optimization and operation research (MAORI project),
and by the support to this program from EDF.
 
\bibliographystyle{plain}

\begin{thebibliography}{}

\end{thebibliography}


\begin{thebibliography}{10}

\bibitem{Adcock13}
B.~Adcock, A.~Hansen, C.~Poon, and B.~Roman.
\newblock Breaking the coherence barrier: asymptotic incoherence and asymptotic
  sparsity in compressed sensing.
\newblock {\em arXiv preprint arXiv:1302.0561}, 2013.

\bibitem{Concorde06}
D.~Applegate, R.~Bixby, V.~Chvatal, and W.~Cook.
\newblock {Concorde TSP solver}.
\newblock {\em URL: http://www.tsp.gatech.edu/concorde}, 2006.

\bibitem{beardwood1959shortest}
J.~Beardwood, J.~H. Halton, and J.~M. Hammersley.
\newblock The shortest path through many points.
\newblock In {\em Mathematical Proceedings of the Cambridge Philosophical
  Society}, volume~55, pages 299--327. Cambridge Univ Press, 1959.

\bibitem{Bigot13}
J.~Bigot, C.~Boyer, and P.Weiss.
\newblock An analysis of block sampling strategies in compressed sensing.
\newblock {\em arXiv preprint arXiv:1305.4446}, 2013.

\bibitem{Billingsley09}
P.~Billingsley.
\newblock {\em Convergence of probability measures}, volume 493.
\newblock Wiley, 2009.

\bibitem{Boyer12}
C.~Boyer, P.~Ciuciu, P.~Weiss, and S.~M\'eriaux.
\newblock {HYR2PICS: Hybrid Regularized Reconstruction for Combined Parallel
  Imaging and Compressive Sensing in MRI}.
\newblock In {\em Proc. of 9th IEEE ISBI conference}, pages 66--69, Barcelona,
  Spain, May 2012.

\bibitem{boyer2013algorithm}
C.~Boyer, P.~Weiss, and J.~Bigot.
\newblock An algorithm for variable density sampling with block-constrained
  acquisition.
\newblock {\em SIAM Journal on Imaging Science}, (in press) 2014.

\bibitem{Bremaud99}
P.~Br\'emaud.
\newblock {\em Markov chains: Gibbs fields, Monte Carlo simulation, and
  queues}, volume~31.
\newblock springer, 1999.

\bibitem{Candes11}
E.~Cand{\`e}s and Y.~Plan.
\newblock A probabilistic and ripless theory of compressed sensing.
\newblock {\em {{IEEE} {T}rans. {I}nf. {T}heory}}, 57(11):7235--7254, 2011.

\bibitem{Candes06}
E.~Cand{\`e}s, J.~Romberg, and T.~Tao.
\newblock Robust uncertainty principles: exact signal reconstruction from
  highly incomplete frequency information.
\newblock {\em {{IEEE} {T}rans. {I}nf. {T}heory}}, 52(2):489--509, 2006.

\bibitem{Chauffert13b}
N.~Chauffert, P.~Ciuciu, J.~Kahn, and P.~Weiss.
\newblock Travelling salesman-based variable density sampling.
\newblock In {\em Proc. of 10th {SampTA} conference}, pages 509--512, Bremen,
  Germany, July 2013.

\bibitem{Chauffert13}
N.~Chauffert, P.~Ciuciu, and P.~Weiss.
\newblock Variable density compressed sensing in {MRI}. {T}heoretical {vs.}
  heuristic sampling strategies.
\newblock In {\em Proc. of 10th IEEE ISBI conference}, pages 298--301, San
  Francisco, USA, Apr. 2013.

\bibitem{Chauffert13c}
N.~Chauffert, P.~Ciuciu, P.~Weiss, and F.~Gamboa.
\newblock {From variable density sampling to continuous sampling using Markov
  chains}.
\newblock In {\em Proc. of 10th {SampTA} conference}, pages 200--203, Bremen,
  Germany, July 2013.

\bibitem{Combettes11b}
P.~L. Combettes and J.-C Pesquet.
\newblock {Proximal splitting methods in signal processing}.
\newblock In {\em {Fixed-Point Algorithms for Inverse Problems in Science and
  Engineering}}, pages 185--212. Springer, 2011.

\bibitem{diaconis1991geometric}
P.~Diaconis and D.~Stroock.
\newblock Geometric bounds for eigenvalues of markov chains.
\newblock {\em The Annals of Applied Probability}, 1(1):36--61, 1991.

\bibitem{Donoho06}
D.~L. Donoho.
\newblock {Compressed sensing}.
\newblock {\em {{IEEE} {T}rans. {I}nf. {T}heory}}, 52(4):1289--1306, Apr. 2006.

\bibitem{foucart2013mathematical}
S.~Foucart and H.~Rauhut.
\newblock A mathematical introduction to compressive sensing.
\newblock {\em Appl. Numer. Harmon. Anal. Birkh{\"a}user, Boston}, 2013.

\bibitem{yukich_gutin2002traveling}
A.~M. Frieze and J.~E. Yukich.
\newblock Probabilistic analysis of the {TSP}.
\newblock In G.~Gutin and A.~P. Punnen, editors, {\em The traveling salesman
  problem and its variations}, volume~12 of {\em Combinatorial optimization},
  pages 257--308. Springer, 2002.

\bibitem{Gross11}
D.~Gross.
\newblock Recovering low-rank matrices from few coefficients in any basis.
\newblock {\em {{IEEE} {T}rans. {I}nf. {T}heory}}, 57(3):1548--1566, 2011.

\bibitem{Haldar11}
J.~P Haldar, D.~Hernando, and Z.-P. Liang.
\newblock {Compressed-sensing MRI with random encoding}.
\newblock {\em {{IEEE} {T}rans. {M}ed. {I}mag.}}, 30(4):893--903, 2011.

\bibitem{hastings1970montecarlo}
W.~K. Hastings.
\newblock {Monte Carlo sampling methods using Markov chains and their
  applications}.
\newblock {\em Biometrika}, 57(1):97--109, Apr. 1970.

\bibitem{Horn91}
R.~Horn and C.~Johnson.
\newblock {\em {Topics in matrix analysis}}.
\newblock Cambridge University Press, Cambridge, 1991.

\bibitem{Jerrum89}
M.~Jerrum and A.~Sinclair.
\newblock Approximating the permanent.
\newblock {\em SIAM journal on computing}, 18(6):1149--1178, 1989.

\bibitem{Juditsky11b}
A.~Juditsky, F.K. Karzan, and A.~Nemirovski.
\newblock On low rank matrix approximations with applications to synthesis
  problem in compressed sensing.
\newblock {\em SIAM J. on Matrix Analysis and Applications}, 32(3):1019--1029,
  2011.

\bibitem{Juditsky11}
A.~Juditsky and A.~Nemirovski.
\newblock On verifiable sufficient conditions for sparse signal recovery via
  $\ell_1$ minimization.
\newblock {\em Mathematical Programming Ser. B}, 127:89--122, 2011.

\bibitem{Kargin07}
V.~Kargin.
\newblock {A large deviation inequality for vector functions on finite
  reversible Markov chains}.
\newblock {\em The Annals of Applied Probability}, 17(4):1202--1221, Aug. 2007.

\bibitem{kim2003simple}
D.~H. Kim, E.~Adalsteinsson, and D.~M. Spielman.
\newblock Simple analytic variable density spiral design.
\newblock {\em Magnetic resonance in medicine}, 50(1):214--219, 2003.

\bibitem{Knoll11}
F.~Knoll, C.~Clason, C.~Diwoky, and R.~Stollberger.
\newblock Adapted random sampling patterns for accelerated {MRI}.
\newblock {\em {{M}agma}}, 24(1):43--50, 2011.

\bibitem{Krahmer12}
F.~Krahmer and R.~Ward.
\newblock Beyond incoherence: stable and robust sampling strategies for
  compressive imaging.
\newblock preprint, 2012.

\bibitem{le2012dubins}
J.~Le~Ny, E.~Feron, and E.~Frazzoli.
\newblock On the dubins traveling salesman problem.
\newblock {\em Automatic Control, IEEE Transactions on}, 57(1):265--270, 2012.

\bibitem{Lezaud98}
P.~Lezaud.
\newblock Chernoff-type bound for finite {M}arkov chains.
\newblock {\em Annals of Applied Probability}, 8(3):849--867, 1998.

\bibitem{Lustig07}
M.~Lustig, D.~L. Donoho, and J.~M. Pauly.
\newblock Sparse {MRI}: The application of compressed sensing for rapid {MR}
  imaging.
\newblock {\em {{M}agn. {R}eson. {M}ed.}}, 58(6):1182--1195, Dec. 2007.

\bibitem{Lustig08}
M.~Lustig, S.~J. Kim, and J.~M. Pauly.
\newblock A fast method for designing time-optimal gradient waveforms for
  arbitrary k-space trajectories.
\newblock {\em {{IEEE} {T}rans. {M}ed. {I}mag.}}, 27(6):866--873, 2008.

\bibitem{marim2010compressed}
M.~M. Marim, M.~Atlan, E.~Angelini, and J-C. Olivo-Marin.
\newblock Compressed sensing with off-axis frequency-shifting holography.
\newblock {\em Optics letters}, 35(6):871--873, 2010.

\bibitem{park2005artifact}
J.~Park, Q.~Zhang, V.~Jellus, O.~Simonetti, and D.~Li.
\newblock {Artifact and noise suppression in GRAPPA imaging using improved
  k-space coil calibration and variable density sampling}.
\newblock {\em Magnetic resonance in medicine}, 53(1):186--193, 2005.

\bibitem{paulin2012concentration}
D.~Paulin.
\newblock Concentration inequalities for markov chains by marton couplings.
\newblock {\em arXiv preprint arXiv:1212.2015}, 2012.

\bibitem{polak2012grouped}
A.~C. Polak, M.~F. Duarte, and D.~L. Goeckel.
\newblock Grouped incoherent measurements for compressive sensing.
\newblock In {\em Statistical Signal Processing Workshop (SSP), 2012 IEEE},
  pages 732--735. IEEE, 2012.

\bibitem{polak2012performance}
A.~C. Polak, M.~F. Duarte, and D.~L. Goeckel.
\newblock Performance bounds for grouped incoherent measurements in compressive
  sensing.
\newblock {\em arXiv preprint arXiv:1205.2118}, 2012.

\bibitem{Pruessmann99}
K.~P. Pruessmann, M.~Weiger, M.~B. Scheidegger, and P.~Boesiger.
\newblock {SENSE}: sensitivity encoding for fast {MRI}.
\newblock {\em Magnetic Resonance in Medicine}, 42(5):952--962, Jul. 1999.

\bibitem{Puy12}
G.~Puy, J.~P Marques, R.~Gruetter, J.~Thiran, D.~Van De~Ville,
  P.~Vandergheynst, and Y.~Wiaux.
\newblock Spread spectrum magnetic resonance imaging.
\newblock {\em {{IEEE} {T}rans. {M}ed. {I}mag.}}, 31(3):586--598, 2012.

\bibitem{Puy11}
G.~Puy, P.~Vandergheynst, and Y.~Wiaux.
\newblock On variable density compressive sampling.
\newblock {\em {{IEEE} {S}ignal {P}rocessing {L}etters}}, 18(10):595--598,
  2011.

\bibitem{Rauhut10}
H.~Rauhut.
\newblock {C}ompressive {S}ensing and structured random matrices.
\newblock In {M}. {F}ornasier, editor, {\em {T}heoretical {F}oundations and
  {N}umerical {M}ethods for {S}parse {R}ecovery}, volume~9 of {\em {R}adon
  {S}eries {C}omp. {A}ppl. {M}ath.}, pages 1--92. de{G}ruyter, 2010.

\bibitem{rivenson2010compressive}
Y.~Rivenson, A.~Stern, and B.~Javidi.
\newblock Compressive fresnel holography.
\newblock {\em Journal of Display Technology}, 6(10):506--509, 2010.

\bibitem{Sidky06}
E.~Y Sidky, C.-M. Kao, and X.~Pan.
\newblock Accurate image reconstruction from few-views and limited-angle data
  in divergent-beam ct.
\newblock {\em Journal of X-ray Science and Technology}, 14(2):119--139, 2006.

\bibitem{spielman1995magnetic}
D.~M. Spielman, J.~M. Pauly, and C.~H. Meyer.
\newblock Magnetic resonance fluoroscopy using spirals with variable sampling
  densities.
\newblock {\em Magnetic resonance in medicine}, 34(3):388--394, 1995.

\bibitem{steele1981subadditive}
J.~M. Steele.
\newblock Subadditive euclidean functionals and nonlinear growth in geometric
  probability.
\newblock {\em The Annals of Probability}, 9(3):365--376, 1981.

\bibitem{Teuber11}
T.~Teuber, G.~Steidl, P.~Gwosdek, C.~Schmaltz, and J.~Weickert.
\newblock Dithering by differences of convex functions.
\newblock {\em SIAM Journal on Imaging Science}, 4(1):79--108, 2011.

\bibitem{Tropp12}
J.~A. Tropp.
\newblock {User-friendly tail bounds for sums of random matrices}.
\newblock {\em Foundations of Computational Mathematics}, pages 1--32, Dec.
  2012.

\bibitem{tsai2000reduced}
C.~M. Tsai and D.~G. Nishimura.
\newblock Reduced aliasing artifacts using variable-density k-space sampling
  trajectories.
\newblock {\em Magnetic resonance in medicine}, 43(3):452--458, 2000.

\bibitem{Wang12}
H.~Wang, X.~Wang, Y.~Zhou, Y.~Chang, and Y.~Wang.
\newblock Smoothed random-like trajectory for compressed sensing {MRI}.
\newblock In {\em Proc. of the 34th annual IEEE EMBC}, pages 404--407, 2012.

\bibitem{Wiaux09}
Y.~Wiaux, G.~Puy, Y.~Boursier, and P.~Vandergheynst.
\newblock Spread spectrum for imaging techniques in radio interferometry.
\newblock {\em Monthly Notices of the Royal Astronomical Society},
  400(2):1029--1038, 2009.

\bibitem{Willett}
R.~M. Willett.
\newblock Errata: Sampling trajectories for sparse image recovery.
\newblock Note, Duke University, 2011.

\bibitem{yukich98}
J.~E. Yukich.
\newblock {\em Probability theory of classical Euclidean optimization
  problems}.
\newblock Springer, 1998.

\end{thebibliography}

\end{document}